\documentclass[aos]{imsart}

\RequirePackage{amsthm,amsmath,amsfonts,amssymb}
\RequirePackage[numbers,sort&compress]{natbib}
\RequirePackage[colorlinks,citecolor=blue,urlcolor=blue]{hyperref}
\RequirePackage{graphicx}
\RequirePackage{color}
\RequirePackage{xcolor} % expands list of colours available: black, blue, brown, cyan, darkgrey, grey, green, lime, magenta, olive, orange, pink, purple, red, teal, violet, white, yellow
\usepackage{float}
\usepackage{comment}
\usepackage[textsize=tiny]{todonotes}
\usepackage{enumitem}

\startlocaldefs
\numberwithin{equation}{section}
\theoremstyle{plain}

\newtheorem{theorem}{Theorem}[section]
\newtheorem{prop}[theorem]{Proposition}
\newtheorem{lemma}[theorem]{Lemma}
\newtheorem{cor}[theorem]{Corollary}
\theoremstyle{remark}
\newtheorem{remark}[theorem]{Remark}

\def\N{\mathbb{N}}
\def\R{\mathbb{R}}
\def\Fcal{\mathcal{F}}
\def\Lcal{\mathcal{L}}
\def\Ocal{\mathcal{O}}
\def\Enorm{\textnormal{E}}
\def\vol{\textnormal{vol}}
\def\Tr{\textnormal{Tr}}
\def\eps{\varepsilon}
\def\fmin{f_{\textnormal{min}}}
\def\iid{\overset{\textnormal{iid}}{\sim}}
\endlocaldefs

\begin{document}

%%%%%%%%%%%%%%%%%%%%%%%%%%%%%%%%%%%%%%%%%%%%%%
\begin{frontmatter}
%\title{Algorithms for statistical inference with low-frequency diffusion data: a PDE approach}
\title{Statistical algorithms for low-frequency diffusion data: a PDE approach}

\runtitle{Inference with low-frequency diffusion data}

\begin{aug}
%%%%%%%%%%%%%%%%%%%%%%%%%%%%%%%%%%%%%%%%%%%%%%%
%% Only one address is permitted per author. %%
%% Only division, organization and e-mail is %%
%% included in the address.                  %%
%% Additional information can be included in %%
\author[A]{\fnms{Matteo}~\snm{Giordano\ead[label=e1]{matteo.giordano@unito.it}}}
\and
\author[B]{\fnms{Sven}~\snm{Wang}\ead[label=e2]{sven.wang@hu-berlin.de}}
%%%%%%%%%%%%%%%%%%%%%%%%%%%%%%%%%%%%%%%%%%%%%%
%% Addresses                                %%
%%%%%%%%%%%%%%%%%%%%%%%%%%%%%%%%%%%%%%%%%%%%%%
\address[A]{ESOMAS Department,
University of Turin\printead[presep={,\ }]{e1}}

\address[B]{Department of Mathematics,
Humboldt University Berlin\printead[presep={,\ }]{e2}}
\end{aug}

%%%%%%%%%%%%%%%%%%%%%%%%%%%%%%%%%%%%%%%%%%%%%%
\begin{abstract}
We consider the problem of making nonparametric inference in a class of
multi-dimensional diffusions in divergence form, from low-frequency data. Statistical analysis in this setting is notoriously challenging due to the intractability of the likelihood and its gradient, and computational methods have thus far largely resorted to expensive simulation-based techniques. In this article, we propose a new computational approach which is motivated by PDE theory and is built around the characterisation of the transition densities as solutions of the associated heat (Fokker-Planck) equation. Employing optimal regularity results from the theory of parabolic PDEs, we prove a novel characterisation for the gradient of the likelihood. Using these developments, for the nonlinear inverse problem of recovering the diffusivity, we then show that the numerical evaluation of the likelihood and its gradient can be reduced to standard elliptic eigenvalue problems, solvable by powerful finite element methods. This enables the efficient implementation of a large class of popular statistical algorithms, including (i) preconditioned Crank-Nicolson and Langevin-type methods for posterior sampling, and (ii) gradient-based descent optimisation schemes to compute maximum likelihood and maximum-a-posteriori estimates. We showcase the effectiveness of these methods via extensive simulation studies in a nonparametric Bayesian model with Gaussian process priors, in which both the proposed optimisation and sampling schemes provide good numerical recovery. %\matteo{Interestingly, the optimisation schemes provided satisfactory numerical recovery while requiring a significantly smaller computational time compared to the employed sampling methods}.
The reproducible code is available at \url{https://github.com/MattGiord/LF-Diffusion}.
\end{abstract}
%%%%%%%%%%%%%%%%%%%%%%%%%%%%%%%%%%%%%%%%%%%%%%

%%%%%%%%%%%%%%%%%%%%%%%%%%%%%%%%%%%%%%%%%%%%%%
%%%%%%%%%%%%%%%%%%%%%%%%%%%%%%%%%%%%%%%%%%%%%%

%%%%%%%%%%%%%%%%%%%%%%%%%%%%%%%%%%%%%%%%%%%%%%
%%%%%%%%%%%%%%%%%%%%%%%%%%%%%%%%%%%%%%%%%%%%%%

\end{frontmatter}
%%%%%%%%%%%%%%%%%%%%%%%%%%%%%%%%%%%%%%%%%%%%%%
%% Please use \tableofcontents for articles %%
%% with 50 pages and more                   %%
%%%%%%%%%%%%%%%%%%%%%%%%%%%%%%%%%%%%%%%%%%%%%%
\setcounter{tocdepth}{1}

%\tableofcontents
\graphicspath{{Figures/}}

%
%
%
%
%

%%%%%%%%%%%%%%%%%%%%%%%%%%%%%%%%%%%%%%%%%%%%%%
\section{Introduction}

Diffusions are mathematical models used ubiquitously across the sciences and in applications. They describe the stochastic time-evolution of a large variety of phenomena, including heat conduction \cite{Baher98}, chemical reactions \cite{Peletier12}, cellular dynamics \cite{Briane20} and financial markets \cite{Shreve04}. See the monograph \cite{Allen07} for further examples and references. In many situations, the `drift' and `diffusivity' parameters of a stochastic process $(X_t, \ t\ge 0)$ are not precisely known, and have to be estimated from discrete-time observations of a particle trajectory
\begin{equation}\label{Eq:LFData}
    X^{(n)}:= (X_0,X_D,X_{2D},...,X_{nD}),
\end{equation}
for some `observation distance' $D>0$. This is the central inferential problem considered in the present article. Due to their unorthodox likelihood structure, which is implicitly determined by the transition probabilities of $(X_t, \ t\ge 0)$, discrete diffusion data have posed formidable difficulties for statistical analysis. While remarkable progresses have recently been made in deriving theoretical recovery guarantees, devising efficient computational algorithms remains a significant challenge -- see below for more discussion.

Here, we shall study these issues in 
%\matteo{the ?important?}
a nonparametric model for diffusion inside a bounded
%insulated, and inhomogeneous medium
region; possible extensions will be discussed below. Taking, throughout, the diffusion domain to be a subset $\Ocal\subset\R^d$, $d\in\N$, such a system is \textit{macroscopically} described by the (Fokker-Planck) parabolic partial differential equation (PDE), which we shall refer to as the `heat equation',
\begin{equation}
\label{Eq:HeatEq}
\begin{cases}
\partial_t u - \nabla\cdot( f\nabla u) = 0, 
& \textnormal{on}\  (0,\infty)\times \Ocal,\\
\partial_\nu u = 0, 
& \textnormal{on}\  (0,\infty)\times \partial\Ocal,
\end{cases}
\end{equation}
encapsulating the changes over time in the substance density $u(t,x)$ at each location $x\in\Ocal$ \cite[Chapter 11]{E10}. Above, $\nabla\cdot$ and $\nabla$ denote, respectively, the divergence and gradient operators, $\nu$ is the inward-pointing unit normal vector with associated normal derivative $\partial_\nu$. The zero-Neumann boundary condition $\partial_\nu u= 0$
corresponds to the boundary being `insulated', and $f:\Ocal\to(0,\infty)$ is a `conductivity' function modelling the spatially-varying intensity at which diffusion occurs throughout the inhomogeneous medium. At the \textit{microscopic} level, the trajectory $(X_t, \ t\ge0)$ of a diffusing particle, started inside $\Ocal$, evolves according to the associated stochastic differential equation (SDE),
\begin{equation}
\label{Eq:SDE}
    dX_t = \nabla f(X_t)dt + \sqrt{2f(X_t)}dW_t + \nu(X_t)dL_t, \qquad t\ge 0,
\end{equation}
where $(W_t, \ t\ge0)$ is a standard $d$-dimensional Brownian motion and the term $\nu(X_t)dL_t$ models reflection of the particle at the insulated boundary $\partial\Ocal$ of the medium via the local time process $(L_t, \ t\ge0)$; see \cite{Tanaka79} for details. The connection between the SDE (\ref{Eq:SDE}) and the PDE (\ref{Eq:HeatEq}) will play a key role throughout this paper.

%

%\begin{figure}
%\includegraphics[width=4.7cm]{F0}
%\includegraphics[width=4.7cm]{ContPath}
%\includegraphics[width=4.7cm]{DiscrObs}
%\caption{Left: a (reparametrised) conductivity function, modelling two areas of more intense diffusion around the points $(.2,.2)$ and $(-.2,-.2)$. Centre: the (unobserved) continuous trajectory $(X_t, \ 0\le t\le T)$, initialised at the origin and run until time $T=5$. Right: low-frequency observations $X^{(100)}=(X_0,X_D,\dots,X_{100D})$ of the continuous trajectory, sampled with time lag $D=.05$.}
%\label{Fig:StatProbl}
%\end{figure}

The statistical reconstruction task consists of determining nonparametrically (i.e.~within some infinite-dimensional function class) $f$ from the discrete measurements (\ref{Eq:LFData})
separated by a (fixed) time lag $D>0$. Often times, because of the characteristics of the data collection process, $D$ cannot be reduced under a certain non-zero threshold -- in the statistical literature, this is referred to as the `low-frequency' regime, which is the main setting our PDE techniques will be targeted at. 
%An illustration of the problem with synthetic data is provided in Figure \ref{Fig:StatProbl}. 
For example, \cite[Chapter 1]{Majda2012} describes a filtering problem in weather forecasting where measurements are inputted in a large scale dynamical system every few hours; see also \cite{Heckert2022} for a similar situation  in systems biology.
Among the concrete applications of divergence form models, we mention the 3D single particle tracking experiments considered e.g.~in \cite{H21}, wherein the diffusion process \eqref{Eq:SDE} arises as an instance with constant potential; see Section \ref{Eq:Potential} for possible extensions to settings with spatially-varying potential energy. We also refer to \cite[p.~200f]{L13} for applications in the context of spatial ecology. %\matteo{\st{Other possible approaches based on stochastic analysis which are typically only available with `high-frequency' or continuous-time measurements of $(X_t, \ t\ge0)$ will not be pursued here; more discussion can be found below.}}

%\todo{Should we even discuss the different measurement schemes here? / to what extent?}
%For example, \cite[Chapter 1]{Majda2012} describes a filtering problem in weather forecasting where measurements are inputted in a large scale dynamical system every few hours, and where high-frequency or continuous-time observations are unfeasible to implement. We also refer to \cite{Heckert2022} for a similar situation in system biology. Motivated by these examples, we shall here consider the problem of making statistical inference on the conductivity function $f$ in the above prototypical diffusion model, based on low-frequency observations $X^{(n)}$ of a particle trajectory $(X_t,\ t\ge0)$ solving the SDE \eqref{Eq:SDE}. An illustration of the problem with synthetic data is provided in Figure \ref{Fig:StatProbl}. 

%

Related `parameter identification' problems for the conductivity in diffusion equations have also been widely studied in the inverse problem literature, largely in applications where observations of a steady-state system are available in the form of (possibly noisy) point evaluations of the solution of a time-independent elliptic PDE. Among the many contributions, we refer to \cite{C80,kohn1984determining,nachman1988reconstructions,uhlmann2009electrical,AN19} for models with boundary measurements in the context of the famous `Caldéron problem', and to \cite{R81,EHN96,S10,nickl2020convergence,GN20} for interior measurements schemes connected to the `Darcy's flow' model. Finally, there is a wide literature on nonparametric coefficient estimation problems in SDEs with high-frequency and continuous-time data; we refer to \cite{%hoffmann1999adaptive,
giordano2022nonparametric,strauch2018adaptive,NR20,hoffmann2022bayesian,comte2007penalized} for a comprehensive overview and further references.

%

%%%%%%%%%%%%%%%%%%%%%%%%%%%%%%%%%%%%%%%%%%%%%%
\subsection{Challenges}

In the present setting, the invariant distribution of the diffusion process $(X_t,\ t \ge 0)$ in \eqref{Eq:SDE} can be shown to coincide with the uniform distribution $\vol(\Ocal)^{-1}dx$ on $\Ocal$ \cite[Chapter 1.11.3]{BGL14}, and therefore is non-informative about the conductivity $f$. Moreover, in the low-frequency regime, common stochastic analysis-based approaches %sample statistics (such as the `sample quadratic variation') or other
%approaches based on stochastic analysis
which underlie high-frequency and continuous-time methods (e.g.~in \cite{%hoffmann1999adaptive,
giordano2022nonparametric,strauch2018adaptive,NR20,hoffmann2022bayesian,comte2007penalized}) cannot be employed to validly estimate $f$. As laid out in \cite[Section 1.2.3]{GHR04}, this is because low-frequency data do not allow to recover `full trajectory properties'.

Instead, the problem must rather be approached using the information contained in the transition densities $p_{D,f}(x,y), \ x,y\in\Ocal$, namely the probability density functions of the conditional laws $\Pr(X_D\in dy|X_0=x)$, and the resulting likelihood function
\begin{equation}
\label{Eq:Likelihood}
    L_n(f)=\prod_{i=1}^np_{D,f}(X_{(i-1)D},X_{iD}). 
\end{equation}
However, apart from certain special cases, the transition densities of diffusion processes, including those of model \eqref{Eq:SDE}, are generally not available in closed form, making the likelihood for low-frequency observations analytically intractable. This is the central issue posing a huge challenge to the design, implementation and theoretical analysis of statistical algorithms.

In such contexts, 
%a large part of the 
many existing parametric and nonparametric methods rely on computational strategies that involve sophisticated (and often computationally onerous) missing data techniques, whereby the unobserved continuous trajectory between the data points is treated as a latent variable and inputted via simulation schemes for diffusion bridges, enabling the approximation of the likelihood for low-frequency observations
by the more tractable one for continuous-time data. This approach was first pioneered, in general diffusion models, by Pedersen \cite{pedersen1995consistency} to construct simulated maximum likelihood estimators, and by Roberts and Stramer \cite{roberts2001inference}, Elerian et al.~\cite{elerian2001likelihood} and Eraker \cite{eraker2001mcmc} to implement Bayesian inference with data-augmentation. See \cite{durham2002numerical,
delyon2006simulation,
beskos2006exact,golightly2008bayesian,
beskos2009monte,lin2010generating,papaspiliopoulos2013data,bladt2016simulation,van2017bayesian,schauer2017guided} and the many references therein. 

%
%
%

%%%%%%%%%%%%%%%%%%%%%%%%%%%%%%%%%%%%%%%%%%%%%%
\subsection{Main contributions}
In this paper, we adopt a novel PDE perspective to address the computational challenges arising in the nonparametric diffusion statistical model \eqref{Eq:SDE} with low-frequency observations. Our main contributions are as follows.
\begin{itemize}
\item We derive theoretical PDE formulae for likelihoods and gradients which are concretely computable via standard finite element methods for elliptic PDEs.
\item We formulate several novel algorithms for posterior sampling and optimisation.
\item We implement and numerically demonstrate the efficacy of the algorithms for computing maximum-a-posteriori, maximum likelihood and posterior mean estimates, as well as for posterior sampling.
\end{itemize}

Let us briefly expand on these points. In Section \ref{Sec:PDEAppraoch}, building on the characterisation of the transition densities $p_{D,f}$ appearing in \eqref{Eq:Likelihood} as the fundamental solutions of the heat equation \eqref{Eq:HeatEq}, we first show how the computation of $p_{D,f}$ can be reduced to a corresponding time-independent eigenvalue problem for the  elliptic (self-adjoint) infinitesimal generator. This will later lead to a simple likelihood evaluation routine, cf.~\eqref{Eq:SpectrLikelihood}, that does not require any data-augmentation step.
%and only incurs in a negligible discretisation error due to series truncation and the tightly controlled numerical approximation associated to the employed elliptic solver. 
Building on such PDE characterisation, the main theoretical result of this article, Theorem \ref{first-der}, is then derived. In particular, we prove a `closed-form' expression for the gradient of the likelihood, by characterising the Frechét derivatives of the maps $f\mapsto p_{D,f}(x,y)$, for $x,y\in\Ocal$ fixed. These are obtained using perturbation arguments for parabolic PDEs and the so-called `variation-of-constants' principle \cite[Chapter 4]{lunardi}, building on the regularisation step developed by Wang \cite{W19} to deal with the singular behavior of the transition densities relative to vanishing time instants. The full argument is fairly technical; it is presented in Section \ref{sec-gradient-pf} in the Supplement \cite{giordanowangsupplement}. Again using the self-adjointness of the generator and numerical methods for elliptic PDEs, we then propose an efficient strategy to numerically evaluate the likelihood gradient, which can serve as the building block for the implementation of gradient-based statistical algorithms. Thus, while our main results are of independent interest for the literature on parabolic PDEs, they are strongly motivated by the primary goal of gaining access to (previously unavailable) core likelihood-based methodologies in the problem of low-frequency diffusion data.%, at least in the context of the divergence form model \eqref{Eq:SDE}.

Section \ref{Sec:Algo} details the statistical procedures which we derive from the above theoretical developments. Our results enable a large algorithmic toolbox of common likelihood-based computational techniques -- in particular, we pursue both gradient-free (preconditioned Crank-Nicolson, pCN) and gradient-based (unadjusted Langevin, ULA) algorithms for posterior sampling, as well as gradient descent methods. These schemes allow to obtain numerical approximations for posterior mean and maximum a posteriori (MAP) estimates, posterior quantiles for uncertainty quantification, as well as (penalised) maximum likelihood type estimators. The detailed description of the algorithms can be found in Sections \ref{Sec:pCN}-\ref{Sec:GradDesc}.

In several simulation studies, presented in Section \ref{sec:Num}, we apply the above methods to a nonparametric Bayesian model with truncated Gaussian series priors. In the large data limit $n\to \infty$, these priors have recently been shown by Nickl \cite{nickl2024inference} (see also \cite{alberti2024low}) to lead to consistent inference of the data-generating `ground truth' conductivity $f$ -- but in principle, our numerical methods are also applicable to other priors. Interestingly, \cite{nickl2024inference} undertook a similar PDE-based point of view to prove the injectivity of the nonlinear map $f\mapsto p_{D,f}$ from the conductivity to the transition densities, providing the first statistical guarantees for nonparametric Bayesian procedures with multi-dimensional low-frequency diffusion data.

%It was shown that certain truncated Gaussian series priors lead to posterior distributions that asymptotically concentrate around the ground truth as the sample size increases. Our work includes an implementation of such consistent statistical methods, showcasing their feasibility and effectiveness in practice. 

Our work also opens the door for the implementation of further Markov Chain Monte Carlo (MCMC) \cite{CRSW13, %graham2015quasi,
CLM16,BGLFS17%sprungk2023metropolis
} and gradient-based optimisation methods, an important direction of future research.

\subsection{Related literature and discussion}

In the seminal paper \cite{GHR04} by Gobet et al.,~spectral methods (related to the ones pursued here) were used to obtain minimax-optimal nonparametric estimators in one-dimensional diffusion models. However, it seems challenging to apply their approach to the present multi-dimensional setting, where the elliptic generator defines a genuine PDE.
%instead of an ODE.
Analogous ideas also underpin the parametric estimators built by Kessler and Sørensen \cite{kessler1999estimating} using certain spectral martingale estimating functions.
%Further of note are the 
We also mention the works by Aït-Sahalia \cite{ait2002maximum,ait2008closed} which (in a parametric setting) derive closed-form likelihood approximations via Hermite polynomials along with resulting approximate maximum likelihood estimators. In contrast to our work, calculations of gradients are not considered there; moreover, the expansions on the eigenbasis of the self-adjoint generator of \eqref{Eq:SDE} considered here lead to rapidly (i.e.~exponentially) decaying remainders terms.  %, while also enabling the use of efficient numerical PDE solvers.

Let us briefly discuss future directions of research which may build on the present work. A first important avenue would be the extension of the developed methodology beyond the divergence form diffusion model \eqref{Eq:SDE}. Natural generalisations encompass anisotropic diffusions with matrix-valued conductivities, models for diffusion in a `potential energy field' (see Section \ref{Sec:Discuss}) and diffusions on $\R^d$, which would necessitate extending our spectral and PDE arguments to unbounded domains. %, e.g.~building on \cite{PV01}.
%for further discussion. 
Secondly, in practical applications, it is likely that measurements may only be available under (e.g.~Gaussian) observational noise. % (e.g.~due to Gaussian measurement errors)
This gives rise to a hidden Markov model (HMM) as described for example in \cite[Chapter 11.5]{Sarkka19}, where the likelihood structure is characterised by convolutions of $p_{D,f}(\cdot,\cdot)$ with the noise density. In this scenario,
 our methods for evaluating the transition densities and their gradients, when combined with smoothing and filtering techniques from the HMM literature, may still be used %as the cornerstone
 to implement likelihood-based inference for joint state and parameter estimation%, combined with smoothing and filtering techniques from the HMM literature
 %with our efficient routines for evaluating the transition densities and their gradients
 ; see e.g.~\cite{S02,%Cappe05,
 mbalawata2013parameter,Foulon24}. Relatedly, we mention recent methodological work which uses tools from computational graph completion in the context of SDEs \cite{DARCYetal}.

Our calculations, combined with the results in \cite{nickl2024inference}, may also pave the way to proving `gradient stability' properties in the sense of \cite{nickl2022polynomial}; further see \cite{bohr2021log,altmeyer2022polynomial, bandeira2023free}, and \cite[Chapter 3]{nickl2023bayesian}. Using the program put forth in \cite{nickl2022polynomial}, a rigorous investigation of the complexity of the employed sampling and optimisation algorithms can then be carried out, with the goal of deriving bounds for the computational cost that scale polynomially with respect to the discretisation dimension and the sample size. Further discussion can be found in Section \ref{Sec:GradStab}. 

Another interesting question concerns the relationship between the observation time lag $D>0$, the `numerical stability' of the proposed methodology, and the `statistical information' contained in the sample. While, algorithmically, low-frequency samples imply better spectral approximations of the likelihood and its gradient, higher sampling frequencies allow to capture finer characteristics which may facilitate statistical convergence. In particular, recent work has shown that nonparametric Bayesian procedures based on Gaussian priors can achieve optimal statistical convergence rates in the model \eqref{Eq:SDE} with `high-frequency' observations (where $D= D_n\approx n^{-\gamma}\to0$ for suitable $\gamma>0$) \cite{hoffmann2022bayesian} and with continuous-time observations \cite{van2006convergence,ns21,van2016gaussian,NR20,giordano2022nonparametric}. In the high-frequency regime, small-time Gaussian heat kernel asymptotics \cite{SC10} may provide good numerical likelihood approximations, related to the (tractable) continuous-time likelihood provided by Girsanov's theorem. Understanding the more detailed `phase transitions' between the different sampling frequencies will inform which methods to employ in practice, see Section \ref{sec:gauss-appr} for a detailed discussion.

%\red{For instance, with} continuous data, likelihood evaluations become tractable via Girsanov's theorem. \todo{This paragraph is now a little bit redundant perhaps?}
%Understanding the more detailed `phase transitions' between the different regimes will inform which methods to employ in practice. Recent work has shown that nonparametric Bayesian procedures based on Gaussian priors can achieve optimal statistical convergence rates in the model \eqref{Eq:SDE} with `high-frequency' observations (where $D= D_n\approx n^{-\gamma}\to0$ for some $\gamma>0$) \cite{hoffmann2022bayesian} and with continuous-time observations \cite{van2006convergence,ns21,van2016gaussian,NR20,giordano2022nonparametric}. 

We also mention the highly successful diffusion generative models \cite{song21} which are based on estimating the drift of a time-reversed SDE `forward noising process' (the \textit{score function}); see Section \ref{sec:generativemodels} for more detailed discussion of the connections to the present setting.

\section{Likelihood and gradient computation via PDEs}
\label{Sec:PDEAppraoch}

Throughout, let $\Ocal \subset \R^d$, $d\in\N$, be a non-empty, open, bounded and convex set with smooth boundary $\partial\Ocal$. It is well-known that for any twice continuously differentiable and strictly positive $f\in C^2(\bar\Ocal)$, $\inf_{x\in\Ocal}f(x) >0$, and any given starting point $X_0=x_0\in\Ocal$ the SDE \eqref{Eq:SDE} has a unique path-wise solution $(X_t, \ t\ge0)$, constituting a continuous-time Markov diffusion process reflected at the boundary (since $f$ and $\nabla f$ are Lipschitz); see \cite{Tanaka79}. In view of these regularity assumptions, for some $\fmin>0$ we maintain
\begin{equation}
\label{Eq:ParamSpace}
    \Fcal := \Big\{ f \in C^2(\bar \Ocal): \inf_{x\in\mathcal O} f(x) \ge \fmin \Big\}
\end{equation}
as the parameter space. Recall the low-frequency observations $X^{(n)}$ from \eqref{Eq:LFData} with measurement distance $D>0$, which we shall keep fixed throughout.

%
%
%

%%%%%%%%%%%%%%%%%%%%%%%%%%%%%%%%%%%%%%%%%%%%%%
\subsection{Parabolic PDE characterisations}

The Markov property of $(X_t,~t\ge 0)$ implies that the likelihood $L_n(f)$ of any $f\in\Fcal$ 
%given by the joint probability density function of $X^{(n)}$,
factorises as a product of the (symmetric) transition densities $p_{t,f}(x,y)=p_{t,f}(y,x), \ t\ge0, \ x,y\in\Ocal$; cf.~\eqref{Eq:Likelihood}. These characterise the conditional laws
\begin{equation}
\label{Eq:TRDFs}
    \Pr(X_{s+t}\in A|X_s=x) = \int_A p_{t,f}(x,y)dy,
    \qquad A\subseteq \Ocal\ \textnormal{measurable}, 
    \qquad s\ge 0,
\end{equation}
and more generally the transition operator
\begin{equation}
\label{Eq:TrOp}
    P_{t,f}[u](x):=\Enorm[u(X_{s+t})|X_s=x]
    =\int_\Ocal p_{t,f}(x,y)u(y)dy, \qquad s\ge0,
\end{equation}
acting on square-integrable test functions $u\in L^2(\Ocal)$. The semigroup $(P_{t,f},~t\ge 0)$ is known to play the role of the `solution operator' for the heat equation \eqref{Eq:HeatEq}; thus the transition densities $p_{t,f}$ also constitute the fundamental solution to \eqref{Eq:HeatEq}. Informally, this means that for $y\in\mathcal O$ fixed, the map $(t,x)\mapsto p_{t,f}(y,x)$ solves \eqref{Eq:HeatEq} with Dirac initial condition,
\begin{equation}
\label{Eq:FundamentalPDE}
\begin{cases}
(\partial_t  - \mathcal L_f)u = 0, & \textnormal{on}\  (0,\infty)\times \Ocal,\\
\partial_\nu u = 0, 
& \textnormal{on}\  (0,\infty)\times \partial\Ocal,\\
u(0,\cdot)= \delta_y(\cdot), & \textnormal{on}\ \Ocal.
\end{cases}
\end{equation}
Here, we denoted by $\mathcal L_f$ the elliptic divergence form operator
\[ \mathcal L_f u = \nabla 
\cdot(f\nabla u)=\sum_{l=1}^d \partial_{x_l}(f\partial_{x_l} v),
\]
a notation that we will use throughout. The operator $\mathcal L_f$, with domain given by the set of functions in the Sobolev space  $H^2(\mathcal O)$ with zero Neumann boundary conditions,
\[ 
    H^2_N(\mathcal O):=\big\{ u\in H^2(\mathcal O):~\partial_\nu u =0 ~\text{on}~\partial\mathcal O \big\}, 
\]
constitutes the infinitesimal generator of the process \eqref{Eq:SDE}. While the transition density functions are not available in closed form, their characterisation through \eqref{Eq:FundamentalPDE} implies a convenient spectral expansion in terms of the eigenpairs of the generator $\Lcal_f$, cf.~\eqref{Eq:FundSol}, which we will use below for evaluating the likelihood $L_n(f)$.

A more intricate question is whether the gradient of the likelihood function $L_n$ also satisfies a PDE characterisation which can be exploited for computational purposes. The key challenge is thus to understand the perturbations of the nonlinear map $f\mapsto p_{t,f}$, which turns out to provide insight into the preceding question -- this is the content of Theorem \ref{first-der}. To understand the intuition behind the theorem, let us fix some perturbation $h\in C^2(\bar{\mathcal O})$ such that $f+h\in \mathcal F$. Then, subtracting the PDE (\ref{Eq:FundamentalPDE}) for $f+h$ and for $f$ yields immediately that the difference $w(t,x):=p_{t,f+h}(y,x)-p_{t,f}(y,x)$ solves (again, informally)
\begin{equation*}
%\label{Eq:FundamentalSol}
\begin{cases}
    (\partial_t - \mathcal L_f)w (t,x)  = \nabla \cdot (h\nabla p_{t,f+h}(y,\cdot))(x), 
    & \textnormal{for}\ (t,x)\in (0,\infty)\times \Ocal,\\
    \partial_\nu w (t,x)= 0, 
    & \textnormal{for}\  (t,x)\in (0,\infty)\times \partial\Ocal,\\
    w(0,x)= 0,& \textnormal{for}\  x\in \Ocal,
\end{cases}
\end{equation*}
which is another instance of the heat equation, now with an inhomogeneity and with zero initial conditions. A natural candidate for the linearisation (in $h$) of the right hand side is $\nabla \cdot (h\nabla p_{t,f+h}(y, \cdot)) \approx \nabla \cdot (h\nabla p_{t,f}(y, \cdot))$, and thus $(\partial_t-\mathcal L_f)^{-1} [\nabla \cdot (h\nabla p_{t,f}(y, \cdot))]$ in turn provides a natural candidate for the linearisation of the transition densities. Here, we have written $(\partial_t-\mathcal L_f)^{-1}$ to informally denote the linear `solution operator' to an inhomogeneous heat equation with zero initial condition, which under suitable regularity conditions is given by the variation-of-constants formula $(\partial_t-\mathcal L_f)^{-1} [\nabla \cdot (h\nabla p_{t,f}(y, \cdot))]= \int_0^t P_{t-s,f}[\nabla \cdot (h\nabla p_{s,f}(y, \cdot))]ds$ -- see e.g.~Chapter 4 of \cite{lunardi}.

Making the above argument rigorous is technically delicate due to the singularity of $\delta_y(\cdot)$ and of the source term $\nabla \cdot (h\nabla p_{t,f})$ for $t\to 0$, which makes the standard parabolic regularity theory (e.g.~from \cite{lunardi}) not directly applicable. Thus, one needs to clarify in which sense the above PDEs hold, and whether existence and uniqueness can be guaranteed suitably for $(\partial_t-\mathcal L_f)^{-1}$ to be well-specified. Generalising a regularisation technique developed in \cite{W19} (in a related one-dimensional model), we accomplish this in the ensuing theorem for dimensions $d\le 3$, proving a variation-of-constants representation for the linearisation of $f\mapsto p_{t,f}$. For $x,y\in\Ocal$ and $D>0$ fixed, define the operator
\[ 
    \Phi(f) \equiv \Phi_{D,x,y}(f) := p_{D,f}(x,y), 
    \qquad f\in\Fcal,
\]
where $\Fcal$ is given by \eqref{Eq:ParamSpace}.
Note that $\Phi$ depends nonlinearly on $f$.%, cf.~eq.~\eqref{Eq:FundSol}. 

\begin{theorem}\label{first-der}
Suppose that $d\le 3$, that $D>0$ and fix any $x,y\in \mathcal O$. %and fix any $x,y\in \mathcal O$. For the parameter space $\mathcal F\subseteq C^2(\bar{\mathcal O})$ from (\ref{Eq:ParamSpace}), denote the transition density map by
%\[ \Phi:\mathcal F\to [0,\infty),~~~~~~ \Phi[f]=p_{D,f}(x,y).\] 
Then, the Fr\'echet derivate of $\Phi$ at $f\in\mathcal F$ is given by the following linear operator
\begin{equation}
\label{Eq:FrechDeriv}
   D\Phi_f: C^2(\bar{\mathcal O})\to \R, ~~~~~~ D\Phi_f [h] := \int_0^{D} P_{D-s,f}
    \big[ \nabla\cdot(h\nabla p_{s,f}(x,\cdot))\big](y) ds.
\end{equation}
More specifically, for any $R>0$ and $\kappa>0$, there exist $\zeta>0$ and $C>0$ such that for any $h\in C^2(\bar{\mathcal O})$ with $f+h \in \mathcal F$ and $\max \big\{\|f\|_{C^{1+\kappa}},\|f+h\|_{C^{1+\kappa}}\big\} \le R$,
\begin{equation}
\label{Eq:FrechEstimate}
    \frac{\big|\Phi(f+h)-\Phi(f)-D\Phi_f[h]\big|}{\|h\|_{C^1}}\le C \|h\|_{C^1}^{\zeta}=o(\|h\|_{C^1}). 
\end{equation}
Here, $C\equiv C(\mathcal O, d, \fmin, \kappa, R,D)$ can be chosen independently of $x,y\in \mathcal O$ and $f$, $h$ as above.
\end{theorem}

%\begin{theorem}\label{first-der}
%    Suppose $d\le 3$ and fix $\eta>0$ a small constant. Moreover, let us fix $x,y\in \mathcal O$ and denote the transition probabilities by $\Phi(f)=p_{f,D}(x,y)$, for any $f\in C^2(\bar{\mathcal O})$ with $f\ge \fmin >0$. Given any such $f,x,y$, we further define the linear operator acting on functions $h\in C^2(\mathcal O)$,
%    \[D\Phi_f [h] := \int_0^{D} P_{D-s,f}\big[ \nabla\cdot(h\nabla p_{s,f}(x,\cdot))\big](y)  ds. \] 
%    Then, for any $h\in C^2(\bar{\mathcal O})$ with $f+h>\fmin $, there exists $C(\mathcal O, d, \alpha, \fmin )$ and some $\omega <1$, $\zeta>1$ such that
%    \[ \Big|\Phi(f+h)-\Phi(f)-D\Phi_f[h]\Big|\le C \max\big\{1, \|f\|_{C^{1+\eta}}^\omega,\|f+h\|_{C^{1+\eta}}^\omega\big\}\|h\|_{C^1}^{\zeta}. \]
%    These constants can be chosen independently of $x,y\in \mathcal O$ and $f,h$ as above.
%\end{theorem}

%

The preceding theorem states that the map $\Phi$ is Fréchet differentiable with respect to $\|\cdot\|_{C^1}$, with derivative at $f\in\mathcal F$ identified by the linear operator \eqref{Eq:FrechDeriv}; the notation $D\Phi_f[\cdot]$ refers to the fact that $D\Phi_f[\cdot]$ is the (unique) bounded linear operator approximating $\Phi\equiv \Phi_{D,x,y}$ locally at $f$. In particular, this is implied by the remainder estimate \eqref{Eq:FrechEstimate}, which holds uniformly for $f$ and $h$ in balls of the H\"older space $C^{1+\kappa}(\Ocal)$. We refer to \cite[Chapter 5]{E10} for the definition of the norms $\|\cdot\|_{C^1}$ and $\|\cdot\|_{C^{1+\kappa}}$. The proof of the result can be found in Section \ref{sec-gradient-pf} of the Supplement \cite{giordanowangsupplement}.
Note that our derivative is obtained `pointwise' in $x,y\in \mathcal O$, thus rigorously providing a gradient formula for $L_n(\theta)$ conditional on any data $X^{(n)}$ rather than just in `quadratic mean', a weaker regularity condition in terms of which several key results from asymptotic frequentist statistical theory  are formulated \cite{V98}. Differentiability in quadratic mean is, in particular, implied by Theorem \ref{first-der}.
%In the proof, we also devote a particular attention to deriving the needed regularity estimates in the H\"older scale (as opposed to the more commonly used Sobolev one), which is instrumental in identifying the Frechét derivative \eqref{Eq:FrechDeriv} with respect to `uniform norms'.
%In turn, this gives us a principled basis to approach the construction of gradient-based statistical algorithms.
The condition $d\le 3$ is crucially required in several places in the proofs e.g.~to relate pointwise with $L^2$-type norms via Sobolev embeddings $\|\cdot\|_\infty\lesssim \|\cdot\|_{H^2}$.

\begin{comment}
\begin{remark}[H\"older regularity]
In fact, one can show by the same proof techniques that we employ for Theorem \ref{first-der}, that the gradient $D\Phi_f$ is H\"older continuous (with respect to the operator norm), that is, for some exponent $\zeta>0$ and any $R>0$ there exists some constant $C>0$ such that for all $f,g,h$ with $f,f+g\in \mathcal F,~h\in C^{1+\eta}$ and $\max \{\|f\|_{C^{1+\zeta}},\|g\|_{C^{1+\zeta}},\|h\|_{C^{1+\zeta}}\big\}\le R $,
\[  \big|D\Phi_{f}[h] - D\Phi_{f+g}[h]\big| \le C \|g\|_{C^1}^\zeta \|h\|_{C^1}.  \]
Such regularity statements for the gradient are essential for understanding the error incurred by Euler discretisation in our algorithms below, and thus may be important for future work. See Theorem \ref{thm:holder} below for the precise statement and proof.
\end{remark}
\end{comment}

In fact, by the same proof techniques that we employ for Theorem \ref{first-der}, one can show that the Frechét derivative $D\Phi_f$ is H\"older continuous (with respect to the operator norm). Such regularity statements for gradients are essential for understanding the discretisation error incurred by the algorithms constructed below, and thus may be important for future work. The proof is presented in Section \ref{Sec:ProofHolder}.

\begin{theorem}\label{thm:holder} Assume the setting of Theorem \ref{first-der} and let $R>0,~\kappa \in (0,1)$. Then, there is some $\zeta\in (0,1)$ (independent of $R,\kappa$) and some $C>0$ such that for all $f,g,h$ with $f,f+g\in \mathcal F$, $h\in C^{1+\kappa}(\Ocal)$ as well as $\max \{\|f\|_{C^{1+\kappa}},\|g\|_{C^{1+\kappa}},\|h\|_{C^{1+\kappa}}\big\}\le R $,
\[  \big|D\Phi_{f+g}[h] - D\Phi_{f}[h]\big| \le C \|h\|_{C^1}\|g\|_{C^1}^\zeta .  \]
\end{theorem}

%
%
%

%%%%%%%%%%%%%%%%%%%%%%%%%%%%%%%%%%%%%%%%%%%%%%
\subsection{Reduction to elliptic eigenvalue problems}
By the divergence theorem, e.g.~\cite[p.~171]{davies1995spectral}, if $v,w\in H^2_N(\mathcal O)$ then
$$
    \langle \Lcal_fv,w\rangle_{L^2}
    =\int_\Ocal \nabla\cdot(f\nabla v)(x)w(x)dx
    =-\int_\Ocal f(x)\nabla v(x). \nabla w(x)dx
    = \langle v,\Lcal_fw\rangle_{L^2},
$$
which shows that $\Lcal_f$ is self-adjoint with respect to the inner product of $L^2(\Ocal)$. By a suitable application of the spectral theorem (e.g.~\cite[p.~582]{TI}), we  deduce the existence of an orthonormal system of eigenfunctions $(e_{f,j},\ j\ge 0)\subset L^2(\Ocal)$ and of associated (negative) eigenvalues $(\lambda_{f,j},\ j\ge 0)\subset[0,\infty)$ such that
\begin{equation}
\label{Eq:EigProb}
\begin{cases}
\Lcal_f e_{f,j}+\lambda_{f,j}e_{f,j}=0, &  \textnormal{on}\  \Ocal,\\
\partial_\nu e_{f,j}=0, &\textnormal{on}\  \partial\Ocal,
\end{cases}
\qquad j \ge 0.
\end{equation}
We will take throughout the increasing ordering $\lambda_{f,j}\le \lambda_{f,j'}$, $j\le j'$. Then it holds that $e_{f,0}= \vol(\Ocal)^{-1}$ is constant with corresponding eigenvalue $\lambda_{f,0}=0$, independently of $f$. For notational convenience, we shall take $\vol(\Ocal)=1$, so that $e_{f,0} =1$.  Also, by ellipticity, the first non-zero eigenvalue satisfies the `spectral gap' estimate $\lambda_{f,1}\ge c$ for some constant $c>0$ only depending on $\Ocal$ and $\fmin$. The eigenvalues diverge following Weyl's asymptotics $\lambda_{f,j}=O(j^{2/d})$ as $j\to\infty$, with multiplicative constants only depending on $\Ocal$, $\fmin$ and $\|f\|_{L^\infty}$. These facts follow similarly to the arguments for the Neumann-Laplacian (here corresponding to the case $f=1$) developed in \cite[p.~403f]{TI}, in view of the boundedness and the boundedness away from zero of $f$. For details, see \cite[Section 3]{nickl2024inference}.

Using this spectral analysis of the generator, we can represent the action of the transition operator $P_{t,f}[v]$ from \eqref{Eq:TrOp} on any `initial condition' $v\in L^2(\Ocal)$ by
%the solution of the heat equation \eqref{Eq:HeatEq} with any initial condition $u(0,x)=v(x)$, $x\in\Ocal$, for $v\in L^2(\Ocal)$, given by the action  of the , can be characterised as
\begin{equation}
\label{Eq:SpectTrOp}
    %u(t,x)=
    P_{t,f}[v](x)=
    \langle v,1\rangle_{L^2}
    +\sum_{j=1}^\infty
    e^{-\lambda_{f,j}t}\langle v,e_{f,j}\rangle_{L^2}e_{f,j}(x),
    \qquad t\ge 0,\qquad x\in\Ocal.
\end{equation}
Accordingly, the transition densities \eqref{Eq:TRDFs}, which form the integral kernels of $P_{t,f}$, satisfy%(and thus, by the fundamental solutions of \eqref{Eq:HeatEq}),
\begin{equation}
\label{Eq:FundSol}
    p_{t,f}(x,y)=
    1
    +\sum_{j=1}^\infty
    e^{-\lambda_{f,j}t}e_{f,j}(x)e_{f,j}(y),
    \qquad t\ge 0,\qquad x,y\in\Ocal.
\end{equation}
We conclude that for any $f\in \mathcal F$, if we have numerical access to the eigenpairs $(e_{f,j},\lambda_{f,j})$, the likelihood $L_n(f)$ may be evaluated using the spectral formula
\begin{equation}
\label{Eq:SpectrLikelihood}
    L_n(f)=\prod_{i=1}^n\Bigg[ 1
    +\sum_{j=1}^\infty
    e^{-\lambda_{f,j}D}e_{f,j}(X_{(i-1)D})e_{f,j}(X_{iD})
    \Bigg],
    \qquad f\in\Fcal.
\end{equation}

Upon closer inspection, we can also derive a spectral representation of the Frechét derivatives $D\Phi_f$ from Theorem \ref{first-der}. Indeed, since 
%Computationally, we shall handle the Frechét derivative \eqref{Eq:FrechDeriv} with similar tools to those underpinning the routine \eqref{Eq:NumLikelihood} for the evaluation of the likelihood,
the transition density functions $p_{s,f}$ and the transition operators $P_{D-s,f}$ can be expanded with respect to the same eigenpairs $\{(e_{f,j},\lambda_{f,j}), \ j\ge 0\}$ of $\Lcal_f$, we obtain a convenient double series expansion of the integrand $P_{D-s,f}
\big[ \nabla\cdot(h\nabla p_{s,f}(x,\cdot))\big](y)$ in \eqref{Eq:FrechDeriv}. This further allows to separate the spatial and time dependency, leading to a closed form expression for the integration in time, which avoids potential numerical instability caused by the singular behaviour of the integrand for $s\to0$. %which allows to concisely capture the dependence on the variable of integration $s$ and to derive a closed form expression for the resulting integral.
In summary, the following spectral characterisation of the linear operator $D\Phi_f$ is obtained; see Section \ref{Sec:ProofSpectrGrad} for the proof.

\begin{cor}\label{Cor:SpectrGrad}
For any $f\in C^2(\bar{\mathcal O})$ satisfying $\inf_{x\in\Ocal}f(x)\ge \fmin>0$, let $(e_{f,j},\ j\ge 0)\subset L^2(\Ocal)$ be the orthonormal system of the eigenfunctions of the elliptic differential operator in divergence form $\Lcal_f[\cdot]=\nabla\cdot(f\nabla[\cdot])$, with associated eigenvalues $(\lambda_{f,j},\ j\ge 0)\subset[0,\infty)$, solving \eqref{Eq:EigProb}. Then, under the assumptions of Theorem \ref{first-der},
\begin{equation}
\label{Eq:SpectrDeriv}
\begin{split}
    D\Phi_f[h]
    =&
    \sum_{j,j'=1}^\infty %e^{- \lambda_{f,j}D}
    C_{f,j,j'}
    \langle h , \nabla e_{f,j}\cdot\nabla e_{f,j'}\rangle_{L^2} e_{f,j'}(x)e_{f,j}(y),\\
    C_{f,j,j'}:=&
    \begin{cases}
    -D e^{- \lambda_{f,j}D}, & \lambda_{f,j}=\lambda_{f,j'}\\
    (e^{-\lambda_{f,j}D}-e^{- \lambda_{f,j'}D})/(\lambda_{f,j} - \lambda_{f,j'}), & \textnormal{otherwise}.
    \end{cases}
\end{split}
\end{equation}
\end{cor}

%%%%%%%%%%%%%%%%%%%%%%%%%%%%%%%%%%%%%%%%%%%%%%
\subsection{Numerical PDE methods}

%

%\begin{figure}
%\includegraphics[width=4.25cm,height=3.25cm]{Eigenvalues}
%\includegraphics[width=4.7cm]{Eigenfun1}
%\includegraphics[width=4.7cm]{EigenfunJ}
%\caption{\matteo{Left: numerical approximation of the eigenvalues $\lambda_{f,j}$ associated to the conductivity function displayed in Figure \ref{Fig:StatProbl} (left). The eigenvalues exhibit a linear growth as expected from Weyl's asymptotics in two-dimensional domains. Centre and right, respectively: numerical approximations of the first and  $J^{\textnormal{th}}$ non-constant eigenfunctions $e_{f,1}$ and $e_{f,J}$, with $J=17$.}}
%\label{Fig:Eigenpairs}
%\end{figure}

While the eigenpairs $(e_{f,j},\lambda_{f,j})$ are generally not available in closed form, the elliptic eigenvalue problem \eqref{Eq:EigProb} has been widely investigated in the literature on numerical techniques for PDEs, with foundational work by Vainikko \cite{V64} and later landmark contributions in \cite{bramble1973rate,chatelin1973convergence,descloux1978spectral,babuvska1989finite,knyazev2006new} among the others. We further refer to the monograph \cite{babuvska1991eigenvalue} and to the recent survey article \cite{boffi2010finite} for overviews. Specifically, the problem can be tackled with efficient and reliable Galerkin methods (e.g., of finite element type), returning approximations $\{(e_{f,j}^{(\varepsilon)},\lambda_{f,j}^{(\varepsilon)}), \ 1 \le j \le J \}$ of the first $J\in\N$ non-constant eigenfunctions. The superscript $(\varepsilon)$ is used as a proxy for the parameter $\varepsilon>0$ governing the reconstruction quality of the employed numerical method, in the sense that smaller values of $\varepsilon$ yield smaller approximations errors, with convergence when $\varepsilon\to0$ (cf.~Remark \ref{Rem:ApprErrors} below). For example, in standard finite element methods based on piece-wise polynomial functions defined over a triangular mesh covering the domain, $\varepsilon$ is typically chosen to be an upper bound for the side length of the mesh elements.

Based on such numerical techniques, the computation of the transition density functions and the likelihood via the spectral characterisations \eqref{Eq:FundSol} and \eqref{Eq:SpectrLikelihood} respectively can be concretely performed by replacing the eigenpairs with their approximations, and by truncating the series at level $J$, resulting in the simple routines
\begin{equation}
\label{Eq:NumTDF}
    p_{t,f}^{(\varepsilon)}(x,y):=1
    +\sum_{j=1}^J
    e^{-\lambda_{f,j}^{(\varepsilon)} t}e_{f,j}^{(\varepsilon)}(x)e^{(\varepsilon)}_{f,j}(y),
    \qquad t\ge0,
    \qquad x,y\in\Ocal,
\end{equation}
\begin{equation}
\label{Eq:NumLikelihood}
    L_n^{(\varepsilon)}(f):=\prod_{i=1}^np_{D,f}^{(\varepsilon)}(X_{(i-1)D},X_{iD}),
    \qquad f\in\Fcal.
\end{equation}
%The latter can then be used to efficiently implement a large class of likelihood-based statistical methods for inference on the conductivity function $f$, including Bayesian procedures with MCMC algorithms of Metropolis-Hastings type for posterior sampling, wherein at each iteration, for the computation of the acceptance probabilities, a single likelihood evaluation is required.
This is the likelihood approximation that we will employ in Section \ref{Sec:pCN} for the implementation of posterior sampling via the pCN method \cite{CRSW13}.

Turning to the gradient, for any fixed direction $h$, the Frechét derivative $D\Phi_f[h]$ can be efficiently computed according to the formula from Corollary \ref{Cor:SpectrGrad}, analogously truncating the double series at some level $J\in\N$, and replacing the eigenpairs with their numerical approximations.  %to obtain approximations $\{(e_{f,j}^{(\varepsilon)},\lambda_{f,j}^{(\varepsilon)}), \ j=1,\dots,J\}$ of the eigenpairs via numerical PDE techniques, as described in Section \ref{Sec:LikelihoodComp}.
For a stable computation of the two internal series, since for eigenvalues with multiplicities finite element methods generally return distinct approximations that differ by small amounts, the conditions $\lambda_{f,j'} = \lambda_{f,j}$ and $\lambda_{f,j'} \neq \lambda_{f,j}$ should be replaced by the requirements 
$|\lambda^{(\varepsilon)}_{f,j'} - \lambda^{(\varepsilon)}_{f,j}|\le \Tr$ and $|\lambda^{(\varepsilon)}_{f,j'} - \lambda^{(\varepsilon)}_{f,j}|> \Tr$, respectively, for a sufficiently small threshold $\Tr>0$ to be specified by the user. This results in the derivatives evaluation routine
\begin{equation}
\begin{split}
\label{Eq:NumDeriv}
    D\Phi_f^{(\varepsilon)}[h]
    =&
    \sum_{j,j'=1}^J C^{(\varepsilon)}_{f,j,j'}
     \langle h , \nabla e_{f,j}^{(\varepsilon)}\cdot\nabla e_{f,j'}^{(\varepsilon)}\rangle_{L^2} e_{f,j'}^{(\varepsilon)}(x)e_{f,j}^{(\varepsilon)}(y),\\
    C^{(\varepsilon)}_{f,j,j'}:=&
    \begin{cases}
    -De^{- \lambda_{f,j}^{(\varepsilon)}D}, & |\lambda^{(\varepsilon)}_{f,j'} - \lambda^{(\varepsilon)}_{f,j}|\le \Tr\\
(e^{-\lambda^{(\varepsilon)}_{f,j}D}-e^{-\lambda^{(\varepsilon)}_{f,j'}D} )/(\lambda^{(\varepsilon)}_{f,j} - \lambda^{(\varepsilon)}_{f,j'}), & \textnormal{otherwise},
    \end{cases}
\end{split}
\end{equation}
which will serve as a basis for the implementation of gradient-based statistical algorithms. In particular, upon discretising the parameter space $\Fcal$, the log-likelihood gradient can be derived from an application of the chain rule and the above derivative formulae, wherein the directions are identified by the `coordinates' in the chosen discretisation scheme -- see Section \ref{Sec:GradAlgo} below for details.% which will discuss and illustrate the implementation of the unadjusted Langevin algorithm for sampling from the posterior distributions arising from certain truncated Gaussian series priors, and of the gradient descent method for computing the associated MAP estimators.

\begin{remark}[Numerical approximation errors]\label{Rem:ApprErrors} The numerical routines \eqref{Eq:NumLikelihood} and \eqref{Eq:NumDeriv} entails two sources of approximation errors, arising, respectively, from the numerical solution of the elliptic eigenvalue problem \eqref{Eq:EigProb} and the truncation of the series appearing in \eqref{Eq:SpectrLikelihood} and \eqref{Eq:SpectrDeriv}. For the latter, explicit error bounds readily follow from Weyl's asymptotics, the available estimates for the norm of the eigenfunctions, and since $D>0$ is fixed. For instance, by Corollary 1 in \cite{nickl2024inference}, provided that $f$ lies in a Sobolev space $H^s(\Ocal)$ of sufficient smoothness $s > d$, the $j^{\textnormal{th}}$ series term of the numerical likelihood formula \eqref{Eq:NumLikelihood} satisfies, for arbitrarily small $\eta>0$ and for constants $c_1,c_2>0$ only depending on $\Ocal,d,\fmin$, and $\|f\|_{H^s}$,
$$
    e^{-\lambda_{f,j}D}e_{f,j}(X_{(i-1)D})e_{f,j}(X_{iD})
    \le c_1 e^{-c_2 D j^{2/d}}j^{1+\eta}
$$
for all $j$ large enough and all $i=1,\dots,n$, whereupon the tails of the series in \eqref{Eq:SpectrLikelihood} are seen to decay exponentially. Thus, depending on the application at hand and the magnitude of the time lag $D$ between consecutive observations, a relatively low truncation level $J$ in \eqref{Eq:NumLikelihood} may be expected to yield the desired accuracy level, cf.~Section \ref{sec:Num}.

Concerning the numerical solution of the elliptic eigenvalue problem \eqref{Eq:EigProb}, there is a wide literature developing error analyses for a variety of finite element methods; see \cite{boffi2010finite} and references therein. Among these, well-known results for the widespread approach based on piece-wise polynomial approximating functions over a triangulation of the domain, combined with the norm estimates in Corollary 1 in \cite{nickl2024inference}, assuming again that $f\in H^s(\Ocal)$ for some $s>d$, yield the error bound for the eigenvalues
\begin{equation}
\label{Eq:ErrBound}
    |\lambda_{f,j} - \lambda_{f,j}^{(\varepsilon)}|
    \le c_3 j^{c_4} \varepsilon^{c_5},
    \qquad j=1,\dots, J,
\end{equation}
with constants $c_3,c_4,c_5>0$ only depending on $\Ocal,d,\fmin$, and $\|f\|_{H^s}$, and where $\varepsilon$ is the (user-specified) maximal side length for the elements in the triangular mesh, cf.~\cite[Section 10.3]{babuvska1989finite}. Analogous bounds also holds for the approximation errors in Sobolev norms of the eigenfunctions. Note that the estimate \eqref{Eq:ErrBound} deteriorates as the index $j$ grows, due to the more pronounced oscillatory behaviour that the eigenfunctions tend to exhibit at higher frequencies. However, as observed earlier, in the presence of a fixed time lag $D$, only a small number $J$ of eigenpairs are generally needed to approximate the series in \eqref{Eq:SpectrLikelihood} and \eqref{Eq:NumDeriv} with high accuracy, so that we may expect the overall error resulting from the numerical solution of the eigenvalue problem \eqref{Eq:EigProb} by finite element methods to be small even for a relatively coarse triangular mesh. 
\end{remark}

\begin{remark}[Computational cost]\label{Rem:CompCost}
The numerical approximation of the eigenpairs 
%underpinning the proposed computational approach 
can be performed via off-the-shelf PDE solvers implemented in many mathematical and statistical software. Since, in light of Remark \ref{Rem:ApprErrors}, only a small number of eigenpairs is needed in practice, this is generally computationally inexpensive, at least
%for applications 
in low dimensional domains (including the `physical' cases $d=1,2,3$). 
%\matteo{For reference, the approximation of the first $J=17$ eigenpairs associated to the conductivity function $f$ displayed in Figure \ref{Fig:StatProbl} (left) required around $.1$ seconds on a MacBook Pro with M1 processor, using the finite element method implemented in MATLAB R2023a Partial Differential Equation Toolbox, based on a discretisation of the domain with an unstructured triangular mesh comprising 1981 nodes}.
For reference, in the numerical experiments presented in Section \ref{sec:Num}, the typical computation time for this operation was of the order of $.1$ seconds on a MacBook Pro with M1 processor, using the finite element method implemented in MATLAB R2023a Partial Differential Equation Toolbox, based on a discretisation of the domain with an unstructured triangular mesh comprising 1981 nodes.
Since the routines \eqref{Eq:NumLikelihood} and \eqref{Eq:NumDeriv} only require a single numerical solution of the eigenvalue problem \eqref{Eq:EigProb} (along with elementary operations),
%, and only further involves exponentiation, products and sums,
we then obtain an overall comparable computational cost, with no additional bottlenecks. In fact, for the numerical likelihood formula \eqref{Eq:NumLikelihood}, since handling a larger number of observations only implies a linear growth in the number of product terms, our proposed approach is scalable with respect to the sample size.
\end{remark}

%
%
%
%
%

%%%%%%%%%%%%%%%%%%%%%%%%%%%%%%%%%%%%%%%%%%%%%%
\section{Applications to statistical algorithms}\label{Sec:Algo}

We now turn to the problem of estimating the conductivity function $f$ from the low-frequency diffusion data $X^{(n)}$. Leveraging the novel approach developed in Section \ref{Sec:PDEAppraoch}, we gain direct access to the large algorithmic toolbox of likelihood-based nonparametric statistical inference, overcoming the need of computationally expensive data-augmentation techniques \cite{roberts2001inference,beskos2006exact,papaspiliopoulos2013data,van2017bayesian,schauer2017guided}. For illustration, we consider the following methods, within a Bayesian model with Gaussian priors:
\begin{itemize}
\item Gradient-free Metropolis-Hastings MCMC algorithms for posterior sampling;
\item Gradient-based posterior sampling methods of Langevin type;
\item Gradient-based optimisation techniques for the computation of the MAP estimates (i.e.~penalised MLE).
\end{itemize}

%
%
%

%%%%%%%%%%%%%%%%%%%%%%%%%%%%%%%%%%%%%%%%%%%%%%
\subsection{A nonparametric Bayesian approach with Gaussian priors}\label{Sec:GradAlgo}

We shall focus on nonparametric Bayesian procedures with prior distributions arising from Gaussian processes. These are among the most universally used priors on function spaces in applications, e.g.~\cite{Rasmussen2006gaussian,S10,nickl2023bayesian}, \cite[Chapter 7]{GN16} and \cite[Chapter 11]{GV17}, and in the estimation problem at hand they have recently been shown by Nickl \cite{nickl2024inference} to lead to `asymptotically consistent' posteriors that concentrate around the ground truth as the sample size increases. For concreteness, let us follow in this section the prior construction of \cite{nickl2024inference}; more general classes will be considered in Section \ref{App:MoreNum} of the Supplement \cite{giordanowangsupplement}. 

%%%%%%%%%%%%%%%%%%%%%%%%%%%%%%%%%%%%%%%%%%%%%%
\subsubsection{Parametrisation}

In order to incorporate the point-wise lower bound required in the definition \eqref{Eq:ParamSpace} of the parameter space $\Fcal$, we model any $f\in\Fcal$ as 
\begin{equation}
\label{Eq:Reparam}
    f(x)=(\phi\circ F)(x)\equiv\phi(F(x)),\qquad x\in\Ocal,
\end{equation}
for some real-valued $F\in C^2(\Ocal)$ and for some smooth and strictly increasing link function $\phi:\R\to[\fmin,\infty)$ (e.g.~the standard choice $\phi(\cdot)=\fmin+\exp(\cdot)$). Under such bijective reparametrisation, we regard $F$ as the unknown functional parameter to be estimated and discretise it by
\begin{equation}
\label{Eq:DiscretisationScheme}
    F \equiv F_\theta := \theta_0+\sum_{k=1}^K \theta_k \eta_k,
    \qquad K\in\N, \qquad
    \theta_0,\dots,\theta_K\in\R,
\end{equation}
where $\eta_k\in C^2(\Ocal)$, $k\in\N$, is some collection of `basis functions'. Correspondingly, we write $f_\theta:=\phi\circ F_\theta$. We will focus below on the case where $\eta_k:=e_{1,k}$ are the (smooth) non-constant eigenfunctions of the standard Neumann-Laplacian with associated (strictly positive) eigenvalues $\lambda_k:=\lambda_{1,k}$, solving \eqref{Eq:EigProb} with $f=1$, cf.~\cite[p.~403f]{TI}. We remark that other bases could be used as well; see Section \ref{App:MoreNum} of the Supplement \cite{giordanowangsupplement} for extensions to stationary Gaussian process priors via piecewise linear basis functions.

%

%%%%%%%%%%%%%%%%%%%%%%%%%%%%%%%%%%%%%%%%%%%%%%
\subsubsection{Prior and posterior}
We assign to $F$ in \eqref{Eq:Reparam} a truncated Gaussian series prior by endowing the vector of Fourier coefficients $\theta:=(\theta_0,\dots,\theta_K)\in\R^{K+1}$ in \eqref{Eq:DiscretisationScheme} with the diagonal multivariate Gaussian prior
\begin{equation}
\label{Eq:TruncatedGP}
    \theta\sim N(0,\sigma^2\Lambda_\alpha),
    \qquad \Lambda_\alpha:=\textnormal{diag}(1,\lambda_1^{-\alpha},
    \dots,\lambda_K^{-\alpha})\in\R^{K+1,K+1},
    \qquad \alpha,\sigma^2>0.
\end{equation}
In the following we will, in slight abuse of notation, interchangeably write $\Pi(\cdot)$ for the prior \eqref{Eq:TruncatedGP} on $\theta$ as well as for the resulting push-forward priors on $F_\theta$ and $f_\theta$. By Bayes' formula (e.g.,~\cite[p.~7]{GV17}), 
%provided that transition density functions $p_{D,f_\theta}(x,y)$ in \eqref{Eq:TRDFs} give rise to measurable maps with respect to $\theta\in\R^{K+1}$ and $x,y\in\Ocal$, 
the posterior distribution $\Pi(\cdot|X^{(n)})$ of $\theta|X^{(n)}$ has probability density function (with respect to the Lebesgue measure of $\R^{K+1}$)
\begin{equation}
\label{Eq:PostPDF}
    \pi(\theta|X^{(n)})\propto \exp\left(\ell_n(f_\theta)-\frac{1}{2}\theta^T\Lambda_\alpha^{-1}\theta\right),
    \qquad \theta\in\R^{K+1},
\end{equation}
where $\ell_n(f_\theta)$ is the log-likelihood of $f_\theta=\phi\circ F_\theta$, that is
\begin{equation}
\label{Eq:Loglik}
    \ell_n(f_\theta):=\log(L_n(f_\theta))
    =\sum_{i=1}^n\log p_{D,f_\theta}(X_{(i-1)D},X_{iD}).
\end{equation}

%

%The aforementioned characterisation of the transition density functions as the fundamental solutions of the heat equation \eqref{Eq:HeatEq} then places the inferential problem at hand in the context of the vast literature on nonlinear Bayesian inversion in PDE models, e.g.~\cite{S10,reich2015probabilistic,DS16,AMOS19,nickl2023bayesian}, whose algorithmic toolbox we can here access by means of the computational approaches to the likelihood established in Section \ref{Sec:LikelihoodComp}.

%

\begin{remark}[$\alpha$-regular Gaussian priors]\label{Rem:RegGPPrior}
For fixed $K\in\N$, the prior \eqref{Eq:TruncatedGP} induces a multivariate Gaussian distribution on the $(K+1)$-dimensional linear space spanned by the basis functions $\{1,\eta_1,\dots,\eta_K\}$. When the latter are taken to be the Neumann-Laplacian eigenfunctions, the prior converges towards an infinite-dimensional `$\alpha$-regular' Gaussian probability measure with RKHS included into the Sobolev space $H^\alpha(\Ocal)$ as $K\to\infty$, e.g.~arguing as in the proof of Lemma 2.3 in \cite{giordano2022nonparametric}, using Proposition 2 in \cite{nickl2024inference} and the results in Section 11.4.5 of \cite{GV17}. In fact, Theorem 10 in \cite{nickl2024inference} gives a precise growth condition on $K$ as a power of the sample size that leads, under certain additional regularity conditions, to rates of contraction for the associated posterior distribution.
\end{remark}

%

%%%%%%%%%%%%%%%%%%%%%%%%%%%%%%%%%%%%%%%%%%%%%%
\subsubsection{Gradients of log-posterior densities}

%In alternative to the pCN algorithm employed in the previous section, the characterisation of the Fréchet derivatives of the transition densities provided in Theorem \ref{first-der} gives access, through the computational approach developed in Section \ref{Sec:GradComp}, to a vast class of gradient-based methods, wherein we may use the resulting geometric information about the likelihood surface to direct the steps of an iterative scheme, allowing a potentially significant improvement in the efficiency with which the  parameter space is explored. In this section, we shall pursue such ideas in the same Bayesian model introduced in Section \ref{Sec:pCN}, based on the discretisation scheme \eqref{Eq:DiscretisationScheme} and the diagonal multivariate Gaussian prior \eqref{Eq:TruncatedGP} for the vector of Fourier coefficients $\theta\in\R^{K+1}$. 

A careful application of Theorem \ref{first-der} and of the chain rule to the maps $\theta \mapsto \log p_{D,f_\theta}(X_{(i-1)D},$ $X_{iD})$, $i=1,\dots,n$, implies the differentiability of the posterior density \eqref{Eq:PostPDF} (and of its logarithm). We identify the resulting formula for the gradient in the next proposition.

\begin{prop}\label{Prop:GradLogPostPDF}
For fixed $K\in\N$ and a given smooth and strictly increasing function $\phi:\R\to [\fmin,\infty)$, let the posterior probability density function $\pi(\cdot|X^{(n)})$ be as in \eqref{Eq:PostPDF}. Then, $\log \pi(\cdot|X^{(n)})$ is continuously differentiable on $\R^{K+1}$. In particular, $\nabla \log \pi(\cdot|X^{(n)})=(\partial_{\theta_k}\log \pi(\cdot|X^{(n)}), \ k=0,\dots,K)$, where
\begin{equation}
\label{Eq:MargDeriv} 
    \partial_{\theta_k} \log \pi(\theta|X^{(n)})
    =\sum_{i=1}^{n}\partial_{\theta_k}\log p_{D,f_{\theta}}(X_{(i-1)D},X_{iD}) - 
    \lambda_k^\alpha\theta_k,
\end{equation}
and where, denoting by $D\Phi_{f_\theta,i}$, $i=1,\dots,n$, the linear operators defined as in \eqref{Eq:FrechDeriv} with $f=f_\theta$, $x=X_{(i-1)D}$ and $y=X_{iD}$,
$$
    \partial_{\theta_k}\log p_{D,f_{\theta}}(X_{(i-1)D},X_{iD})
    =\frac{ D\Phi_{f_\theta,i} \big[ (\phi' \circ F_\theta) \eta_k\big] }{p_{D,f_\theta}(X_{(i-1)D},X_{iD})}.
$$
\end{prop}

In view of the developments in Section \ref{Sec:PDEAppraoch}, the above formula can be implemented % of the above formula can be approached
by replacing the transition densities $p_{D,f_\theta}$ and the Frechét derivatives $D\Phi_{f_\theta,i}$ with their numerical counterparts $p_{D,f_\theta}^{(\varepsilon)}$ and $D\Phi_{f_\theta,i}^{(\varepsilon)}$ defined as after \eqref{Eq:NumTDF} and \eqref{Eq:NumDeriv} respectively. This results in the gradient evaluation routine $\nabla^{(\varepsilon)} \log \pi(\cdot|X^{(n)}):=(\partial^{(\varepsilon)}_{\theta_k}\log \pi(\cdot|X^{(n)}), \ k=0,\dots,K)^T$, where
\begin{equation}
\label{Eq:NumGrad}
    \partial^{(\varepsilon)}_{\theta_k} \log \pi(\theta|X^{(n)})
    :=\sum_{i=1}^{n}\frac{ D \Phi^{(\varepsilon)}_{f_\theta,i}
    [ (\phi' \circ F_\theta) \eta_k] }{p^{(\varepsilon)}_{D,f_\theta}(X_{(i-1)D},X_{iD})} - 
    \lambda_k^\alpha\theta_k.
\end{equation}
For the latter, we note that a single solution via finite element methods of the elliptic eigenvalue problem \eqref{Eq:EigProb} with $f=f_\theta$ is required across all the partial derivatives, whereupon the exponential coefficients $C^{(\varepsilon)}_{f_\theta,j,j'}$ appearing in \eqref{Eq:NumDeriv} can also be calculated before the specification of the directions $h=(\phi' \circ F_\theta) \eta_k$. Thus the numerical gradient formula \eqref{Eq:NumGrad} is only marginally more computationally expensive than the efficient likelihood routine 
\eqref{Eq:NumLikelihood}.

%
%
%

%%%%%%%%%%%%%%%%%%%%%%%%%%%%%%%%%%%%%%%%%%%%%%
\subsection{Posterior inference via the pCN algorithm}\label{Sec:pCN}

To perform inference based on the posterior density \eqref{Eq:PostPDF}, we begin by considering the class of `zeroth-order' Metropolis-Hasting MCMC sampling methods, whose implementation in the present setting is readily enabled by the numerical likelihood formula \eqref{Eq:NumLikelihood}. We focus here on the widespread pCN algorithm, which is commonly employed in function space settings and inverse problems due to its `robustness' properties with respect to the discretisation dimension $K$ of the Gaussian prior field \cite{CRSW13,HSV14}. Given the low-frequency diffusion data \eqref{Eq:LFData} and under the prior construction \eqref{Eq:TruncatedGP}, the algorithm generates a $\R^{K+1}$-valued Markov chain $(\vartheta_m,\ m\ge 0)$ by repeating the following steps, given some initialisation point $\vartheta_0\in\R^{K+1}$,
\begin{enumerate}
	\item Draw a prior sample $\Psi\sim\Pi(\cdot)$ and for a given `stepsize' $ \delta >0$ define the `proposal' $p := \sqrt{1-2 \delta }\vartheta_m + \sqrt{2 \delta }\Psi$.
\item Set
\begin{equation}\label{Eq:AccProb}
	\vartheta_{m+1} :=
\begin{cases}
	p, & \textnormal{with probability}\ \min\Big\{1,\frac{L_n(f_p)}{L_n(f_{\vartheta_m})}\Big\},\\
	\vartheta_m, & \textnormal{otherwise},
	\end{cases}
\end{equation}
where $L_n$ is the likelihood function in \eqref{Eq:Likelihood}.
\end{enumerate}

The first step only involves the straightforward task of drawing a sample $\Psi$ from the diagonal multivariate Gaussian prior $\Pi(\cdot)$. The second step requires the evaluation of the proposal likelihood $L_n(f_p)$, which is achieved through the routine \eqref{Eq:NumLikelihood}. Apart from the latter, all the operations involved within the two steps are elementary; therefore the computational complexity of each pCN iteration is largely driven by the cost of evaluating the numerical likelihood formula \eqref{Eq:NumLikelihood}. One can thus expect generally modest computation times, cf.~Remark \ref{Rem:CompCost}. An excellent scalability with respect to the sample size is also obtained.

The proposals and acceptance probabilities prescribed by the pCN algorithm are of Metropolis-Hastings type, e.g.~Proposition 1.2.2 in \cite{nickl2023bayesian}, so that the generated Markov chain can be shown to be reversible and to have invariant probability measure equal to the posterior distribution $\Pi(\cdot|X^{(n)})$, cf.~\cite{tierney1998note}. The posterior mean estimator $\bar \theta_n := E^\Pi[\theta|X^{(n)}]$ is then numerically evaluated through the Monte Carlo average $\bar \vartheta_M:=(M+1)^{-1}\sum_{m=0}^M\vartheta_m$ for some large $M\in \N$, typically after discarding an initial `burnin' batch of iterates.
%in order to account for the convergence from the initialisation point towards equilibrium. Similarly, the posterior quantiles can be approximated by the empirical quantiles of the MCMC samples.
The corresponding posterior mean %estimator 
$\bar F_n:=E^\Pi[F|X^{(n)}]$ of the %unknown 
functional parameter $F$ is given by
\begin{equation}
\label{Eq:PostMean}
    F_{\bar \vartheta_M}=\bar \vartheta_{M,0}+\sum_{k=1}^K \bar \vartheta_{M,k}\eta_k.
\end{equation}

%
%
%

%%%%%%%%%%%%%%%%%%%%%%%%%%%%%%%%%%%%%%%%%%%%%%
\subsection{Posterior inference via the unadjusted Langevin algorithm}\label{Sec:Langevin}

The results of Section \ref{Sec:PDEAppraoch} can also be used to build a variety of gradient-based iterative methods, which allows the  incorporation of geometric information on the likelihood surface. Here, we shall focus on MCMC algorithms of Langevin type, arising from the discretisation of the SDE
\begin{equation}
\label{Eq:LangevinEq}
    d \vartheta_t = \frac{1}{2}\nabla\log\pi(\vartheta_t|X^{(n)})dt + dB_t,
    \qquad t\ge 0,
    \qquad \vartheta_0\in\R^{K+1},
\end{equation}
where $(B_t,\ t \ge 0)$ is a standard $(K+1)$-dimensional Brownian motion, $\pi(\cdot|X^{(n)})$ is the posterior density from (\ref{Eq:PostPDF}) and $\nabla\log\pi(\cdot|X^{(n)})$ is as after \eqref{Eq:MargDeriv}. The solution $(\vartheta_t, \ t\ge0)$ to \eqref{Eq:LangevinEq} is well-known to have stationary distribution equal to the posterior $\Pi(\cdot|X^{(n)})$, cf.~p.~45-47 in \cite{BGL14}. The standard Unadjusted Langevin Algorithm (ULA) arises from the Euler-Maruyama discretisation of (\ref{Eq:LangevinEq}). For a stepsize $\delta>0$, the ULA generates an $\R^{K+1}$-valued Markov chain $(\vartheta_m,\ m\ge 0)$ via
\begin{equation}
\label{Eq:ULA}
    \vartheta_{m+1}
    =\vartheta_m+\frac{\delta}{2}\nabla\log\pi(\vartheta_m|X^{(n)})
    +\sqrt{\delta} B_m,
    \quad B_m\iid N(0,I_{K+1}),
    \quad \vartheta_0\in\R^{K+1}.
\end{equation}

From \eqref{Eq:ULA}, it is seen that the central operation underlying each step of the ULA is the computation of the gradient of the log-posterior density for the current state of the chain. In the problem at hand, this task can be efficiently (and scalably) performed in concrete via the evaluation routine \eqref{Eq:NumGrad}. 

Using the gradient computations developed here, further MCMC methods such as the popular Metropolis-adjusted Langevin Algorithm (MALA, see \cite[Section 1.4]{RT96} and \cite[Section 4.3]{CRSW13}) can similarly be implemented. %For example, in the popular Metropolis-adjusted Langevin Algorithm the ULA updates \eqref{Eq:ULA} are used to construct the proposals within a Metropolis-Hastings architecture, and are subject to an accept-reject step similar to the one specified for the pCN algorithm in \eqref{Eq:AccProb}, leading to similar dimension-robustness properties.

\subsection{MAP estimation via gradient descent}
\label{Sec:GradDesc}

Lastly, we also consider optimisation techniques for the posterior density $\pi(\cdot|X^{(n)})$ in \eqref{Eq:PostPDF}. The associated MAP estimator is defined as any element
\begin{equation}
\label{Eq:MAP}
    \hat\theta_n\in \underset{\theta\in\R^{K+1}}{\textnormal{argmax}}
    \left\{\log\pi(\theta|X^{(n)})\right\}
    = \underset{\theta\in\R^{K+1}}{\textnormal{argmax}} \left\{
    \ell_n(\theta) - \frac{1}{2}\theta^T\Lambda_\alpha^{-1}\theta
    \right\}.
\end{equation}
Identifying conductivities $f=f_\theta\in\Fcal$ with the corresponding coefficient vectors $\theta\in\R^{K+1}$ through the parametrisation \eqref{Eq:DiscretisationScheme}, the MAP estimator $\hat\theta_n$ can be interpreted as a discretisation of the penalised MLE
$$
    \hat f_n\in \underset{f\in\Fcal}{\textnormal{argmax}}
    \left\{
    \log [L_n(f)] - \|f\|^2_{H^\alpha}
    \right\},
$$
with Sobolev norm penalty corresponding to the limiting RKHS norm for the truncated Gaussian series prior \eqref{Eq:TruncatedGP} when the eigenbasis of the Neumann-Laplacian is used, cf.~Remark \ref{Rem:RegGPPrior}.

The gradient-descent algorithm for solving \eqref{Eq:MAP} is then given by
\begin{equation}
\label{Eq:GradDesc}
    \vartheta_{m+1}
=\vartheta_m+\delta\nabla\log\pi(\vartheta_m|X^{(n)}),
    \qquad \vartheta_0\in\R^{K+1},
    \qquad \delta>0,
\end{equation}
%equipped with 
with a suitable stopping criterion. The gradient evaluations can again be tackled via \eqref{Eq:NumGrad}.
%The MAP estimator $\hat\theta_n$ is then concretely computed through the output $\vartheta_M$ of the $M^{\textnormal{th}}$ step, for some $M\in\N$ typically chosen according to a suitable convergence criterion. The associated MAP estimator $\hat F_n$ of the reparametrised functional coefficient $F$ is approximated by
%$$   F_{\vartheta_M}:=\vartheta_{M,0}+\sum_{k=1}^K\vartheta_{M,k}e_k.
%$$

%

%Similarly to the ULA for posterior sampling described in Section \ref{Sec:Langevin}, the single gradient evaluation required at each iteration of the gradient descent algorithm \eqref{Eq:GradDesc} can be tackled via the efficient routine \eqref{Eq:NumGrad}.

%

\begin{remark}[Multimodality of the posterior]
In view of the nonlinear dependence of the transition densities $p_{D,f}$ on the conductivity function $f$, cf.~\eqref{Eq:FundSol}, the log-likelihood and the log-posterior density are generally non-concave and potentially multimodal. This can indeed be seen in our numerical experiments; see Section \ref{App:MoreNum} of the Supplement \cite{giordanowangsupplement}. Thus, while gradient-based iterative schemes may be able to compute global maximisers under specific conditions and `warm starts' (see e.g.~\cite[Chapter 11]{EHN96}, \cite{nickl2022polynomial}), 
in general one should only expect to reach convergence towards a local optimum. The global convergence behaviour of gradient-based schemes in SDE models is a highly challenging topic for future research.

%finding a global solution to the maximisation problem \eqref{Eq:MAP} may be challenging in practice, and we should regularly only expect to reach convergence towards a local one. 

%

On the other hand, optimisation methods are generally computationally attractive. In our simulation studies, the MAP estimates yielded satisfactory reconstructions while only requiring a small fraction of the number of iterations compared to MCMC; see Section \ref{Sec:GradResults}. In the present setting, the MAP estimate could also be used as a fast-to-compute initialisation point for the posterior samplers described in Sections \ref{Sec:pCN}-\ref{Sec:Langevin} to reduce burnin times.
\end{remark}

%%%%%%%%%%%%%%%%%%%%%%%%%%%%%%%%%%%%%%%%%%%%%%
\section{Numerical experiments}\label{sec:Num}

We tested the proposed computational approach in extensive simulation studies, using the non-parametric Bayesian model introduced in Section \ref{Sec:GradAlgo}.
%Further numerical results are deferred to Section \ref{App:MoreNum} of the Supplement \cite{giordanowangsupplement}.
All the numerical experiments were carried out on a MacBook Pro with M1 processor and 8GB RAM. The required finite element computations for elliptic eigenvalue problems were performed using MATLAB Partial Differential Equation Toolbox (R2023a release), based on a triangulation of the domain $\Ocal$ comprising 1981 nodes (with maximal side length $\varepsilon=.05$).

%
%
%

%%%%%%%%%%%%%%%%%%%%%%%%%%%%%%%%%%%%%%%%%%%%%%
\subsection{Data generation}\label{sec:DataGen}
Throughout, we worked on the unit area disk $\Ocal=\{(x_1,x_2)\in\R^2:x_1^2+x_2^2\le\pi^{-1}\}$ and took the true conductivity function to be $f_0(x_1,x_2)= 1.1+10e^{-(7.25x_1-1.5)^2-(7.25x_2-1.5)^2}+10e^{-(7.25x_1+1.5)^2-(7.25x_2-1.5)^2}$, cf.~Figure \ref{Fig:PostMean} (left). The trajectory of the diffusion process \eqref{Eq:SDE} was simulated via the Euler-Maruyama scheme, generating the sequence of `continuous-time' states $(x_r, \ r\ge 0)\subset\Ocal$ by the iteration
$$
	x_{r+1} = x_r + \nabla f_0(x_r) \delta_t + \sqrt { 2 f_0(x_r)} \delta_t W_r, 
	\qquad W_r\iid N(0,1),
$$
modified to incorporate (elastic) boundary reflections, by reflecting any `proposed' iterate falling outside of $\Ocal$ with respect to the tangent line at the boundary. We set the initial condition $x_0$ at the origin and the time stepsize $ \delta_t = 5\times 10^{-6}$. We repeated the scheme $5\times10^8$ times, giving the time horizon $T = 5\times10^8 \times  \delta_t = 2500$. From the simulated trajectory, discrete observations $X^{(n)}=(X_0,X_D,\dots, X_{nD})$ were sampled with time lags $D=.05$ according to $X_{iD}:=x_{iD/\delta_t}$. Note that $D/\delta_t=10^4\gg 1$, resulting in realistic low-frequency data.

\subsection{Parameterisation and prior specification} Across the experiments, we used the truncated Gaussian series priors from \eqref{Eq:TruncatedGP} with truncation level $K=68$, regularity $\alpha=1$ and variability $\sigma^2=500$. Moreover, we used the parameterisation of conductivities $f$ given by (\ref{Eq:Reparam})-(\ref{Eq:DiscretisationScheme}), with link function $\phi(\cdot) = \fmin +\exp (\cdot)$ and $\fmin=.1$. The $L^2$-norm of the reparameterised ground truth $F_0=\log(f_0-f_{\min})$ is $\|F_0\|_2=.8727$, while the $L^2$-approximation error resulting from `projecting' $F_0$ onto the linear space spanned by the eigenfunction $\{1,\eta_1,\dots,\eta_K\}$ equals $.0848$, leading to a `benchmark' relative error of $9.72\%$.

%%%%%%%%%%%%%%%%%%%%%%%%%%%%%%%%%%%%%%%%%%%%%%
\subsection{Results for the pCN algorithm}\label{Sec:pCNResults}

\begin{figure}
\includegraphics[width=3.5cm]{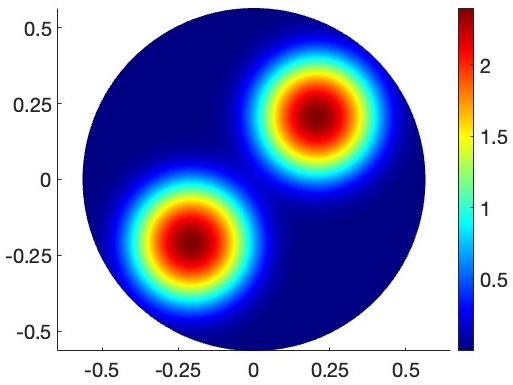}
\includegraphics[width=3.5cm]{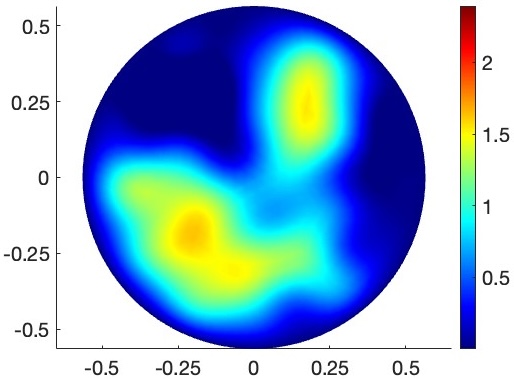}
\includegraphics[width=3.5cm]{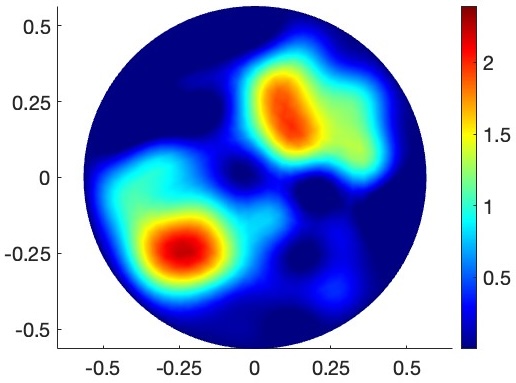}
\includegraphics[width=3.5cm]{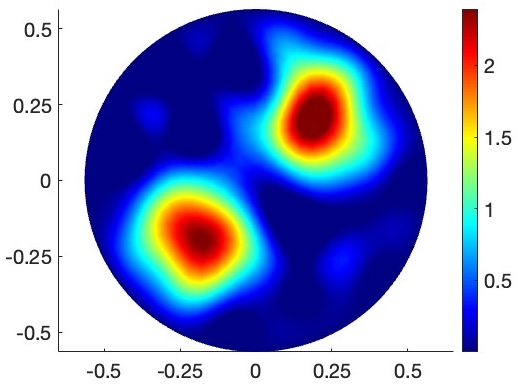}
\caption{Left to right: the (reparametrised) true conductivity function $F_0$, and the posterior mean estimates $\bar F_n$ for $n=500,~%1000,
2500,~%5000,10000,
50000$, obtained via the pCN algorithm. Computation times ranged between 55 and 59 minutes.}
\label{Fig:PostMean}
\end{figure}

\begin{table}
\caption{$L^2$-estimation errors for the posterior mean estimates (obtained via the pCN algorithm)}
\label{Tab:Errs}
\centering
\renewcommand{\arraystretch}{1.2}
\begin{tabular}{ c|c|c|c|c|c|c} 
$n$                              & 500 & 1000 & 2500 & 5000 & 10000 & 50000\\
 \hline
 $\|F_0 - \bar F_n\|_2$ & .4846 & .3953 & .3532 & .3343 & .3188 & .2097 \\
  \hline
 $\|F_0 - \bar F_n\|_2/\|F_0\|_2$ & $55.73\%$ & $45.29\%$ & $40.47\%$ & $38.30\%$ & $36.53\%$ & $24.02\%$
\end{tabular}
\end{table}

\begin{figure}
\includegraphics[width=4.7cm,height=3cm]{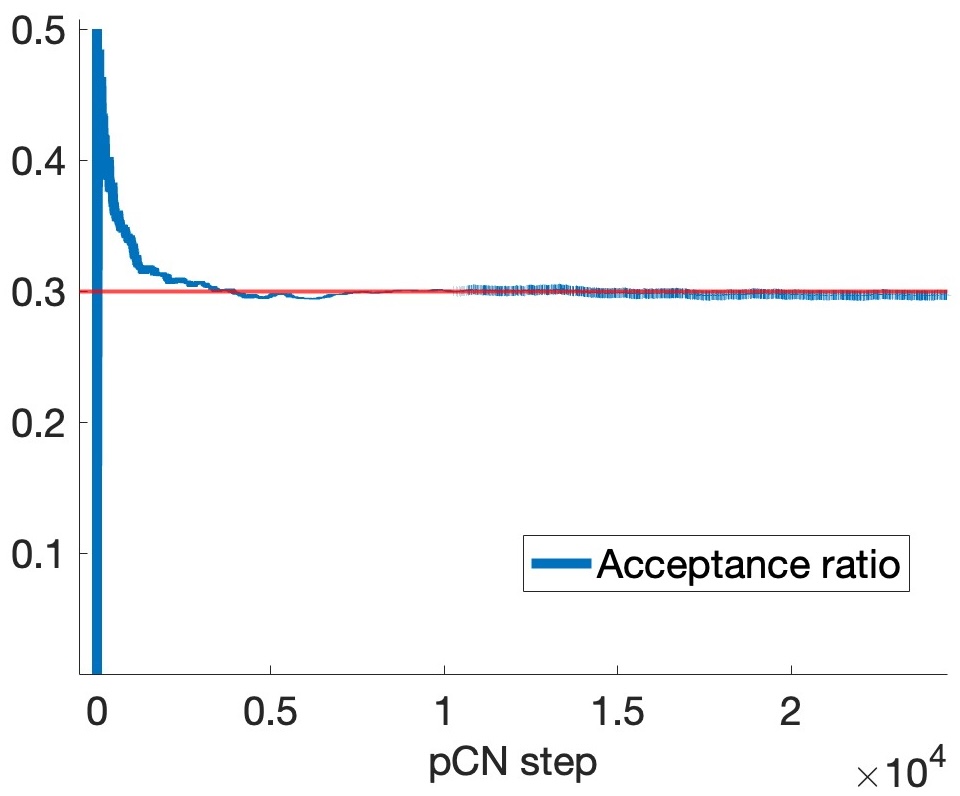}
\includegraphics[width=4.7cm,height=3cm]{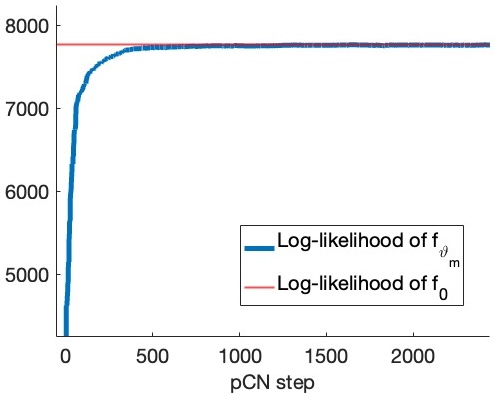}
\includegraphics[width=4.7cm,height=3cm]{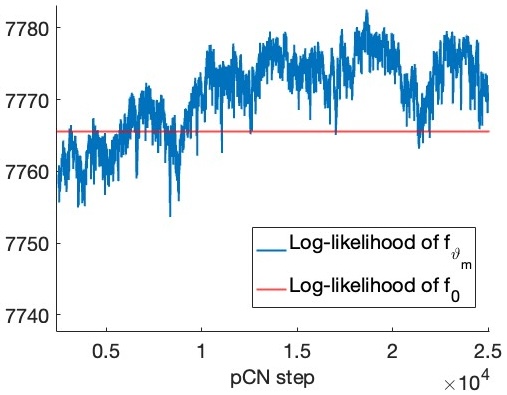}
\caption{Left: the acceptance ratio along the 25000 iterations of the pCN algorithm, for the case $n=50000$. Centre and right, respectively: the log-likelihood $\log(L_n(f_{\vartheta_m}))$ for the first 2500 chain steps, and for the steps from the $2500^{\textnormal{th}}$ to the $25000^{\textnormal{th}}$, again for $n=50000$.}
\label{Fig:pCNPlots}
\end{figure}

The Monte Carlo approximations $F_{\bar \vartheta_M}$ for the posterior mean estimators $\bar F_n$ of $F=\log(f-\fmin)$ are plotted in Figure \ref{Fig:PostMean}, with increasing sample sizes $n=500,2500,50000$.  As expected from the posterior consistency result of \cite{nickl2024inference}, they show a progressively improved reconstruction, see Table \ref{Tab:Errs}.

The stepsize in the pCN algorithm was chosen (depending on the sample size) amongst $\delta\in\{.01,.005,.005,.001$, $.001,.0001\}$ 
to achieve acceptance probabilities of around $30\%$ after the burnin phase, cf.~Figure \ref{Fig:pCNPlots} (left). Each run was initialised at the `cold start' $\vartheta_0=0$ and terminated after $M=25000$ iterations, with 2500 burnin samples. During such burnin phases, the generated chains were observed to effectively move towards regions with high posterior probability; see Figure \ref{Fig:pCNPlots} (centre and right).

The evaluation of the likelihood ratios to compute the acceptance probabilities \eqref{Eq:AccProb} was carried out via the routine \eqref{Eq:NumLikelihood}. The truncation level $J$ for the series in \eqref{Eq:NumLikelihood} was chosen adaptively across the iterations so to include all the approximate eigenvalues $0<\lambda_1^{(\varepsilon)}\le\dots\le\lambda^{(\varepsilon)}_J\le 250$, after which the coefficients in \eqref{Eq:NumLikelihood} satisfy $e^{-D\lambda_j^{(\varepsilon)}}\le e^{-.05\times 250}=3.7267\times 10^{-6}$ for all $j> J$. For sample size $n=50000$ and $M=25000$ MCMC iterates, computation took approximately 56 minutes, with an average runtime of .13 seconds per iteration.

Further details on the results for the pCN algorithm are provided in Appendix \ref{App:MoreNum}, including trace-plots of the marginal posterior distributions of the coefficients $\theta_0,\dots,\theta_K$, as well as additional numerical experiments investigating the role of the initialisation point, other choices of Gaussian priors, and the recovery of different ground truths.

%
%
%

%%%%%%%%%%%%%%%%%%%%%%%%%%%%%%%%%%%%%%%%%%%%%%
\subsection{Results for the ULA}\label{Sec:ULAResults}

\begin{figure}
\includegraphics[width=3.5cm]{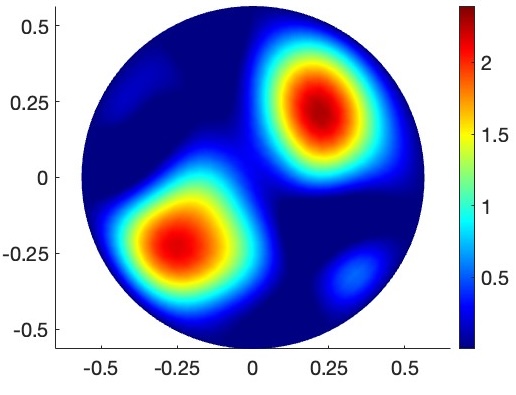}
\includegraphics[width=6cm,height=3cm]{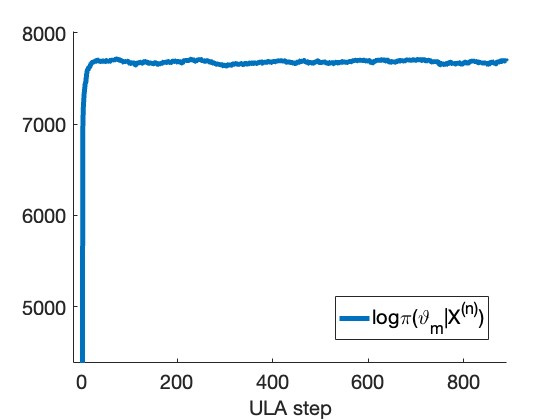}
\caption{Left: the posterior mean estimate $\bar F_n$, for $n=50000$, obtained via the ULA, to be compared to the ground truth $F_0$ shown in Figure \ref{Fig:PostMean} (left). The overall computational time was 90 minutes. Right: the log-posterior density $\log\pi(\vartheta_m|X^{(n)})$ for the first 1000 chain steps.}
\label{Fig:ULAEstim}
\end{figure}

In the same setting as in Section \ref{Sec:pCNResults}, we next consider the ULA, cf.~Section \ref{Sec:Langevin}. Figure \ref{Fig:ULAEstim} (left) shows the ULA posterior mean estimate $\bar F_n$ of the reparametrised conductivity $F$. The $L^2$-estimation error was $.20327$ (relative error $23.28\%$). A total of $M=10000$ iterations were used, with stepsize $\delta=.000025$ and initialisation $\vartheta_0=0$. Compared to the pCN method, through the incorporation of the gradient, the ULA chain was observed to move very efficiently towards the regions of higher posterior probability, see Figure \ref{Fig:ULAEstim}. Accordingly, a shorter burnin of 250 samples was employed. Within each iteration, the gradient evaluation was performed through the routine \eqref{Eq:NumGrad}. The computational parameters for the finite element method were specified exactly as for the pCN method. The runtime was around 90 minutes, with average computation time of .54 seconds per iteration.

%

%
%

%%%%%%%%%%%%%%%%%%%%%%%%%%%%%%%%%%%%%%%%%%%%%%
\subsection{Results for the MAP estimator (via gradient descent)}\label{Sec:GradResults}

We conclude with the optimisation methods from Section \ref{Sec:GradDesc}. The MAP estimate is shown in Figure \ref{Fig:MAPEstim} (left), with associated $L^2$-estimation error equal to $.2622$ (relative error $30\%$). The MAP estimate was computed via gradient descent, initialised at $\vartheta_0=0$ (`cold start') and with stepsize $\delta=.00001$. In total, $M=116$ iterations were necessary to achieve convergence; see Figure \ref{Fig:MAPEstim} (centre). The required gradient evaluations were performed exactly as for the ULA in Section \ref{Sec:ULAResults}, resulting in an overall computation time of around 1 minute.

Interestingly, the obtained MAP estimate appears to be visually close to the posterior mean estimates calculated via the pCN algorithm and the ULA, with comparable estimation error. This is despite the multimodality of the posterior; see Section \ref{App:MoreNum} of the Supplement \cite{giordanowangsupplement} for further illustration and discussion. 

%
%
%
%
%

%

%%%%%%%%%%%%%%%%%%%%%%%%%%%%%%%%%%%%%%%%%%%%%%
\section{Summary and further discussion}\label{Sec:Discuss}

\begin{figure}
\includegraphics[width=3.5cm]{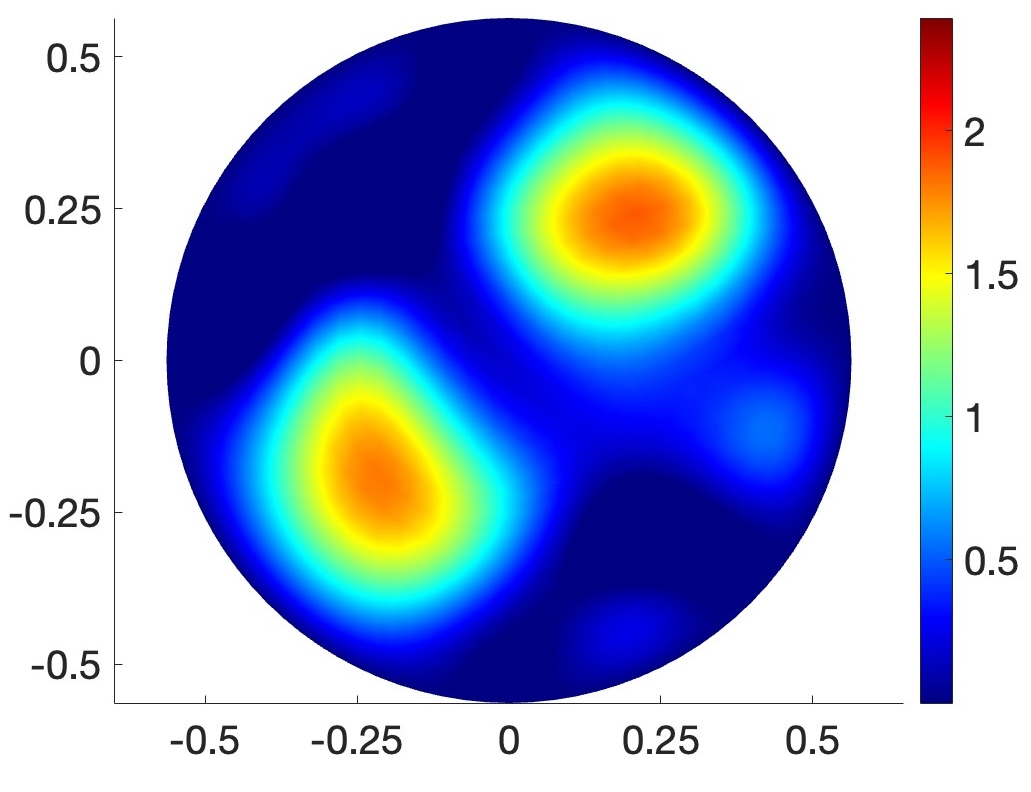}
\includegraphics[width=4.7cm,height=3cm]{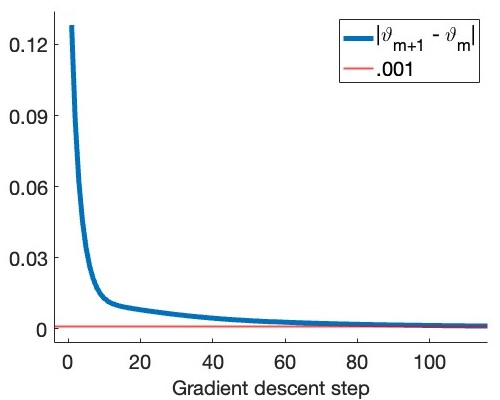}
\includegraphics[width=4.7cm,height=3cm]{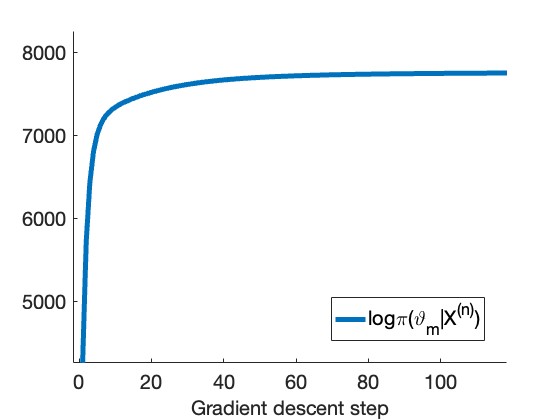}
\caption{Left: the MAP estimate, with $n=50000$, computed by the gradient descent, to be compared to the ground truth $F_0$ shown in Figure \ref{Fig:PostMean} (left). The overall computational time was 1 minute. Centre: the distances $|\vartheta_{m+1} - \vartheta_m|$ between consecutive gradient descent iterates.
Right: the log-posterior density $\log\pi(\vartheta_m|X^{(n)})$ along the gradient descent steps.}
\label{Fig:MAPEstim}
\end{figure}

We have developed novel approaches to statistical inference with discrete observations from a class of stochastic diffusion models. Using the PDE characterisation of the transition density functions as fundamental solutions of certain divergence-form heat equations, we have employed abstract parabolic PDE arguments to derive a novel closed-form expression for gradients of likelihoods. Leveraging spectral theory for the elliptic generators, we have then derived powerful series representations for the likelihood and its gradient. Computationally, we have argued (and demonstrated through simulation studies) that this leads to reliable and efficient numerical routines based on finite element methods for elliptic eigenvalue problems, that can be used %as the basis 
for the construction of gradient-free and gradient-based statistical methods. This has been illustrated through the implementation and empirical investigation of some commonly used MCMC algorithms and optimisation techniques for nonparametric Bayesian inference with Gaussian priors.

While our work provides novel theoretical and methodological contributions to the challenging problem of inference with low-frequency diffusion data, it also raises several interesting open research questions. We conclude by discussing some of these. %that are the object of ongoing work. 

%
%
%
%%%%%%%%%%%%%%%%%%%%%%%%%%%%%%%%%%%%%%%%%%%%%%
\subsection{Extensions to general reversible models with potential energy}
\label{Eq:Potential}

The scope of the present investigation has been focused on diffusions of the form \eqref{Eq:SDE}. Extensions to other models are of primary interest and a natural avenue for future research. Many of the ideas developed here can be generalised to the important case where the diffusing particles are subject to a spatially-varying `potential energy' field $U:\Ocal\to\R$ which exerts a `displacement force' directed towards the local minima of $U$ via the gradient vector field $\nabla U$. This results in a general reversible diffusion process with gradient drift vector field (e.g.~\cite[p.~47]{BGL14}),
\begin{equation}
\label{Eq:GeneralSDE}
    dX_t = f(X_t)\nabla U(X_t)dt + \nabla f(X_t)dt + \sqrt{2f(X_t)}dW_t + \nu(X_t)dL_t, \qquad t\ge 0,
\end{equation}
with associated (Gibbs-type) invariant density and infinitesimal generator
\begin{equation}
\label{Eq:GeneralInvDistr}
\mu(x)\propto e^{-U(x)},\qquad x\in\Ocal;
\qquad
    \Lcal_{f,\mu}\phi:=\frac{1}{\mu}\nabla\cdot(\mu f\nabla \phi).
\end{equation}
The latter is again in divergence form, and is self-adjoint for the $L^2$-inner product with respect to $\mu$ (as opposed to the standard $L^2$-inner product for model \eqref{Eq:SDE}). Conditional inference on the conductivity function $f$ given any $U$ may then be pursued with minimal modifications to the approach developed in Section \ref{Sec:PDEAppraoch}. In fact, we also expect the gradient formulae from Theorem \ref{first-der} and Corollary \ref{Cor:SpectrGrad} to extend to generators $\mathcal L_{f,\mu}$ at the expense of additional technicalities, by adapting our parabolic PDE techniques.

In the realistic scenario of $U$ in \eqref{Eq:GeneralSDE} being also unknown, 
we envision decomposing the problem by obtaining a preliminary estimate $\hat\mu_n$ of the invariant probability density function $\mu$ (for instance via techniques from~\cite{GHR04,NS17, dexheimer2022adaptive,dalalyan2007asymptotic,giordano2022nonparametric}).
Based on $\hat\mu_n$, the statistical analysis on $f$ can then be performed by replacing the generator $\Lcal_{f,\mu}$ in \eqref{Eq:GeneralInvDistr} with a plug-in estimate $\Lcal_{f,\hat\mu_n}$. Alternatively, a joint Bayesian model can be considered by endowing the additional drift component $\nabla U$ in \eqref{Eq:GeneralSDE} with a prior distribution (such as those employed in \cite{giordano2022nonparametric}), which would then require to modify our MCMC algorithms to be of `Gibbs' type. However, these generalizations are beyond the scope of the present work.

%For posterior sampling, this would then require to employ the MCMC methods considered here within a Gibbs-type structure, alternating runs targeting the conditional posterior distribution of $f$ (given $ \nabla U$), to that of $\nabla U$ (given $f$).

%%%%%%%%%%%%%%%%%%%%%%%%%%%%%%%%%%%%%%%%%%%%%%
\subsection{Bounds on the computational complexity via gradient stability}
\label{Sec:GradStab}

Another important %related 
issue concerns the %precise characterisation of the 
computational complexity of the employed posterior sampling and optimisation algorithms, %we develop,
%for statistical inference, 
going beyond the qualitative invariance properties discussed in Sections \ref{Sec:pCN} and \ref{Sec:Langevin}. A recent line of work initiated by Nickl and Wang in \cite{nickl2022polynomial} has shown that polynomial-time bounds on the iteration complexity of MCMC methods may be obtained 
%that bounds on the required number of operations that scale polynomially with respect to the sample size, the dimension of the parameter space and the precision level can be obtained 
under `local convexity' properties for the (negative) log-likelihood near the ground truth, quantified via a stability estimate for the log-likelihood gradient (i.e. a lower bound for its minimal eigenvalue), jointly with regularity properties of the associated Hessian matrix.

In the present setting, this program can indeed be pursued, albeit under a considerable amount of technical work. In particular, one may approach the verification of the key gradient stability condition %can be shown to follow using 
using the characterisation of the Frechét derivatives in \eqref{Eq:FrechDeriv} and \eqref{Eq:SpectrDeriv}. %using certain finer properties of the eigenpairs of the infinitesimal generator. 
Similarly, our `bootstrap' PDE arguments also provide a blueprint for deriving the required higher-order regularity properties for the log-likelihood.% can also approached, by iterating the first-order arguments developed in the proof of Theorem \ref{first-der}.

\subsection{Small-time Gaussian asymptotics}\label{sec:gauss-appr} The heat kernels of the SDE (\ref{Eq:SDE}) are well-known to satisfy small-time heat kernel asymptotics of the form
\begin{equation}\label{eq:gauss-asymptotics}
    \lim_{t\to 0} -4t \log p_{t,f}(x,y) = d_f^2(x,y), 
\end{equation}
where $d^2_f(x,y)$ is an intrinsic Riemannian metric associated to the diffusion coefficient $f$. Such heat kernel asymptotics have been studied extensively in analysis and probability, see, e.g., seminal works by Varadhan \cite{V67}, Aronson \cite{A67} and \cite{N97,SC10}. The last reference provides an excellent survey containing more references and quantitative counterparts of (\ref{eq:gauss-asymptotics}).

For small observation distances $D\approx 0$, these asymptotics naturally give rise to numerical approximation for $p_{D,f}(\cdot,\cdot)$, which in turn yields an approximation of the $\log$-likelihood $\sum_{i=1}^n\log p_{D,f}(X_{(i-1)D},X_{iD})$ that may be interpreted as a Riemann sum counterpart of the continuous-time $\log$-likelihood given by Girsanov's theorem, cf.~\cite[Section 1]{papaspiliopoulos2013data}.
%which could also be = \blue{would this correspond to time-discretising girsanov?}.
Interestingly, while the accuracy of our spectral approximations improves rapidly as $D$ increases (due to the exponential decay of eigenvalues, see Remark \ref{Rem:ApprErrors}), the accuracy of the above Gaussian approximations in general deteriorate as $D$ grows; in particular, a convergent approximation cannot be expected when $D$ is bounded away from $0$. Understanding the precise phase transition and developing data-driven rules to decide between small-time approximations versus our spectral PDE approach are interesting avenues for future research. 

%depending on the statistical data available, 
%\matteo{phase transition, switching point?}
%
%depending on the statistical data available, 
%\matteo{phase transition, switching point?}
%

\subsection{Diffusion generative models}\label{sec:generativemodels}

Generative models based on SDEs \cite{song21, cao2024survey}, such as the prominent \textit{DALL-E} \cite{dall-e} and \textit{Stable Diffusion} \cite{stable} models for images, have recently achieved 
%, introduced originally in , have 
spectacular success. %We mention the prominent \textit{DALL-E} \cite{dall-e} and \textit{Stable Diffusion} \cite{stable} image generation models. 
Given high-dimensional data $X_1,...,X_n$, these models aim to sample from some unknown (possibly conditional) probability distribution $P_0$ underlying the data. %, as well as its conditional distributions (e.g.~given through text prompts).
This is achieved by first `diffusing' the data through an ergodic (e.g.~Ornstein-Uhlenbeck type) forward process and subsequently reversing the SDE to return to the data-generating distribution. The drift of the reverse process crucially features the \textit{score function} $\nabla \log p_t(\cdot;X_1,\dots,X_n)$, where $p_t(\cdot;X_1,\dots,X_n)$ is the time-$t$ marginal distribution of the forward processes. The key idea is to learn an approximation $\hat s(t,x) \approx \nabla \log p_t(x;X_1,\dots,X_N)$ which then provides an estimate for the `true' score function associated to $P_0$ -- see e.g.~\cite{song21} for details and \cite{OSS23} for a nonparametric statistical analysis.

Both in diffusion generative models and in our work, unknown SDE parameters are inferred from data. The key difference is that in our setting the data-generating process is inherently governed by SDEs while in generative modelling, SDEs are used to induce an (approximate) `coupling' between $P_0$ and the stationary distribution of the forward process. It would be interesting to utilise the present developments in the context of generative models -- for instance, divergence form reflected diffusions (\ref{Eq:SDE}) provide a flexible class of ergodic forward processes on bounded domains with uniform equilibrium distribution, which may be used to enhance existing generative models based on reflected diffusions \cite{reflected-generative}.

%
%
%
%
%

%%%%%%%%%%%%%%%%%%%%%%%%%%%%%%%%%%%%%%%%%%%%%%
\section{Proofs of Theorem \ref{thm:holder}, Corollary \ref{Cor:SpectrGrad} and Proposition \ref{Prop:GradLogPostPDF}}

In order to achieve a self-contained exposition, we deferred the proof of our main result, Theorem \ref{first-der}, to Section \ref{sec-gradient-pf} of the Supplement \cite{giordanowangsupplement}. We here present the proofs of Theorem \ref{thm:holder}, Corollary \ref{Cor:SpectrGrad} and Proposition \ref{Prop:GradLogPostPDF}.  

%
%
%

%%%%%%%%%%%%%%%%%%%%%%%%%%%%%%%%%%%%%%%%%%%%%%
\subsection{Proof of Theorem \ref{thm:holder}}\label{Sec:ProofHolder} We will use throughout notation and facts from the proof of Theorem \ref{first-der}, which can be found in the supplementary Section \ref{sec-gradient-pf}. Let $R>0$, and let $g,h\in C^{1+\kappa}(\Ocal)$ with $\|g\|,\|h\|_{C^{1+\kappa}}\le R$. For some regularisation parameter $\delta>0$ which will be chosen below in dependence on $\|g\|_{C^1}$, let $D\Phi_{f+g}^\delta$ and $D\Phi^\delta_{f}$ be `regularised Frechét derivatives' defined as in \eqref{Eq:RegularDeriv} in the supplement, and make the decomposition
\begin{equation*}
    \begin{split}
        \big\| D\Phi_{f+g}[h]&- D\Phi_f[h] \big\|_{L^\infty}\\
        &\le \big\|D\Phi_{f+g}[h]-D\Phi_{f+g}^\delta[h] \big\|_{L^\infty}\\
        &\quad + \big\|D\Phi_f[h]-D\Phi_f^\delta[h] \big\|_{L^\infty} +\big\|D\Phi_{f+g}^\delta [h]-D\Phi_f^\delta[h] \big\|_{L^\infty}\\
    &=:\|I\|_{L^\infty}+\|II\|_{L^\infty}+\|III\|_{L^\infty}.
    \end{split}
\end{equation*}
We estimate each of the terms separately. Our goal is to show that for some $\zeta>0$,
\[ 
    \|I\|_{L^\infty}+\|II\|_{L^\infty}+\|III\|_{L^\infty} = O(\|h\|_{C^1}\|g\|_{C^1}^\zeta).
\]
Arguing like in the estimates for term $IV$ in the proof of Theorem \ref{first-der}, there exists some constant $\eta>0$ (specifically, chosen as in \eqref{eta-choice})
such that 
\[ \|I\|_{L^\infty}+\|II\|_{L^\infty}= O(\|h\| \delta^\eta).\]

It remains to estimate term $III$. To this end, we define the function $w:[0,T]\to L^2(\mathcal O)$ by
\begin{equation*}
        w :=  (\partial_t-\mathcal L_{f+g})^{-1}\big[ \Lcal_h u^\delta_{f+g}\big] - (\partial_t-\mathcal L_{f})^{-1}\big[ \Lcal_h u^\delta_{f}\big],
\end{equation*}
where we have used the previous notation $\mathcal L_h[\cdot]= \nabla \cdot (h\nabla[\cdot])$ (despite $h$ potentially taking nonpositive values). By inspection of the definition of $D\Phi^\delta$, we have that $III= w(D)$. The function $w$ satisfies the PDE
\begin{equation*}
    \begin{split}
        (\partial_t-\mathcal L_{f}) w &= \Lcal_h u^\delta _{f + g} + (\mathcal L_{f + g} - \mathcal L_f) (\partial_{t} - \mathcal{L}_{f + g})^{-1}[\Lcal_h u^{\delta}_{f + g}] - \Lcal_h u^\delta_f\\
        &= \Lcal_h \big( u^\delta _{f + g}- u^\delta_f \big) + \Lcal_g (\partial_{t} - \mathcal{L}_{f + g})^{-1}[\Lcal_h u^{\delta}_{f + g}].
    \end{split}
\end{equation*}
Using Theorem \ref{thm:lunardi-main} in the Supplement and the Sobolev embedding $H^2(\Ocal)\subset C(\bar\Ocal)$ (holding since $d\le 3$), we obtain that for any $\alpha\in (0,1)$, 
%we then obtain \textcolor{red}{For $\alpha$ small enough}
\begin{align*}
    \|III\|_{L^\infty} &= \|w(D)\|_{L^\infty} \\
    &\lesssim \|w\|_{C^\alpha_{\alpha+\mu}((0,T];H^2_N(\Ocal))} \\
    &\lesssim \big\|\Lcal_h \big( u^\delta _{f + g}- u^\delta_f \big) + \Lcal_g (\partial_{t} - \mathcal{L}_{f + g})^{-1}[\Lcal_h u^{\delta}_{f + g}]\big\|_{C^\alpha_{\alpha+\mu}((0,T];L^2(\Ocal))}\\
    &\le \big\|\Lcal_h \big( u^\delta _{f + g}- u^\delta_f \big)\big\|_{C^\alpha_{\alpha+\mu}((0,T];L^2(\Ocal))} \\
    &\quad + \big\| \Lcal_g (\partial_{t} - \mathcal{L}_{f + g})^{-1}[\Lcal_h u^{\delta}_{f + g}]\big\|_{C^\alpha_{\alpha+\mu}((0,T];L^2(\Ocal))}
    =: III_a+III_b,
\end{align*}
where the spaces $C^\alpha_{\alpha+\mu}((0,T];L^2(\Ocal))$ and $C^\alpha_{\alpha+\mu}((0,T];H^2_N(\Ocal))$ are defined in (\ref{C-alphabeta}) in the supplement. Noting that the difference $u^\delta _{f + g}- u^\delta_f$ equals the term $R_0^\delta [g]$ defined in \eqref{R0-def}, we now choose $\alpha>0$ small enough -- like in the Lemmas \ref{R0-reg} and \ref{delta-reg-est} -- to obtain that
\begin{align*}
    \big\|\Lcal_h \big( u^\delta _{f + g}- u^\delta_f \big)\big\|_{C^\alpha_{\alpha+\mu}((0,T];L^2(\Ocal))}& \lesssim \|h\|_{C^1}\big\|u^\delta _{f + g}- u^\delta_f \big\|_{C^\alpha_{\alpha+\mu}((0,T];H^2_N(\Ocal))}\\
    &= \|h\|_{C^1}\big\|R_0^\delta[g] \big\|_{C^\alpha_{\alpha+\mu}((0,T];H^2_N(\Ocal))}\\
    &\lesssim \|h\|_{C^1}\|g\|_{C^1} \delta^{-\gamma},
\end{align*}
where $\gamma\equiv \gamma(\alpha)=\alpha(d/2+2)+(1-\alpha)d/4$. By our choice of $\alpha>0$ it is ensured that $\gamma(\alpha)<1$.

The term $III_b$ can be treated similarly. Using Theorem \ref{thm-lunardi}, 
\begin{align*}
    \big\| \Lcal_g (\partial_{t}& - \mathcal{L}_{f + g})^{-1}[\Lcal_h u^{\delta}_{f + g}]\big\|_{C^\alpha_{\alpha+\mu}((0,T];L^2(\Ocal))}\\
    &\lesssim \|g\|_{C^1} \big\|(\partial_{t} - \mathcal{L}_{f + g})^{-1}[\Lcal_h u^{\delta}_{f + g}]\big\|_{C^\alpha_{\alpha+\mu}((0,T];H^2_N(\Ocal))}\\
    &\lesssim \|g\|_{C^1} \big\|\Lcal_h u^{\delta}_{f + g}\big\|_{C^\alpha_{\alpha+\mu}((0,T];L^2(\Ocal))}.
\end{align*}
The last expression is almost identical to the right hand side of the PDE (\ref{R0-PDE}) satisfied by $R_0^\delta[h]$, except with $u_{f+g}$ in place of $u_{f+h}$. Following the exact same argument as in the proof of Lemma \ref{R0-reg}, it can be shown that $\big\|\Lcal_h u^{\delta}_{f + g}\big\|_{C^\alpha_{\alpha+\mu}((0,T];L^2(\Ocal))}=O( \|h\|_{C^1}\delta^{-\gamma})$. Overall, we have obtained that $III= O(\|h\|_{C^1}\|g\|_{C^1} \delta^{-\gamma})$.

Finally, choose $\delta:=\|g\|_{C^1}$; upon combining all the preceding estimates, this yields 
\[ 
    \|I\|_{L^\infty} + \|II\|_{L^\infty}+\|III\|_{L^\infty} \lesssim \|h\|_{C^1} (\|g\|_{C^1}^\eta + \|g\|_{C^1}^{1-\gamma}),
\]
which concludes the proof upon choosing $\zeta = \min \{\eta, 1-\gamma\}$.
\qed

%%%%%%%%%%%%%%%%%%%%%%%%%%%%%%%%%%%%%%%%%%%%%%
\subsection{Proof of Corollary \ref{Cor:SpectrGrad}}\label{Sec:ProofSpectrGrad}
For $f$ and $h$ satisfying the assumptions of Corollary \ref{Cor:SpectrGrad}, recall the expression of the Frechét derivative from Theorem \ref{first-der},
$$
    D\Phi_f [h] = \int_0^{D} P_{D-s,f}
    \big[ \nabla\cdot(h\nabla p_{s,f}(x,\cdot))\big](y) ds.
$$
By the spectral representations \eqref{Eq:SpectTrOp} and \eqref{Eq:FundSol} for the transition operators and density functions, we can then write the above integrand as, for any fixed $x,y\in\Ocal$ and $s\in(0,D]$, recalling that $e_{f,0}=1$ (independently of $f$),
\begin{equation}
\begin{split}
\label{Eq:Interm}
    P_{D-s,f}&\big[ \nabla\cdot(h\nabla p_{s,f}(x,\cdot))\big](y)\\
    &=\sum_{j=0}^\infty e^{-(D-s) \lambda_{f,j}}\langle \nabla\cdot(h\nabla p_{s,f}(x,\cdot)), e_{f,j}\rangle_{L^2} e_{f,j}(y)\\
%    &=\sum_{j\ge0}e^{-(D-s) \lambda_{f,j}}
%    \left[\sum_{j'\ge1}e^{-s \lambda_{f,j'}}e_{f,j'}(x)
%    \langle \nabla\cdot(h \nabla e_{f,j'}),e_{f,j}\rangle_{L^2}\right] e_{f,j}(y)\\
    &=\sum_{j=0}^\infty \sum_{j'=1}^\infty e^{-D \lambda_{f,j}} e^{s(\lambda_{f,j} - \lambda_{f,j'})}\langle \nabla\cdot(h \nabla e_{f,j'}),e_{f,j}\rangle_{L^2}e_{f,j'}(x)e_{f,j}(y).
\end{split}
\end{equation}
Now,
\begin{align*}
    \langle \nabla\cdot(h \nabla e_{f,j'}),e_{f,j}\rangle_{L^2}
    & =
    \langle h\Delta e_{f,j'},e_{f,j}\rangle_{L^2}
    +\langle \nabla h\cdot\nabla e_{f,j'},e_{f,j}\rangle_{L^2},
\end{align*}
and since by Green's first identity (e.g.~\cite[Section C.2]{E10}), recalling that $\partial_\nu e_{f,j'} = \nabla e_{f,j'}\cdot\nu = 0$,
\begin{align*}
    \langle h \Delta e_{f,j'},e_{f,j}\rangle_{L^2}
    &=\int_{\partial \mathcal O}he_{f,j}\nabla e_{f,j'}\cdot \nu d\sigma-\int_{\mathcal O} \nabla(he_{f,j})\cdot\nabla e_{f,j'} dx\\
    &=-\int_{\mathcal O} h\nabla e_{f,j}\cdot\nabla e_{f,j'} dx
    -\int_{\mathcal O} e_{f,j}\nabla h \cdot\nabla e_{f,j'} dx\\
    &= - \langle h , \nabla e_{f,j}\cdot\nabla e_{f,j'}\rangle_{L^2}
    -\langle \nabla h\cdot\nabla e_{f,j'}, e_{f,j}\rangle_{L^2},
\end{align*}
we have
\begin{align*}
    \langle \nabla\cdot(h \nabla e_{f,j'}),e_{f,j}\rangle_{L^2}
    & =
    - \langle h , \nabla e_{f,j}\cdot\nabla e_{f,j'}\rangle_{L^2}.
\end{align*}
Replaced into \eqref{Eq:Interm}, this gives
\begin{align*}
    P_{D-s,f}&\big[ \nabla\cdot(h\nabla p_{s,f}(x,\cdot))\big](y)\\
    &=
    -\sum_{j,j'=1}^\infty 
    e^{-D\lambda_{f,j}} e^{s (\lambda_{f,j} - \lambda_{f,j'})}\langle h , \nabla e_{f,j}\cdot\nabla e_{f,j'}\rangle_{L^2} e_{f,j'}(x)e_{f,j}(y),
\end{align*}
and finally
\begin{align*}
    D\Phi_f[h] 
    &=-
    \sum_{j,j'=1}^\infty   e^{-D\lambda_{f,j}} \left(\int_0^D e^{s(\lambda_{f,j} - \lambda_{f,j'})} ds\right) \langle h , \nabla e_{f,j}\cdot\nabla e_{f,j'}\rangle_{L^2} e_{f,j'}(x) e_{f,j}(y)\\
    &=
    \sum_{j,j'=1}^\infty C_{f,j,j'} \langle h , \nabla e_{f,j}\cdot\nabla e_{f,j'}\rangle_{L^2} e_{f,j'}(x)e_{f,j}(y),
\end{align*}
where 
$$
C_{f,j,j'}:=
    \begin{cases}
    -D e^{- \lambda_{f,j}D}, & \lambda_{f,j}=\lambda_{f,j'}\\
    (e^{-\lambda_{f,j}D}-e^{- \lambda_{f,j'}D})/(\lambda_{f,j} - \lambda_{f,j'}), & \textnormal{otherwise}.
    \end{cases}
$$
\qed

%
%
%

%%%%%%%%%%%%%%%%%%%%%%%%%%%%%%%%%%%%%%%%%%%%%%
\subsection{Proof of Proposition \ref{Prop:GradLogPostPDF}}
For fixed $x,y\in\Ocal$, denote for convenience
\begin{equation}
\label{Eq:AuxFun}
    l: \R^{K+1} \to \R, \qquad l(\theta)=\log p_{D,f_\theta} (x,y).
\end{equation}
The proof then clearly follows if we show that $l$ is continuously differentiable and that for $k=0,\dots,K$,
\begin{equation}
\label{Eq:Wish}
    \partial_{\theta_k}l(\theta)=
    \frac{ D\Phi_{f_\theta} \big[ (\phi' \circ F_\theta) e_k\big] }{p_{D,f_\theta}(x,y)},   
\end{equation}
where $D\Phi_{f_\theta}$ is the linear operator defined as in \eqref{Eq:FrechDeriv} with $f=f_\theta$. Recalling the notation $\Phi(f)\equiv \Phi_{D,x,y}(f)=p_{D,f}(x,y)$, $f\in\Fcal$, the map $l$ in \eqref{Eq:AuxFun} is seen to be the result of the composition
$$
    l=\log \circ \Phi\circ \Lambda\circ \Theta,
$$
where
$$
    \Lambda : C^2(\bar \Ocal)\to\Fcal, \qquad F_\theta\mapsto f_\theta=\phi\circ F_\theta;
    \qquad \Theta:\R^{K+1}\to C^2(\bar \Ocal),
    \qquad \theta\mapsto F_\theta.
$$

The linear function $\Theta$ is smooth (in the Frechét sense), with $D\Theta_\theta=\Theta$ for all $\theta\in\R^{K+1}$. Further, it is easy to see that in view of the regularity of $\phi$, the function $\Lambda$ also is smooth, with Frechét derivative given by the linear operator
$$
    D\Lambda: C^2(\bar\Ocal)\to C^2(\bar\Ocal),
    \qquad D\Lambda_F[h]=(\phi'\circ F)h.
$$
Theorem \ref{first-der} and Theorem \ref{thm:holder} together imply that $\Phi$ is continuously differentiable. The differentiability of $l$ is then obtained the chain rule, e.g.~\cite[Section E.4]{E10}. In particular, with $\xi_k$ the $k^{\textnormal{th}}$ element of the standard basis of $\R^{K+1}$,
\begin{align*}
    \partial_{\theta_k}l(\theta)
    &=Dl_\theta[ \xi_k]\\
    &= \Bigg( D\log_{[\Phi\circ \Lambda\circ \Theta] (\theta)} \circ D\Phi_{[\Lambda\circ \Theta] (\theta)} \circ D\Lambda_{\Theta(\theta)}\circ D\Theta_{\theta} \Bigg) [ \xi_k]\\
    &=\frac{ \big(D\Phi_{[\Lambda\circ \Theta] (\theta)} \circ D\Lambda_{\Theta(\theta)}\circ \Theta\big) [ \xi_k]}{[\Phi\circ \Lambda\circ \Theta] (\theta)}\\
    &=\frac{\big( D\Phi_{f_{\boldsymbol\theta}} \circ D\Lambda_{F_{\boldsymbol\theta}}\big) [e_k] }{p_{D,f_\theta}(x,y)}
    =\frac{ D\Phi_{f_{\boldsymbol\theta}} \big[(\phi'\circ F_\theta)e_k \big] }{p_{D,f_\theta}(x,y)}.
\end{align*}
This concludes the derivation of \eqref{Eq:Wish} and the proof of the proposition.

%%%%%%%%%%%%%%%%%%%%%%%%%%%%%%%%%%%%%%%%%%%%%%
%\section{Outline of the proof of Theorem \ref{first-der}}\label{Sec:ProofOutline}

%%%%%%%%%%%%%%%%%%%%%%%%%%%%%%%%%%%%%%%%%%%%%%
%% Support information, if any,             %%
%% should be provided in the                %%
%% Acknowledgements section.                %%
%%%%%%%%%%%%%%%%%%%%%%%%%%%%%%%%%%%%%%%%%%%%%%
\begin{acks}[Acknowledgments]
The authors would like to thank Richard Nickl and Markus Reiß for several helpful discussions, as well as the Associate Editor and three anonymous Referees for many helpful comments.
\end{acks}

%
%
%
%
%

%%%%%%%%%%%%%%%%%%%%%%%%%%%%%%%%%%%%%%%%%%%%%%
%% Funding information, if any,             %%
%% should be provided in the                %%
%% funding section.                         %%
%%%%%%%%%%%%%%%%%%%%%%%%%%%%%%%%%%%%%%%%%%%%%%
\begin{funding}
M.~G.~gratefully acknowledges the support of MUR - Prin 2022 - Grant no. 2022CLTYP4, funded by the European Union – Next Generation EU, and the affiliation to the "de Castro" Statistics Initiative, Collegio Carlo Alberto, Torino. S.~W. gratefully acknowledges the support of the Air Force Office of
Scientific Research Multidisciplinary University Research Initiative (MURI) project ANSRE.

%
%The second author was supported in part by NIH Grant ???????????.
\end{funding}

%% if your bibliography is in bibtex format, uncomment commands:
\bibliographystyle{imsart-number} % Style BST file (imsart-number.bst or imsart-nameyear.bst)
\bibliography{RefLFDiff}       % Bibliography file (usually '*.bib')

\newpage 
%%%%%%%%%%%%%%%%%%%%%%%%%%%%%%%%%%%%%%%%%%%%%%
\appendix
\section{Proof of Theorem \ref{first-der}}\label{sec-gradient-pf}

%Given some fixed constant $\fmin >0$, we introduce the following set of sufficiently regular and positive coefficient functions
%\begin{equation}\label{mathcal-F}
%    \mathcal F = \big\{ f \in C^2(\bar{\mathcal O}): \inf_{x\in\mathcal O} f(x) > \fmin  \big\}.
%\end{equation}
%In this section, we prove the characterisation for the gradient of the transition densities in Theorem \ref{first-der}. Recall the notation $\mathcal L_f[\cdot] = \nabla \cdot (f\nabla [\cdot])$ for the infinitesimal generator of $(X_t, \ t\ge0)$, with domain
%\[ \mathcal D(\mathcal L_f) = \Big\{ g\in H^2(\mathcal O): \partial_\nu g= 0~\textnormal{on}~\partial \mathcal O  \Big\}.\]
%Moreover, let $\varphi_\delta(x,y)$ and $(P_{t,0}, \ t\ge0)$ respectively denote the transition densities and the transition semigroup of the reflected Brownian motion on $\mathcal O$,
%\begin{equation}\label{refl-BM}
%    dY_t= \sqrt 2 dW_t+ \nu(Y_t)dL_t.
%\end{equation}

Given some fixed constant $\fmin >0$, recall the definition of the parameter space $\Fcal$ in \eqref{Eq:ParamSpace}, and the notation $\mathcal L_f[\cdot] = \nabla \cdot (f\nabla [\cdot])$ for the infinitesimal generator of $(X_t, \ t\ge0)$, with domain
\[ 
    \mathcal D(\mathcal L_f) = H^2_N(\Ocal) = \Big\{ u\in H^2(\mathcal O): \partial_\nu u= 0~\textnormal{on}~\partial \mathcal O  \Big\},
\]
equipped with the $H^2(\Ocal)$-norm. Moreover, let $\varphi_t(x,y),\ x,y\in\Ocal$, and $(P_{t,0}, \ t\ge0)$ respectively denote the transition densities and the transition semigroup of the reflected Brownian motion on $\mathcal O$,
\begin{equation}\label{refl-BM}
    dY_t= \sqrt 2 dW_t+ \nu(Y_t)dL_t, \qquad t\ge0.
\end{equation}

%
%
%

%%%%%%%%%%%%%%%%%%%%%%%%%%%%%%%%%%%%%%%%%%%%%%
\subsection{Proof strategy} 

In order to prove the gradient characterisation stated in Theorem \ref{first-der}, we first introduce a sequence of regularised transition densities $p_{t,f}^\delta$ that are shown to satisfy certain parabolic PDEs whose initial conditions become singular as $\delta\to 0$ (Section \ref{sec:reg-def}). For each fixed $\delta>0$, we then use the regularity theory for parabolic PDEs (reviewed in Section \ref{sec:parabol}) to estimate the differences between the regularised transition densities (Section \ref{sec:diff-est}). Using a recursive argument, higher order differences can also be controlled (Section \ref{sec:higher-order}). Finally, the result is obtained by employing a careful limiting argument to let the regularisation parameter $\delta\to 0$ (Section \ref{sec:delta-to-0}).

%
%
%

%%%%%%%%%%%%%%%%%%%%%%%%%%%%%%%%%%%%%%%%%%%%%%
\subsection{Regularised transition densities}\label{sec:reg-def}
For any conductivity function $f\in \mathcal F$, we define a regularised version of the transition density $p_{t,f}$ by
\[ 
    p^\delta_{t,f}(x,y)
    :=P_{t,f}[\varphi_\delta(\cdot,y)] (x) = P_{\delta,0}[p_{t,f}(x,\cdot)] (y),
    \qquad t,\delta>0, \qquad x,y\in\Ocal.  
\]
It will be helpful to regard these as functions from $[0,T]$ to $L^2(\mathcal O)$ for some $T>0$, where $y$ is fixed and the space variable is $x$. To this end, we introduce the notation
\[ 
    u_{f}^\delta(t)
    := p_{t,f}^\delta(\cdot,y),
    \qquad  u_f^\delta: [0,T]\to L^2(\mathcal O).
\]
By standard parabolic PDE theory (see \cite[Proposition 4.1.2]{lunardi}), it is clear that $u_{f}^\delta$ uniquely solves the initial value problem
\begin{equation}\label{delta-parabolic}
\begin{cases}
    (\partial_t-\mathcal L_f)u_f^\delta(t) = 0, &\textnormal{for} \ t>0,\\
    u_f^\delta(0)=\varphi_\delta(\cdot,y).
\end{cases}    
\end{equation}

Now suppose that $h$ is a `small perturbation' such that $f+h\in \mathcal F$. Then, using (\ref{delta-parabolic}) and the corresponding PDE with $\mathcal L_{f+h}$ in place of $\mathcal L_f$, we see that the difference $u_{f+h}^\delta -u_f^\delta$ constitutes the (unique) solution to
\begin{equation}\label{R0-PDE}
      \begin{cases}
    (\partial_t-\mathcal L_{f})w(t)   = \nabla \cdot(h\nabla u_{f+h}^\delta(t)), &\textnormal{for}~t>0,\\
    w(0)=0.
\end{cases}
\end{equation}

%
%
%

%%%%%%%%%%%%%%%%%%%%%%%%%%%%%%%%%%%%%%%%%%%%%%
\subsection{A key parabolic regularity result}\label{sec:parabol}

We will crucially use the PDE characterisation (\ref{R0-PDE}) to derive a norm bound for $u_{f+h}^\delta -u_f^\delta$ (in a suitable function space). To do so, we will make extensive use of the optimal regularity theory of parabolic PDEs, e.g.~\cite{lunardi}. For the convenience of the reader, we shall summarise some key results below -- details are left to Appendix \ref{app:parabolic}.

It is well-known that $\mathcal L_f: H^2_N(\mathcal O)\to L^2(\mathcal O)$ is a `sectorial operator' in the sense that its `resolvent set' contains a large enough sector in the complex plane $\mathbb C$, and that the resolvents satisfy a suitable norm estimate; see Appendix \ref{app:parabolic} for the precise definition. Consequently, the transition semigroup $e^{t\mathcal L_f}$ can be interpreted using the holomorphic functional calculus. For any bounded and continuous function $g:[0,T]\to L^2(\mathcal O)$ and any initial condition $u_0\in H^2_N(\mathcal O)$, consider the initial value problem
\begin{equation}\label{parabol2}
\begin{cases}
        (\partial_t - \Lcal_f)u(t)= g(t), &  t\in (0,T),\\
        u(0)=u_0.
\end{cases}
\end{equation}
Then, by standard theory, laid out e.g.~in Chapter 4.1.~of \cite{lunardi}, the unique solution to (\ref{parabol2}) is given by the variation-of-constants formula\footnote{This should be interpreted as a Bochner integral of $L^2(\mathcal O)$-valued functions. We shall not be concerned with distinguishing different notions of solution to (\ref{parabol2}), since all solutions considered here will constitute `strict' solutions. This is the strongest notion of those considered in the abstract parabolic theory; see, e.g.~\cite{lunardi}.}
\begin{equation}\label{eq:parasol}
    u(t) = e^{t\mathcal L_f} u_0 +\int_0^te^{(t-s)\mathcal L_f} g(s)ds,\qquad t\in[0, T].
\end{equation}
Thus, the regularity theory for the solutions to (\ref{parabol2}) reduces to the study of the above representation formula. We will also use the notation $(\partial_t-\mathcal L_f)^{-1}$ to denote the solution operator mapping (suitably smooth and integrable) functions $g:[0,T]\to L^2(\mathcal O)$ to the solution of \eqref{eq:parasol} with $u_0=0$,
\begin{equation}\label{parabol-inverse}
    \big[(\partial_t-\mathcal L_f)^{-1}g  \big](t) = \int_0^te^{(t-s)\mathcal L_f} g(s)ds, \qquad t\in [0,T].
\end{equation}

For $\alpha\in (0,1)$, let $C^\alpha([0,T];L^2(\Ocal))$ denote the space of $\alpha$-H\"older continuous functions from $[0,T]$ to $L^2(\mathcal O)$. For $0<\alpha<1$, $\beta>0$ and $X\in \{L^2(\mathcal O),H^2_N(\mathcal O)\}$, we introduce the spaces
\begin{equation}\label{C-alphabeta}
\begin{split}
    C^\alpha_\beta((0,T];X)&:=
    \bigcap_{0<\eps\le T} C^\alpha([\eps,T];X) \cap  \Big\{ u : \sup_{0<t\le T} t^{\beta-\alpha} \|u(t)\|_{X}<\infty\Big\}\\
    &\qquad\cap 
\Big\{ u: \sup_{0<\eps<T} \eps^\beta \|u\|_{C^\alpha([\varepsilon,T],X)} <\infty \Big\},
\end{split}
\end{equation}
normed by 
\[ 
    \|u\|_{C^\alpha_\beta((0,T];X)} := \sup_{0<t\le T} t^{\beta-\alpha} \|u(t)\|_{X} +  \sup_{0<\eps<T} \eps^\beta \|u\|_{C^\alpha([\varepsilon,T];X)}. 
\]
Note that if $\beta>\alpha$ the above norm allows for $\|u(t)\|_X$ to blow up at polynomial rate when $t\to 0$.

The following `optimal regularity' estimate for the solution $u$ given by the variation-of-constants formula (\ref{eq:parasol}) is a version of Theorem 4.3.7 in \cite{lunardi}, and is used repeatedly throughout our proofs. A more general version can be found in Theorem \ref{thm-lunardi} below. We shall only need the case $u_0=0$.

\begin{theorem}\label{thm:lunardi-main}
For some $\alpha,\mu\in (0,1)$, assume that $g\in C^\alpha_{\alpha+\mu} ((0,T];L^2(\mathcal O))$, and let $u:[0,T]\to L^2(\mathcal O)$ be given by (\ref{eq:parasol}). Then, $u\in C^\alpha_{\alpha+\mu}((0,T]; H^2_N(\mathcal O))$ and $\partial_t u\in C^\alpha_{\alpha+\mu}((0,T]; L^2(\mathcal O))$. Moreover, there exists a constant $C>0$ (independent of $g$) such that
\[ \|u\|_{C^\alpha_{\alpha+\mu}((0,T];H^2_N(\mathcal O))} + \|\partial_t u\|_{C^\alpha_{\alpha+\mu}((0,T];L^2(\mathcal O))} \le C \|g\|_{C^\alpha_{\alpha+\mu}((0,T];L^2(\mathcal O))}. \]
\end{theorem}

Informally speaking, the theorem asserts that the solution (\ref{eq:parasol}) satisfies two types of `parabolic regularity'. Firstly, the regularity \textit{in space} of $u$ is of order $H^2_N(\Ocal)$ whenever the spatial smoothness of $g$ is of $L^2(\Ocal)$ type. Secondly, $u$ also possesses regularity \textit{in time}, in that the smoothness of $\partial_t u$ matches that of $g$.

%
%
%

%%%%%%%%%%%%%%%%%%%%%%%%%%%%%%%%%%%%%%%%%%%%%%
\subsection{Estimates for differences of regularised transition densities}\label{sec:diff-est}
We are now ready to use the PDE (\ref{R0-PDE}) to derive bounds for differences between regularised transition densities.

\begin{lemma}\label{R0-reg} Suppose that $d\le 3$ and that $f,h:\mathcal O\to \R$ are such that $f, f+h\in \mathcal F$, where $\mathcal F$ is given by (\ref{Eq:ParamSpace}). Let
\[
    w_{f,h}
    :=u_{f+h}^\delta-u_{f}^\delta,
    \qquad w_{f,h}: [0,T]\to L^2(\mathcal O). 
\]
Then, there exist some constants $\alpha\in (0,1)$ (sufficiently small), $\mu\in (0,1)$ (sufficiently large), $\gamma\in (0,1)$ as well as $C>0$ only depending on $\fmin ,\|f\|_{C^1}, \|h\|_{C^1}, \mu$ and $\gamma$ such that 
\[ 
    \|w_{f,h}\|_{C^\alpha_{\alpha+\mu}((0,T];H^2_N(\mathcal O))}
    +  \big\|\partial_t w_{f,h} \big\|_{C^{\alpha}_{\alpha+\mu}((0,T];L^2(\mathcal O))}\le C \|h\|_{C^1}\delta^{-\gamma}. \]
\end{lemma}

The preceding bound is the main technical result of this section. It will be the key for us to employ a bootstrap argument to also control `higher order approximations' of the transition densities; see Section \ref{sec:higher-order} below.

\begin{proof}
Since $w_{f,h}$ satisfies the parabolic PDE \eqref{R0-PDE}, in light of Theorem \ref{thm:lunardi-main}, it suffices to derive a regularity estimate for the inhomogeneity term in \eqref{R0-PDE}, which we shall denote by
\[ 
    g(t)\equiv g_{f,h}^\delta(t):=\nabla \cdot (h \nabla u_{f+h}^\delta(t)\big),\qquad g: [0,T]\to L^2(\mathcal O). 
\]
Our goal is to bound the $C^\alpha_{\alpha+\mu}((0,T]; L^2(\mathcal O))$-norm of $g$ for $\alpha$ and $\mu$ to be suitably chosen later in the proof. This is achieved in four steps; Steps 1-3 deal with estimating 
\[
    \sup_{\eps\in(0,T]} \eps^{\alpha+\mu}\|g\|_{C^{\alpha}([\eps,T]; L^2(\mathcal O))},
\]
while Step 4 deals with bounding \[\sup_{t\in (0,T]}t^\mu\|g(t)\|_{L^2}.\]

\textbf{Step 1.}
Fix any $\eps \in (0,T)$, as well as any $\eps<t'<t\le T$. Then, we have that
\begin{align*}
    \|g(t)-g(t')\|_{L^2} &= \Big\| \nabla \cdot (h\nabla (u_{f+h}^\delta(t) -u_{f+h}^\delta(t')  )) \Big\|_{L^2}\\
    & \lesssim \|h\|_{C^1} \big\| u_{f+h}^\delta(t) -u_{f+h}^\delta(t') \big\|_{H^2}.
\end{align*}
In order to further bound the right hand side, we use the norm equivalence
\begin{equation}\label{normequiv}
   C^{-1} \big( \|\mathcal L_{f+h} u\|_{L^2}+ \| u\|_{L^2}\big)\le  \|u\|_{H^2}\le C \big( \|\mathcal L_{f+h} u\|_{L^2}+ \| u\|_{L^2}\big)
\end{equation}
from Lemma \ref{norm-equiv}, holding for all $u\in H^2_N(\Ocal)$ and for some $C$ that only depends on $\fmin $ and $\|f+h\|_{C^1}$. It follows that
\begin{align}
   \big\| u_{f+h}^\delta(t) &-u_{f+h}^\delta(t') \big\|_{H^2} \notag\\
   &\simeq \big\|\mathcal L_{f+h} \big[ 
    u_{f+h}^\delta(t) -u_{f+h}^\delta(t') \big]\big\|_{L^2} 
    + \big\|
    u_{f+h}^\delta(t) - \label{decomp} u_{f+h}^\delta(t')\big\|_{L^2} \\
    &=: I+II.\notag
\end{align}

\textbf{Step 2: Term I.}
Term $I$ is the most difficult to bound among the two, and we treat it first. Denoting by $Id$ the identity operator, and using the semigroup property of $(P_{t,f+h}, \ t\ge0)$, the definition of $u_{f+h}^\delta$ as well as the fact that $\mathcal L_{f+h}$ and $P_{t,f+h}$ commute, we have that
\begin{align*}
    \big\|\mathcal L_{f+h} \big[ 
    u_{f+h}^\delta(t) -u_{f+h}^\delta(t') \big]\big\|_{L^2}
    &=\big\| \mathcal L_{f+h} (P_{t-t',f+h} - Id ) \big[P_{t',f+h} [\varphi_\delta (\cdot, y)]\big] \big\|_{L^2}\\
    &= \big\| (P_{t-t',f+h} - Id ) \big[\mathcal L_{f+h} P_{t',f+h} [\varphi_\delta (\cdot, y)]\big] \big\|_{L^2}.
\end{align*}
Moreover, standard analytic semigroup theory implies that for some $M<\infty$ it holds that
\begin{equation}\label{M-bounds}
    \begin{split}
       \sup_{0<t\le T} t\|\mathcal L_{f+h}P_{t,f+h}\|_{L^2\to L^2} \le  M;
       \qquad \sup_{0<t\le T} t\|\mathcal L_{f+h}P_{t,f+h}\|_{H^2_N\to H^2_N} &\le  M,
    \end{split}
\end{equation}
see e.g.~the estimates (2.1.1) in \cite{lunardi}. Above, for $X\in\{L^2(\Ocal),H^2_N(\Ocal)\}$, we have denoted by $\|\cdot\|_{X\to X}$ the usual operator norm. Next, for $\alpha\in(0,1)$ we denote by $\mathcal D_\alpha$ a suitable interpolation space between $L^2(\mathcal O)$ and $H^2_N(\mathcal O)$ defined in (\ref{D-alpha}) below. This space has norm
\[ 
    \|u\|_{\mathcal D_\alpha}:=\|u\|_{L^2}+\sup_{t\in (0,1)}t^{-\alpha}\|P_{t,f}u-u\|_{L^2};
\]
see (\ref{Dalpha-decay}). Using the interpolation inequality 
\[
    \|u\|_{\mathcal D_\alpha}\lesssim \|u\|_{L^2}^{1-\alpha}\|u\|_{H^2}^\alpha,\qquad u\in \mathcal D_\alpha,
\] 
as well (\ref{M-bounds}), we may therefore estimate
\begin{align*}
    \big\| (P_{t-t',f+h} & - Id ) \big[\mathcal L_{f+h} P_{t',f+h} [\varphi_\delta (\cdot, y)]\big] \big\|_{L^2} \\
    &\lesssim (t-t')^\alpha \big\| \mathcal L_{f+h} P_{t',f+h} [\varphi_\delta (\cdot, y)] \big\|_{\mathcal D_\alpha}\\
    &\lesssim (t-t')^\alpha \big\| \mathcal L_{f+h} P_{t',f+h} [\varphi_\delta (\cdot, y)] \big\|_{L^2}^{1-\alpha}\big\| \mathcal L_{f+h} P_{t',f+h} [\varphi_\delta (\cdot, y)] \big\|_{H^2}^{\alpha}\\
    &\lesssim (t-t')^\alpha (t')^{-1} \big\| \varphi_\delta (\cdot, y)\big\|_{L^2}^{1-\alpha} \big\| \varphi_\delta (\cdot, y)\big\|_{H^2}^{\alpha}\\
    &\lesssim (t-t')^\alpha \eps^{-1} \big\| \varphi_\delta (\cdot, y)\big\|_{L^2}^{1-\alpha} \big\| \varphi_\delta (\cdot, y)\big\|_{H^2}^{\alpha}.
\end{align*}
Next, we use the heat kernel estimates from Lemma \ref{lem-heatkernel} below to obtain the bound
\begin{equation}
\label{termI-estimate}
\begin{split}
    %\big\| (P_{t-t',f+h} &- Id ) \big[\mathcal L_{f+h}P_{t',f+h} %[\varphi_\delta (\cdot, y)]\big] \big\|_{L^2} \\
    %&\lesssim 
    (t-t')^\alpha \eps^{-1} &\big\| \varphi_\delta (\cdot, y)\big\|_{L^2}^{1-\alpha} \big\| \varphi_\delta (\cdot, y)\big\|_{H^2}^{\alpha}\\
    &\lesssim (t-t')^\alpha \eps^{-1}\delta^{-(1-\alpha)d/4}\delta^{-\alpha (d/2+2)}
    = (t-t')^\alpha \eps^{-1}\delta^{-\gamma}, 
\end{split}
\end{equation}
where we have set
\begin{equation}\label{gamma-alpha}
    \gamma\equiv\gamma(\alpha) := (1-\alpha)d/4 + \alpha (d/2+2)= d/4 + \alpha\big(2+d/4\big).
\end{equation}
Since $d\le 3$, we can (and will) choose a sufficiently small value for $\alpha>0$ to guarantee that $\gamma<1$, as required in the statement of the lemma.

\textbf{Step 3: Term II; combining the estimates.}
Term $II$ can be estimated similarly. Indeed, it holds that for some $M>0$,
\begin{equation*}
    \begin{split}
       \sup_{0<t\le T} \|\mathcal L_{f+h}P_{t,f+h}\|_{L^2\to L^2} \le  M;\qquad \sup_{0<t\le T} \|\mathcal L_{f+h}P_{t,f+h}\|_{H^2_N\to H^2_N} \le  M.
    \end{split}
\end{equation*}
Then, arguing as above, we obtain that
\begin{align}
    \big\| (P_{t-t',f+h} - Id ) \big[P_{t',f+h} [\varphi_\delta (\cdot, y)]\big] \big\|_{L^2} &\lesssim (t-t')^\alpha \big\| \varphi_\delta (\cdot, y)\big\|_{L^2}^{1-\alpha} \big\| \varphi_\delta (\cdot, y)\big\|_{H^2}^{\alpha}\notag\\
    &\lesssim (t-t')^\alpha \delta^{-(1-\alpha)d/4}\delta^{-\alpha (d/2+2)} \notag \\
    &= (t-t')^\alpha \delta^{-\gamma}. \label{termII-estimate}
\end{align}
Combining (\ref{decomp}), (\ref{termI-estimate}) and (\ref{termII-estimate}) then implies that $\|g\|_{C^\alpha([\eps,T]; L^2(\mathcal O))} \lesssim \|h\|_{C^1} \eps^{-1}\delta^{-\gamma}$. Now choose any $\mu \in (1-\alpha,1)$. For such $\mu$, we have proven that 
\[ \sup_{\eps\in(0,T]} \eps^{\alpha+\mu}\|g\|_{C^{\alpha}([\eps,T]; L^2(\mathcal O))}<\delta^{-\gamma}\|h\|_{C^1}. \]

\textbf{Step 4.}
There remains to bound $\sup_{0<t\le 1}t^\mu \|g(t)\|_{L^2}$. Arguing as before and using again the norm equivalence (\ref{normequiv}),
\begin{align*}
    \|g(t)\|_{L^2} &\lesssim \|h\|_{C^1}\|P_{t,f+h}[\varphi_\delta(\cdot,y)]\|_{H^2}\\
    &\simeq \|P_{t,f+h}\mathcal L_{f+h}[\varphi_\delta(\cdot,y)]\|_{L^2}+\|P_{t,f+h}[\varphi_\delta(\cdot,y)]\|_{L^2} .
\end{align*}
Just like in Steps I and II, we focus on the first term, as the second one is easier to control. By another application of the standard properties of analytic semigroups, it holds that for some $M<\infty$, 
\[
    \sup_{0<t\le 1}\|P_{t,f+h}\mathcal L_{f+h}\|_{L^2\to L^2}\le t^{-1} M;
    \qquad \sup_{0<t\le 1}\|P_{t,f+h}\mathcal L_{f+h}\|_{H^2\to L^2}\le M;
\]
cf.~eq.~(2.1.1) in \cite{lunardi}. Using interpolation theory, e.g.~Theorem 1.2.6 in \cite{lunardi}, and by possibly changing the constant $M$, we then obtain that
\[\sup_{0<t\le 1}\|P_{t,f+h}\mathcal L_{f+h}\|_{\mathcal D_\alpha\to L^2}\le t^{\alpha-1} M.\]
It follows that 
\[ \|P_{t,f+h}\mathcal L_{f+h}[\varphi_\delta(\cdot,y)]\|_{L^2}\lesssim t^{\alpha-1}\|\varphi_\delta(\cdot,y)\|_{L^2}^{1-\alpha}\|\varphi_\delta(\cdot,y)\|_{H^2}^{\alpha}. \]
We now may follow the same arguments as in (\ref{termI-estimate}) and (\ref{termII-estimate}) to derive the estimate,
\[ 
    \|g(t)\|_{L^2}\lesssim t^{\alpha-1}\delta^{-\gamma}. 
\]
Choosing $\alpha$ and $\gamma$ as in (\ref{gamma-alpha}), arguing as in (\ref{termI-estimate}) and using that $\mu>1-\alpha$, we then obtain that
\[ 
    \sup_{0<t\le T}t^{\mu} \|g(t)\|_{L^2}\le \sup_{0<t\le T}t^{\mu-(1-\alpha)} \|g(t)\|_{L^2}\lesssim T^{\mu-(1-\alpha)} \|h\|_{C^1} \delta^{-\gamma}. 
\]
This completes the proof of the lemma.
\end{proof}

%
%
%

%%%%%%%%%%%%%%%%%%%%%%%%%%%%%%%%%%%%%%%%%%%%%%
\subsection{Higher-order approximations}\label{sec:higher-order}

One can interpret the difference $w_{f,h}=u^\delta_{f+h}-u^\delta_{f}$ as the remainder term of a `zeroth order' (i.e.~constant) approximation of the operator $\tilde f\mapsto u^\delta_{\tilde f}$ at some fixed point $f\in\mathcal F$. In light of this interpretation, fixing some $f\in \mathcal F$, we introduce the notation 
\[
    M_0^\delta[h](t) :=  u_{f}^\delta(t), \qquad M_{0}^\delta[h]:[0,T]\to L^2(\mathcal O),
\]
for the zeroth order approximation, and define the remainder term by
\begin{align}\label{R0-def}
    R_0^\delta[h](t)&:= u_{f+h}^\delta(t)-u_{f}^\delta(t), \qquad R_0^\delta: [0,T]\to L^2(\mathcal O).
\end{align}
While this notation may at first seem artificial, it will prove very convenient for generalisations to higher order local polynomial approximations in directions $h$ around some $f\in\Fcal$, which we now define recursively. The following construction is adapted from \cite{W19}.

For any $k\ge 1$ and $f,h$ as above, and recalling the solution operator (\ref{parabol-inverse}), we define the $k^{\textnormal{th}}$ order `monomial' approximation term as 
\begin{align*}
    M_{k}^\delta[h](t) := (\partial_t-\mathcal L_f)^{-1} \big[ \nabla \cdot \big(h \nabla M_{k-1}^\delta[h] \big) \big] (t).
\end{align*}
Naturally, since we are mostly interested in the characterisation of the first-order derivative, the most important term will be the linear one, $M_1^\delta[h]$; see Corollary \ref{cor-regularised}. %For those readers concerned about whether this expression is well-defined, for now we can treat these definitions as formal expressions -- in the next lemma, it will be recursively shown that the terms $M_{k-1}^\delta[h]$ are suitably regular such that the action of the solution operator $(\partial_t-\mathcal L_f)^{-1}$ upon them is well-defined.
The above definitions should at first be understood to be formal expressions -- in the next lemma, it will be recursively shown that the terms $M_{k-1}^\delta[h]$ are suitably regular such that the action of the solution operator $(\partial_t-\mathcal L_f)^{-1}$ upon them is well-defined.

Before doing so, we note that it is clear from the recursive definition that, for all $k\ge 1$, $M_k^\delta[h]$ is homogeneous of degree $k$ in $h$. Indeed, the sum $\sum_{l=0}^k M_l^\delta [h]$ will serve as our `candidate' $k^{\textnormal{th}}$ order local polynomial approximation to the transition densities around $f$. We denote the associated remainder term by
\[ 
    R_k^\delta[h]:= u_{f+h}^\delta - \sum_{l=0}^k M_l^\delta[h]. 
\]
In order to show that the approximation by $\sum_{l=0}^k M_l^\delta [h]$ is the correct one, we will prove that the size of $R_k^\delta[h]$ is of order at most $o(\|h\|_{C^1}^k)$. Indeed, from the definition, it is easily seen that, for each $k\ge 1$, $R_k^\delta[h]$ also satisfies a recursive relationship, which reads
\begin{equation}\label{Rk-PDE}
\begin{cases}
  (\partial_t-\mathcal L_f) R_{k-1}^\delta [h](t)= [\nabla \cdot (h \nabla R_{k-1}^\delta[h])] (t), & \textnormal{for}\ t> 0,\\
  R_k^\delta[h](0)=0.
\end{cases}
\end{equation}
%Analogously to $R_0^\delta$, we define the time-integrated version of the remainder terms as
%\[ Q_k^\delta[h](t)=\int_0^t R_k^\delta[h](s) ds \]
%Then, we can easily observe that $R_0^\delta[h]$ and $Q_0^\delta[h]$ satisfy the following PDEs, respectively:
%\begin{equation*}
%    \begin{cases}
%        (\partial_t-\mathcal L_f)R_k^\delta[h](t) &= \nabla \cdot \big( hR_{k-1}^\delta [h]\big),~~~~ \text{for}~t>0,\\
%        R_k^\delta[h](0) &=0,
%    \end{cases}
%\end{equation*}
%as well as 
%\begin{equation}\label{Qk-PDE}
%    \begin{cases}
%        (\partial_t-\mathcal L_f)Q_k^\delta[h](t) &= \nabla \cdot \big( hQ_{k-1}^\delta [h]\big),~~~~ \text{for}~t>0,\\
%        Q_k^\delta[h](0) &=0.
%    \end{cases}
%\end{equation}
Using this characterisation, we will now extend Lemma \ref{R0-reg} (which established an upper bound for $R_0^\delta[h]$) to all $R_k^\delta[h],~k\ge 1$, by induction (again, we note that the case $k=1$ will be the most important for us, but proving the estimate for all $k\ge 1$ is no more complicated than the argument for $k=1$).

\begin{lemma}\label{delta-reg-est}
Suppose that $d\le 3$ and that $f,h:\mathcal O\to \R$ are such that $f,f+h\in\mathcal F$. Then, for any $k\ge 0$, there exists a sufficiently small constant $\alpha>0$, some $\gamma\in (0,1)$ and some $C>0$ such that 
\[ \|R_k^\delta[h]\|_{C^\alpha_{\alpha+\mu}((0,T];H^2_N(\mathcal O))}+  \|\partial_t R_k^\delta[h]\|_{C^{\alpha}_{\alpha+\mu}((0,T];L^2(\Ocal))}\le C \|h\|_{C^1}^{k+1}\delta^{-\gamma}. \]
%In particular, it also holds that
%\[ \|R_k^\delta[h]\|_{C^{\alpha}([0,T];L^2)}\le \|Q_k^\delta[h]\|_{C^{1+\alpha}([0,T];L^2)} \le  C \|h\|_{C^1}^{k+1}\delta^{-\gamma}. \]
\end{lemma}

The key implication of the preceding lemma is that it identifies the directional derivatives of the regularised transition densities. Using the lemma with $k=1$, we immediately obtain that the Frech\'et derivative of the regularised maps $f\in\mapsto u^\delta_f$ must be given by the operator $h\mapsto M_\delta^1 [h]$. In light of this, we introduce the notation
\begin{equation}
\label{Eq:RegularDeriv}
    D\Phi_f^\delta [h] := M_1^\delta[h](D). 
\end{equation}

\begin{cor}\label{cor-regularised}
Suppose that $d\le 3$. Then, with $\alpha,\mu,\gamma \in (0,1)$ as in Lemma \ref{delta-reg-est}, we have that
\begin{align*}
    \big\| u^\delta_{f+h} - u^\delta_f -  D\Phi_f^\delta [h]  \big\|_{C^\alpha_{\alpha+\mu}((0,T]; H^2_N(\mathcal O))} &= \big\|R_\delta^1[h] \big\|_{C^\alpha_{\alpha+\mu}((0,T]; H^2_N(\mathcal O))} 
    = O(\delta^{-\gamma}\|h\|^2_{C^1}),
\end{align*}
as $h\to 0$. In particular, using the Sobolev embedding $H^2(\Ocal)\subset C(\Ocal)$, it follows that for any fixed and positive $D>0$, 
    \[ \big\| u^\delta_{f+h}(D) - u^\delta_f(D) -  D\Phi_f^\delta [h]  \big\|_{L^\infty} = O(\delta^{-\gamma}\|h\|^2_{C^1}).\]
\end{cor}
In fact, following the arguments of \cite{W19}, one could also prove that the regularised transition densities are direction-wise analytic. However, this is not needed for the algorithmic purposes of this article, and we shall omit this generalisation in order to avoid additional technicalities.

We conclude this section with the proof of Lemma \ref{delta-reg-est}.

\begin{proof}[Proof of Lemma \ref{delta-reg-est}]  We proceed by induction on the order $k$. \textbf{Induction start $k=0$.} This is shown in Lemma \ref{R0-reg}.

\textbf{Induction step $(k-1)\to k$.}
Suppose now that $k\ge 1$, and that the claim holds for some $k-1$. In view of the PDE characterisation (\ref{Rk-PDE})  of $R_k^\delta[h]$ and by the regularity estimate from Theorem \ref{thm:lunardi-main}, our proof strategy is to derive a bound for $\|\nabla\cdot (h\nabla R_{k-1}^\delta[h])\|_{C^\alpha_{\alpha+\mu}((0,T];L^2(\mathcal O))}$. To this end, for each $t\in(0,T]$, we estimate
\begin{align*}
    t^\mu \|\nabla\cdot (h\nabla R_{k-1}^\delta[h])(t)\|_{L^2}&\lesssim t^\mu \|h\|_{C^1} \|R_{k-1}^\delta[h](t)\|_{H^2}\\
    &\lesssim \|h\|_{C^1} \|R_{k-1}^\delta[h]\|_{C^\alpha_{\alpha+\mu}((0,T];H^2_N(\Ocal))}.
    \end{align*}
Similarly, it holds that for any $t>t'\ge \eps>0$,
\begin{align*}
%    \eps^{\alpha+\mu} \|\nabla\cdot (h\nabla [R_{k-1}^\delta[h](t)&-R_{k-1}^\delta[h](t')])\|_{L^2}\\
    \eps^{\alpha+\mu} \|\nabla\cdot (h\nabla [R_{k-1}^\delta[h])(t)&-\nabla\cdot (h\nabla R_{k-1}^\delta[h])(t')\|_{L^2}\\
    &\le \eps^{\alpha+\mu}\|h\|_{C^1}\|R_{k-1}^\delta[h](t)-R_{k-1}^\delta[h](t')\|_{H^2}\\
    &\le \eps^{\alpha+\mu}\|h\|_{C^1} (t-t')^\alpha\|R_{k-1}^\delta[h](t)\|_{C^\alpha([\eps,T]; H^2_N(\mathcal O))}\\
     &\le \|h\|_{C^1} (t-t')^\alpha \|R_{k-1}^\delta[h]\|_{C^\alpha_{\alpha+\mu} ([0,T]; H^2_N(\mathcal O))}.
\end{align*}
Thus, using the above two displays together with the PDE (\ref{Rk-PDE}), the regularity estimate (\ref{thm-lunardi}) and the induction hypothesis, we obtain
\begin{align*}
    \|R_k^\delta[h]\|_{C^\alpha_{\alpha+\mu}((0,T];H^2_N(\mathcal O))} &+  \|\partial_t R_k^\delta[h]\|_{C^{\alpha}_{\alpha+\mu}((0,T];L^2(\Ocal))}\\
    &\le \big\|\nabla \cdot (h\nabla R_{k-1}^\delta[h]) \big\|_{C^\alpha_{\alpha+\mu}((0,T];L^2(\mathcal O))}\\
    &\lesssim \|h\|_{C^1} \big\| R_{k-1}^\delta[h]\big\|_{C^\alpha_{\alpha+\mu}((0,T];H^2_N(\mathcal O))} 
    \le C \|h\|_{C^1} \|h\|_{C^1}^k \delta^{-\gamma},
\end{align*}
    which proves the claim.
\end{proof}

%
%
%

%%%%%%%%%%%%%%%%%%%%%%%%%%%%%%%%%%%%%%%%%%%%%%
\subsection{Taking the limit $\delta\to 0$}\label{sec:delta-to-0}

Corollary \ref{cor-regularised} characterises the linearisation of the regularised transition densities for fixed $\delta>0$. To prove Theorem \ref{first-der}, we need to take the limit $\delta\to 0$. Since the regularity estimates obtained in the previous section depend in a specific manner on $\delta$, this requires some care; in particular, $\delta$ will be chosen according to the size of the perturbation $h$.

\begin{proof}[Proof of Theorem \ref{first-der}] Fix $x,y\in \mathcal O$, and $f$ and $h$ as in the hypotheses. Our goal to show that the (unregularised) transition densities satisfy
\[\Big
    | p_{D,f+h}(x,y) - p_{D,f}(x,y) - D\Phi_f[h](x) \Big| = o (\|h\|_{C^1}).
\]
To mirror the preceding notation, we write $u_f(t):= p_{t,f}(\cdot,y)$, $t\ge0$. Then, it suffices to show that 
\[\|u_{f+h}(D) - u_f (D) - D\Phi_f[h]\|_{L^\infty} = o(\|h\|_{C^1}).\]
We take the `intuitive' approach of approximating each of the above three terms with their $\delta$-regularised version, obtaining the decomposition
\begin{align*}
    \|u_{f+h}(D) - u_f (D) - D\Phi_f[h]\|_{L^\infty}
    &\le \|u_{f+h}^\delta(D) - u_{f+h}^\delta(D) -D\Phi^\delta_f [h] \|_{L^\infty}
    \\
    &\quad + \|u_{f+h}^\delta(D) - u_{f+h}(D)\|_{L^\infty} \\
    &\quad + \|u_{f}^\delta(D)- u_{f}(D)\|_{L^\infty}+ \|D\Phi^\delta_f [h]- D\Phi_f [h] \|_{L^2}\\
    &=: I+II+III+IV.
\end{align*}

\textbf{Choice of $\delta$.} We now make the crucial choice for the regularisation parameter $\delta$ in dependence of $h$. Let $\gamma \in (0,1)$ be the constant from Lemma \ref{delta-reg-est}. Then, we fix some $\beta\in (1,\gamma^{-1})$, and set $\delta:=\|h\|_{C^1}^{\beta}$. For this choice of $\delta$, we will be able to show that all four terms in the preceding decomposition are of order $o(\|h\|_{C^1})$.

\textbf{Term $I$.} The first term is exactly equal to $\|R_1^\delta[h](D)\|_{L^\infty}$. Thus, using Lemma \ref{delta-reg-est} as well as the Sobolev embedding $H^2(\mathcal O) \subset C(\mathcal O)$ (since $d\le 3$), we see that for some $\alpha>0$,
\begin{align*}
    I&=\|R_1^\delta[h](D)\|_{L^\infty}\\
    &\lesssim \|R_1^\delta[h](D)\|_{H^2}
    \lesssim  \|R_1^\delta[h]\|_{C^\alpha_{\alpha+\mu}((0,T]; H^2_N(\mathcal O))}
\lesssim \|h\|_{C^1}^2\delta^{-\gamma} = \|h\|_{C^1}^{2-\gamma\beta}= o(\|h\|_{C^1}).
\end{align*}
For the last identity, we used that $\beta\gamma<1$.

\textbf{Term $II$.} %To estimate the second term, we notice that $u_{f+h} (D) \in H^2_N$, i.e., $u_{f+h}$ belongs to the domain of the Neumann Laplace operator.
Let us denote the transition operators of the reflected Brownian motion by $(P_t,\ t>0)$. Since $(P_t,\ t>0)$ constitutes an analytic semigroup generated by the Laplace operator acting on $C(\mathcal O)$ with domain $C^2_N(\mathcal O):=\{u\in C^2(\Ocal):\partial_\nu u= 0\}$ (see Section \ref{Lf-sectorial} below for details), we can use the inequality (\ref{C2-decay}) together with the assumption that $\|h\|_{C^2}\le 1$ to obtain
\[  
    II= \| (P_\delta -Id ) u_{f+h}(D) \|_{C(\mathcal O)} \lesssim \delta \|u_{f+h}(D)\|_{C^2_N(\mathcal O)}\lesssim \delta \|u_{f+h}(D)\|_{C^{2+\eta}} = o(\|h\|_{C^1}).
\]
Here, we also used the fact that
\[ \sup_{h\in C^2: \|h\|_{ C^{1+\eta}} \le 1}\|u_{f+h}\|_{C^{2+\eta}} <\infty,\]
which follows from Lemma \ref{lem-heatkernel-uniform} below.

\textbf{Term $III$.} The third term can be treated in the same way as Term $II$, with $u_{f+h}$ replaced by $u_f$.

\textbf{Term $IV$.} We further decompose
\begin{align*}
    IV &= \|D\Phi_f^\delta[h] - D\Phi_f[h] \|_{L^\infty} \\
    &= \Big\| \int_0^D P_{D-s,f} \big[ \Lcal_h \big( p_{s,f}^\delta(x,\cdot) - p_{s,f}(x,\cdot) \big)\big]  ds \Big\|_{L^\infty}\\
    &=  \Big\| \int_0^{D/2} P_{D-s,f} \big[ \Lcal_h \big( p_{s,f}^\delta(x,\cdot) - p_{s,f}(x,\cdot) \big)\big]  ds \Big\|_{L^\infty}\\
    &\quad+\Big\| \int_{D/2}^D P_{D-s,f} \big[ \Lcal_h \big( p_{s,f}^\delta(x,\cdot) - p_{s,f}(x,\cdot) \big)\big]  ds \Big\|_{L^\infty}\\
    &=: IV_a+IV_b.
\end{align*}

\textbf{Term $IV_a$.} The main difficult is to deal with the singularity of $p_{s,f}$ for $s\to 0$. Using Lemma \ref{Pt-infinity-est}, we obtain that for some $C>0$ and any $t>D/2$, 
\begin{align*}
   \big\| P_{t,f} \big[ \Lcal_h \big( p_{s,f}^\delta(x,\cdot) - &p_{s,f}(x,\cdot) \big)\big]\big\|_{L^\infty} \le  \big\|\nabla \cdot \big(h\nabla \big( p_{s,f}^\delta(x,\cdot) - p_{s,f}(x,\cdot) \big)\big)\big\|_{(H^2_N)^*}\\
   &=\sup_{\varphi\in H^2_N:\|\varphi\|_{H^2}\le 1} \Big| \int_{\mathcal O} \big[ \nabla \cdot \big(h \nabla\big( p_{s,f}^\delta(x,\cdot) - p_{s,f}(x,\cdot) \big) \big] (z) \varphi(z) dz \Big|.
    %&=\sup_{\varphi\in H^2_N:\|\varphi\|_{H^2_N}\le 1} \Big|- \int_{\mathcal O} h (z) \langle \nabla \big( p_{s,f}^\delta(x,\cdot) - p_{s,f}(x,\cdot) \big) (z), \nabla \varphi(z) \rangle_{\R^d} dz \Big|\\
    \end{align*}
Using the self-adjointness of the differential operator $\nabla\cdot(h\nabla [\cdot])$, the dual characterisation of $\|\cdot\|_{L^2}$, (\ref{H2-decay}) as well as Lemma \ref{lem-heatkernel}, for any $\eta\in (0,1/2)$ we can further estimate the right hand side, up to a multiplicative constant, by 
\begin{align*}
   %\big\|h \nabla \big( p_{s,f}^\delta(x,\cdot) - p_{s,f}(x,\cdot) \big)\big\|_{H^{-1}}&\le \|h\|_{C^1} \big\| \nabla \big( p_{s,f}^\delta(x,\cdot) - p_{s,f}(x,\cdot) \big)\big\|_{H^{-1}}\\
    \sup_{\varphi\in H^2_N:\|\varphi\|_{H^2}\le 1} &\Bigg| \int_{\mathcal O} \big( p_{s,f}^\delta(x,z) - p_{s,f}(x,z) \big)  [\nabla \cdot \big(h \nabla \varphi \big)](z) dz \Bigg|\\
    &\lesssim \|h\|_{C^1} \big\|  p_{s,f}^\delta(x,\cdot) - p_{s,f}(x,\cdot) \big\|_{L^2}\\
    &= \|h\|_{C^1} \big\| (P_{\delta,0} - Id) p_{s,f}(x,\cdot) \big\|_{L^2} \\
    &\lesssim \|h\|_{C^1}  \delta^\eta \big\|p_{s,f}(x,\cdot) \big\|_{H^{2\eta}} \\
    &\lesssim \|h\|_{C^1}  \delta^\eta \big\|p_{s,f}(x,\cdot) \big\|_{L^2}^{1-\eta}\big\|p_{s,f}(x,\cdot) \big\|_{H^2}^\eta
    \lesssim \|h\|_{C^1}  \delta^\eta s^{-d(1-\eta)/4} s^{-\eta (d/2+2)}.
\end{align*}
Now let 
\begin{equation}\label{eta-choice}
    \gamma\equiv \gamma (\eta):= d(1-\eta) /4 + \eta ( d/2+2).
\end{equation}
Note that this corresponds to the exponent defined in \eqref{gamma-alpha}. By choosing $\eta>0$ small enough and using the fact that $d\le 3$, we can ensure that $\gamma<1$, so that the previous expression is integrable at $s=0$. As a consequence, for any such $\eta$, we have
\[ IV_a \lesssim \|h\|_{C^1}\delta^\eta = o(\|h\|_{C^1}). \]

\textbf{Term $IV_b$.} We now turn to term $IV_b$. Here, we fix some $y\in \mathcal O$. Then, using the self-adjointness of $\Lcal_h$ as well as of $P_{\delta,0}-Id$, we obtain
\begin{align*}
    \int_{D/2}^D P_{D-s,f} & \big[ \Lcal_h \big( p_{s,f}^\delta(x,\cdot) - p_{s,f}(x,\cdot) \big)\big] (y) ds \\
    &= \int_{D/2}^D \langle P_{D-s,f}(y,\cdot) , \Lcal_h \big( p_{s,f}^\delta(x,\cdot) - p_{s,f}(x,\cdot) \big)\rangle_{L^2} ds\\
    &=\int_{D/2}^D \langle (P_{\delta,0}-Id)\Lcal_h (P_{D-s,f}(y,\cdot)), p_{s,f}(x,\cdot) \rangle_{L^2} ds\\
    &= \int_{D/2}^D P_{s,f}\big[ (P_{\delta,0}-Id)\Lcal_h (P_{D-s,f}(y,\cdot)) \big](x) ds \\
     &= \int_0^{D/2} P_{D-s,f}\big[ (P_{\delta,0}-Id)\Lcal_h (p_{s,f}(y,\cdot)) \big] (x)ds\\
     &\le  \int_0^{D/2} \big\| P_{D-s,f}\big[ (P_{\delta,0}-Id)\Lcal_h (p_{s,f}(y,\cdot)) \big] \big\|_{L^\infty} ds.
\end{align*}
Now, we can employ a similar chain of estimates (albeit slightly more complicated) as we did for term $IV_a$. Another type of interpolation space between $L^2(\Ocal)$ and $H^2_N(\Ocal)$, different from the previously used ones, will be useful here. We follow Chapter 1, Section 2.1 of \cite{LM72}. We denote by $[H^2_N(\mathcal O), L^2(\mathcal O)]_\eta,~\eta\in [0,1]$, the interpolation spaces obtained through the construction there, defined as the domain of fractional powers $\Delta^\eta$ of the Neumann-Laplacian, satisfying $H^2_N(\mathcal O)\subseteq [H^2_N(\mathcal O), L^2(\mathcal O)]_\eta\subseteq L^2(\mathcal O)$. By the duality result is Theorem 6.2 of Chapter 1 in \cite{LM72} and identifying $(L^2(\Ocal))=(L^2(\Ocal))^*$, it holds that for all $\eta\in [0,1]$,
\begin{align}\label{interpol-dual}
    ([H^2_N(\mathcal O), L^2(\mathcal O)]_{1-\eta})^*=[L^2(\mathcal O), (H^2_N(\mathcal O))^*]_{\eta}.
\end{align}

For any $s\in [0,D]$, using Lemma \ref{Pt-infinity-est}, we can then estimate the integrand as follows, for any $\eta\in (0,1)$,
\begin{align*}
    \big\| P_{D-s,f}\big[ &(P_{\delta,0}-Id)\Lcal_h (p_{s,f}(y,\cdot)) \big] \big\|_{L^\infty}\\
    &\lesssim  \big\|(P_{\delta,0}-Id)\Lcal_h (p_{s,f}(y,\cdot)) \big\|_{(H^2_N)^*}\\
    &\lesssim \sup_{\varphi\in H^2_N(\Ocal):\|\varphi\|_{H^2}\le 1}\Big|\langle (P_{\delta,0}-Id)\Lcal_h (p_{s,f}(y,\cdot)) ,\varphi \rangle_{L^2}\Big| \\
    &\lesssim \sup_{\varphi\in H^2_N(\Ocal):\|\varphi\|_{H^2}\le 1}\Big|\langle p_{s,f}(y,\cdot) , \Lcal_h (P_{\delta,0}-Id)\varphi \rangle_{L^2}\Big|\\
    & \lesssim  \|p_{s,f}(y,\cdot)\|_{[H^2_N, L^2]_\eta} \sup_{\varphi\in H^2_N(\Ocal):\|\varphi\|_{H^2}\le 1}\| \Lcal_h (P_{\delta,0}-Id)\varphi \|_{([H^2_N, L^2]_\eta)^*}\\
    & \lesssim \|p_{s,f}(y,\cdot)\|_{H^{2}}^\eta \|p_{s,f}(y,\cdot)\|_{L^{2}}^{1-\eta} \sup_{\varphi\in H^2_N(\Ocal):\|\varphi\|_{H^2}\le 1}\| \Lcal_h (P_{\delta,0}-Id)\varphi \|_{[L^2,(H^2_N)^*]_{1-\eta}},
\end{align*}
the last step holding since $([L^2(\Ocal),H^2_N(\Ocal)]_\eta)^* = [L^2(\Ocal),(H^2_N(\Ocal))^*]_{1-\eta}$ (see Theorem 6.2 in \cite{LM72}). Then, using Lemma \ref{Lh-interpol} below, as well as the boundedness of $P_{\delta,0}-Id$ as an operator from $H^2_N(\Ocal)$ to itself, we see that
\begin{align*}
    \| \Lcal_h (P_{\delta,0}-Id)\varphi \|_{([L^2,H^2_N]_\eta)^*}&\lesssim \|h\|_{C^1} \| (P_{\delta,0}-Id)\varphi \|_{[L^2,H^2_N]_{1-\eta}}\\
    &\lesssim \|h\|_{C^1}\| (P_{\delta,0}-Id)\varphi \|_{L^2}^\eta \| (P_{\delta,0}-Id)\varphi \|_{H^2}^{1-\eta}\\
    &\lesssim \|h\|_{C^1} \delta^\eta \|\varphi \|_{H^2}.
\end{align*}
Now, we can choose $\eta>0$ small enough, just as after \eqref{eta-choice}, such that $\gamma(\eta)<1$. Combining the preceding estimates, we obtain that 
\begin{align*}
      \big\| P_{D-s,f}\big[ &(P_{\delta,0}-Id)\Lcal_h (p_{s,f}(y,\cdot)) \big] \big\|_{L^\infty} \lesssim \|h\|_{C^1}\delta^\eta s^{-\gamma},
\end{align*}
which is integrable on $s\in (0,D/2)$. Upon integration, we have $IV_b = O(\|h\|_{C^1}\delta^\eta)= o(\|h\|_{C^1})$, and the proof is complete.
\end{proof}

\begin{lemma}\label{Lh-interpol}
    Let $\eta\in [0,1]$. Then, there exists some $C>0$ such that for all $h\in C^1(\Ocal)$ and all $u \in H^2_N(\mathcal O)$,
\[ \|\Lcal_h u\|_{[L^2,(H^2_N)^*]_{1-\eta}}\le C\|h\|_{C^1} \|u\|_{[H^2_N, L^2]_{\eta}}. \]
\end{lemma}

\begin{proof}
Clearly, $\Lcal_h:H^2_N(\Ocal)\to L^2(\Ocal)$ is bounded and satisfies, for any $u\in H^2_N(\Ocal)$,
\[ 
    \|\Lcal_h u\|_{L^2}= \|\nabla \cdot (h\nabla u)\|_{L^2}\lesssim \|h\|_{C^1}\|u\|_{H^2}. 
\] 
Using a standard duality argument, we then obtain that for some $C>0$, for any $u\in H^2_N(\Ocal)$,
\begin{align*}
    \|\Lcal_hu\|_{(H^2_N)^*}
    &= \sup_{\varphi\in H^2_N(\Ocal):\|\varphi\|_{H^2}\le 1} \Big| \int_\Ocal \varphi (x)\Lcal_hu(x)dx \Big|\\   &=\sup_{\varphi\in H^2_N(\Ocal):\|\varphi\|_{H^2}\le 1}\Big| \int_\Ocal u(x) \Lcal_h\varphi(x)dx \Big|\\
    &\le C\sup_{\varphi\in L^2(\Ocal):\|\varphi\|_{L^2}\le 1}\Big| \int u(x) \varphi(x)dx \Big| =  C\|u\|_{L^2}.
\end{align*}
Following the argument of the Theorem 5.1 in \cite{LM72}, we then see that for all $u\in H^2_N(\mathcal O)$,
\[
    \|\Lcal_hu\|_{[L^2, (H^2_N)^*]_\eta} \lesssim \|u\|_{[H^2_N,L^2]_\eta}. 
\]
But at the same time, by (\ref{interpol-dual}), we have that $[L^2(\Ocal), (H^2_N(\Ocal))^*]_\eta=[H^2_N(\Ocal), L^2(\Ocal)]_{1-\eta}$, completing the proof.
\end{proof}

\section{Auxiliary technical results on parabolic equations and heat kernels}
\label{app:aux}

We review some basic definitions and facts about analytic semigroups, which constitute some crucial technical tools in the proofs of our main results.

%
%
%

%%%%%%%%%%%%%%%%%%%%%%%%%%%%%%%%%%%%%%%%%%%%%%
\subsection{Background on parabolic PDEs}\label{app:parabolic}

A key idea underpinning the analysis to follow is to interpret a collection of transition operators $(P_{t}, \ t\ge0)$ as an analytic semigroup generated by a so-called `sectorial' operator $\Lcal$. This allows to study the regularity properties of the solutions to abstract Cauchy problems of the form
\begin{equation}\label{parabol}
\begin{cases}
        \partial_t u(t) = \mathcal L u(t) + g(t),& t\in (0,T),\\
        u(0)=u_0.
\end{cases}
\end{equation}
Naturally, in our setting, the relevant operator is given by the infinitesimal generator $\Lcal_f$ of the diffusion process (\ref{Eq:SDE}).

%
%
%

%%%%%%%%%%%%%%%%%%%%%%%%%%%%%%%%%%%%%%%%%%%%%%
\subsubsection{Sectorial operators and analytic semigroups}
We follow the presentation in the monograph by Lunardi  \cite{lunardi}. Let $X$ be a complex Banach space, and suppose that $\mathcal L:\mathcal D(\mathcal L)\to X$ is a linear operator with domain $\mathcal D(\mathcal L)\subseteq X$, constituting a (not necessarily densely defined) subspace of $X$. Heuristically, $\mathcal L$ is called a \textit{sectorial operator} if its resolvent set $\rho(\mathcal L)\subseteq\mathbb C$ contains a `sector', and if the resolvents $R(\lambda,\mathcal L)$, $\lambda\in\rho(\Lcal)$, satisfy a certain decay. In particular, $\mathcal L$ is sectorial if there exist constants $\omega\in\R$, $\theta\in (\pi/2,\pi)$, and $M>0$ such that 
\begin{align*}
    \rho(\mathcal L)&\supset S_{\theta,\omega} := \{ \lambda \in \mathbb C: \lambda\neq \omega,~|\arg(\lambda-\omega)|<\theta \},
\end{align*}
and, denoting by $\|\cdot\|_{X\to X}$ the usual operator norm,
\begin{align}
    \| R(\lambda,\mathcal L)\|_{X\to X} &\le \frac{M}{|\lambda-\omega|},\qquad\forall \lambda \in S_{\theta,\omega}.\label{res-est}
\end{align}
It turns out that sectoriality is sufficient to that the operator exponentials $(e^{t\mathcal L}$,\ $t>0)$ are well-defined via functional calculus. They then form the analytic semigroup 
\[ 
    (P_t,\ t\ge 0),\qquad P_t:=e^{t\mathcal L},
\]
in the sense that the map $t\mapsto P_t$ is  analytic.

%

%%%%%%%%%%%%%%%%%%%%%%%%%%%%%%%%%%%%%%%%%%%%%%
\subsubsection{Characterisation of $\mathcal D(\mathcal L)$ and interpolation spaces}
The domain $\mathcal D(\mathcal L)$ of the sectorial operator $\Lcal$ is naturally equipped with the graph norm
\[ 
    \|u\|_{\mathcal D(\mathcal L)}:= \|u\|_X + \|\mathcal Lu\|_X,\qquad u\in \mathcal D(\mathcal L).
\]
Moreover, the behaviour of the map $t\mapsto e^{t\mathcal L}u$ at $t=0$ characterises the domain $\mathcal D(\mathcal L)$. Namely, it holds that $\lim_{t\to 0}(e^{t\mathcal L}u-u)/t$ exists (as a limit in $X$) and equals $\mathcal Lu$ if and only if $u\in \mathcal D(\mathcal L)$ and $\mathcal Lu\in \overline{\mathcal D(\mathcal L)}$. As an consequence, for every $u\in \mathcal D(\mathcal L)$ with $\Lcal u\in \overline{\mathcal D(\mathcal L)}$, an alternative characterisation of the graph norm of $u$ is
\begin{equation}\label{DL-decay}
    \|u\|_{\mathcal D(\mathcal L)}\simeq \|u\|_X + \sup_{0< t\le 1} t^{-1} \|e^{t\mathcal L}u-u\|_X.
\end{equation}
One can similarly define certain interpolation spaces $\mathcal D_\alpha$, for $\alpha\in (0,1)$, between $X$ and $\mathcal D(\mathcal L)$, which are of importance in the proofs.  Specifically, set
\begin{equation}\label{D-alpha}
    \mathcal D_\alpha:=\Big\{u\in X: \sup_{t\in (0,1]}t^{-\alpha} \|e^{t\mathcal L} u-u \|_X <\infty \Big\},
\end{equation}
which is equipped with norm
\begin{equation}\label{Dalpha-decay}
    \|u\|_{\mathcal D_\alpha} := \| u\|_{X} + \sup_{t\in (0,1] } t^{-\alpha} \| e^{t\mathcal L}u -u \|_{X}.
\end{equation}
We refer to Section 2.2.1 and Proposition 2.2.4 in \cite{lunardi} for further details on the interpretation of these spaces -- note that our spaces $\mathcal D_\alpha$ correspond to the spaces $D_\Lcal(\alpha)=(X,\mathcal D(\Lcal))_{\alpha}$ from \cite{lunardi}. The above interpolation can be defined for instance via the `K-method'; see the display (1.2.4) in \cite{lunardi}.

%

%%%%%%%%%%%%%%%%%%%%%%%%%%%%%%%%%%%%%%%%%%%%%%
\subsubsection{The case $\mathcal L=\mathcal L_f$}\label{Lf-sectorial}

For our purposes, the relevant instance is the transition semigroup $(P_{t,f}, \ t\ge0)$ of the reflected diffusion process (\ref{Eq:SDE}), with associated (elliptic) infinitesimal generator 
$\mathcal L_f[\cdot]= \nabla \cdot(f\nabla [\cdot])$ equipped with zero Neumann boundary conditions. Two different realisations of the operator $\mathcal L_f$ will play a role, firstly viewed as an operator on the Hilbert space $X=L^2(\mathcal O)$ with domain 
\[ 
    H^2_N(\bar{\mathcal O})=\big\{ u \in H^2(\mathcal O): \partial_\nu u =0~\text{on}~\partial\mathcal O \big\},
\]
and secondly as an operator on the Banach space $X=C(\bar{\mathcal O})$ with domain $C^2_N(\bar{\mathcal O})$ defined similarly to $H^2_N(\Ocal)$. It is well-known that $\mathcal L_f$ forms a sectorial operator on both spaces. Indeed, on $L^2(\mathcal O)$, this follows from Theorems 3.1.2 and 3.1.3 in \cite{lunardi}, where the main tool for the verification of the resolvent bounds (\ref{res-est}) are the a-priori estimates by Agmon-Douglis-Nirenberg \cite{ADN64}. On $C(\bar{\mathcal O})$, similar resolvent bounds can be verified; see Section 3.1.5, and especially Corollary 3.1.24 (ii), in \cite{lunardi}. We also note that the zero Neumann boundary conditions imposed here satisfy the non-tangentiality condition (3.1.3) in \cite{lunardi}, since
\[ 
    \inf_{x\in \partial \mathcal O} \Big|\sum_{i=1}^n\nu_i^2(x)\Big|
    =\inf_{x\in \partial \mathcal O}|\nu(x)|= 1>0. 
\]

The equations \eqref{DL-decay} and \eqref{Dalpha-decay} then directly yield useful characterisations of Sobolev- and H\"older-type norms in terms of the behaviour of $P_{t,f}u-u$, which we use throughout. Specifically, we note that for any $u\in H^2_N(\mathcal O)$, $f\in \mathcal F$ and $\alpha\in [0,1]$, it holds that
\begin{equation}\label{H2-decay}
  \|P_{t,f} u -u\|_{L^2} \lesssim t^\alpha \|u\|_{H^{2\alpha}},\qquad t\in (0,1).  
\end{equation}
Similarly, for $u\in C^2_N(\mathcal O)$, 
\begin{equation}\label{C2-decay}
    \|P_{t,f} u -u\|_{L^\infty} \lesssim t \|u\|_{C^2},\qquad t\in (0,1).    
\end{equation}

%

%%%%%%%%%%%%%%%%%%%%%%%%%%%%%%%%%%%%%%%%%%%%%%
\subsubsection{Regularity estimates for parabolic PDEs}
We now provide an overview on some key regularity results for the parabolic PDE (\ref{parabol}). A variety of notions of solutions to such equations exists in the literature. Since, via the regularisation argument developed in Section \ref{sec-gradient-pf}, we only need to consider a sequence of parabolic PDEs with regularised initial conditions, the strongest of these notions will suffice for our purposes: thus, when we speak of a solution to the Cauchy problem (\ref{parabol}), we mean a \textit{strict solution} in the sense that for all $t\in [0,T]$, $\partial_t u(t)=\mathcal Lu(t)+g(t)$ and $u(0)=u_0$.

When $g:[0,T]\to X$ is continuous, one then shows that any (strict) solution to (\ref{parabol}) must necessarily be given by the variation-of-constants formula
\[ 
    u(t) = e^{t\mathcal L} u_0(t) +\int_0^te^{(t-s)\mathcal L} g(s)ds,\qquad 0\le t\le T,
\]
where the integral on the right hand side is in the Bochner sense; see Proposition 4.1.2 in \cite{lunardi}. For our purposes, this property is always fulfilled.

Turning to the regularity of the solutions to the Cauchy problem (\ref{parabol}), the following (scales of) function spaces are needed: firstly, for $\alpha\in (0,1)$, let $C^\alpha([0,T];X)$ be the set of $\alpha$-H\"older continuous functions from $[0,T]$ to $X$. Secondly, we introduce spaces that allow for certain singularities at $t=0$ -- this is central for the proofs in Section \ref{sec-gradient-pf}, as it enables to derive regularity estimates that are uniform with respect to the regularisation parameter introduced there. For $0<\alpha<1$ and $\beta>0$, set
\begin{equation*}
\begin{split}
    C^\alpha_\beta((0,T];X)&:=
    \bigcap_{0<\eps\le T} C^\alpha([\eps,T];X) \cap  \Big\{ u : \sup_{0<t\le T} t^{\beta-\alpha} \|u(t)\|_{X}<\infty\Big\}\\
    &\qquad\cap 
\Big\{ u: \sup_{0<\eps<T} \eps^\beta \|u\|_{C^\alpha([\varepsilon,T],X)} <\infty \Big\}.
\end{split}
\end{equation*}

The following result is deduced from Theorem 4.3.1 (iii) and Theorem 4.3.7 in \cite{lunardi}. Since we shall only deal with the case $u_0=0$, let
\begin{equation}
\label{Eq:HomogSol}
    u(t)= \int_0^te^{(t-s)\mathcal L} g(s)ds.
\end{equation}

\begin{theorem}\label{thm-lunardi}
(i) Suppose that $g\in C^\alpha([0,T];X)$ for some $\alpha\in (0,1)$, and let $g(0) \in \mathcal D(\mathcal L)$. Then, $u$ in \eqref{Eq:HomogSol} constitutes a strict solution to the Cauchy problem (\ref{parabol}). Furthermore, $u\in C^\alpha([0,T];D(\mathcal L))\cap C^{1+\alpha}([0,T];X)$.

\textbf{(ii)} Suppose additionally that $g\in C^\alpha_{\alpha+\mu}((0,T];X)$ for some $\alpha,\mu\in (0,1)$. Then, there exists $C>0$ (independent of $g$) such that
\[ 
    \|u\|_{C^\alpha_{\alpha+\mu}((0,T];\mathcal D(\mathcal L))}+\|\partial_t u\|_{C^\alpha_{\alpha+\mu}((0,T];X)} \le C\|g\|_{C^\alpha_{\alpha+\mu}((0,T];X)}. 
\]
\end{theorem}

\subsection{Estimates for elliptic second order differential operators}
The proofs in Section \ref{sec-gradient-pf} also require a number of estimates for the solutions of elliptic PDEs, both in $L^2$-type norms and in uniform norms. We begin providing the estimates for the former. The first result is a basic lemma concerning the equivalence of the graph norm of $\mathcal L_f$, for $f$ belonging to the parameter space $\mathcal F$ from (\ref{Eq:ParamSpace}), and the $H^2(\mathcal O)$-norm.

\begin{lemma}\label{norm-equiv}
Let $R>0$. There exists a constant $C>0$ only depending on $\fmin ,R$ and $\mathcal O$ such that for all $f\in\mathcal F$ with $\|f\|_{C^1}\le R$, and all $u\in H^2_N(\mathcal O)$,
\[ 
    C^{-1} \|u\|_{H^2}\le \|u \|_{L^2} + \| \mathcal L_f u\|_{L^2} \le C \|u\|_{H^2}. 
\]
\end{lemma}

\begin{proof}
To prove the second inequality, we estimate
\[ 
    \|\mathcal L_f u \|_{L^2}= \| f\Delta u + \nabla f \cdot \nabla u \|_{L^2} \lesssim \|f\|_{L^\infty} \|\Delta u\|_{L^2} +  \|f\|_{C^1} \| \nabla u \|_{L^2} \lesssim R \|u\|_{H^2}.
\]
Thus, it remains to show the first inequality. For that, it suffices to prove that $\|\Delta u\|_{L^2}\le C (\|u\|_{L^2}+\|\mathcal L_f u\|_{L^2})$ for some $C>0$ only depending on $\fmin ,R$ and $\mathcal O$. To this end, we use the definition of $\mathcal L_f$, the interpolation inequality 
\[\|\nabla u\|_{L^2}\lesssim \|u\|_{L^2}^{1/2}\|u\|_{H^2}^{1/2} \lesssim \|u\|_{L^2} + \|u\|_{L^2}^{1/2}\|\Delta u\|_{L^2}^{1/2},\]
as well as Cauchy's inequality with $\eps>0$, to obtain that for some universal constants $c,c'>0$, and for $\eps\in (0,1]$,
\begin{equation}
    \begin{split}
    \label{lapl-bound}
	\|\Delta u\|_{L^2}&= \| f^{-1} (\mathcal L_f u- \nabla f\cdot \nabla u) \|_{L^2}\\
	& \le \fmin ^{-1} \|\mathcal L_f u \|_{L^2} + \fmin ^{-1}\|\nabla f\cdot \nabla u\|_{L^2}\\
	&\le \fmin ^{-1} \|\mathcal L_f u \|_{L^2} +  c \fmin ^{-1}\|f\|_{C^1} \big( \|u\|_{L^2}+ \|u\|_{L^2}^{1/2}\|\Delta u\|_{L^2}^{1/2} \big)\\
	& \le  \fmin ^{-1} \|\mathcal L_f u \|_{L^2} +  c' \fmin ^{-1}\|f\|_{C^1}\Big(\frac{\|u\|_{L^2}}{ 2\eps} + \frac{\eps \|\Delta u\|_{L^2}}2 \Big).
    \end{split}
\end{equation}
Now, choosing $\eps>0$ small enough and subtracting the term containing $\Delta u$ on the right hand side, we obtain that for some $C>0$ only depending on $\fmin , R$ and $\mathcal O$, 
\[ 
    \frac 12\|\Delta u\|_{L^2} \le C \big( \|\mathcal L_f u \|_{L^2} + \|u \|_{L^2}\big),
\]
as desired. This concludes the proof of the lemma.
\end{proof}

%

%\textcolor{red}{[Use H-Z spaces only in the proof]}
Next, we study the analogous mapping properties of $\mathcal L_f$ with respect to the H\"older spaces $C^{2+\eta}(\Ocal)$ and $C^\eta(\Ocal)$ for $\eta\in (0,1)$. Define
\[ 
    C^{2+\eta}_N(\mathcal O)
    := \big\{u\in C^{2+\eta}(\Ocal): \partial_\nu u  =0 ~\text{on}~\partial \mathcal O\big\}. 
\]

\begin{lemma}\label{Holder} Suppose that $f\in \mathcal F$. For any $\eta\in (0,1)$, there exists $C>0$ only depending on $\|f\|_{C^{1+\eta}},\fmin$ and  $\mathcal O$ such that for all $u\in C^{2+\eta}_N(\Ocal)$,
\[ 
    C^{-1}\|u\|_{C^{2+\eta}} \le   \|u\|_{C^\eta} + \|\mathcal L_f u\|_{C^\eta} \le C\|u\|_{C^{2+\eta}}. 
\]
\end{lemma}

\begin{proof}
Throughout, we will use the fact that for $\eta\in (0,1)$, the spaces $C^\eta(\mathcal O)$ and $~C^{2+\eta}(\mathcal O)$ equal the H\"older-Zygmund spaces $\mathcal C^\eta(\mathcal O)$ and $\mathcal C^{2+\eta}(\mathcal O)$, with equivalent norms (see~e.g.~\cite{T78} for definitions). For these, we have the classical multiplication inequalities
\[ 
    \|u v\|_{\mathcal C^\alpha} \lesssim \|u\|_{\mathcal C^\alpha} \|v\|_{\mathcal C^\alpha},\qquad u,v\in \mathcal C^\alpha(\mathcal O),\qquad \alpha\ge 0.  
\]
The second inequality in the statement of the lemma then follows from the estimate
\begin{align*}
    \|\mathcal L_f u \|_{\mathcal C^{\eta}}&= \| f\Delta u + \nabla f \cdot \nabla u\|_{\mathcal C^\eta} \\ 
    & \lesssim \|f\|_{\mathcal C^\eta} \|\Delta u\|_{\mathcal C^{2+\eta}} + \|f\|_{\mathcal C^{1+\eta}} \| \nabla u \|_{\mathcal C^{1+\eta}} 
    \lesssim \|f\|_{\mathcal C^{1+\eta}} \|u\|_{\mathcal C^{2+\eta}}.
\end{align*}
For the first inequality, %let us denote the Neumann boundary operator by 
%\[\mathcal B:u\mapsto Bu= \nabla u \cdot \nu |_{\partial \mathcal O}.\]
we use the fact that the Laplace operator establishes (jointly with the trace operator $\Tr[\cdot]$) a topological isomorphism 
\[ 
    u\mapsto (\Delta u,\Tr[u]), \qquad 
    (\Delta, \Tr) : \mathcal C^{\alpha+2}_N(\mathcal O) \to \mathcal C^\alpha(\mathcal O)\times \mathcal C^{\alpha+2}(\partial \mathcal O),
\]
for any $\alpha\ge0$; see, for instance, Theorem 4.3.4 in \cite{T78} or also display (5.6) in \cite{nickl2020convergence}. Moreover, we note that by the chain rule, for any $f\in\mathcal F$, since $f\ge \fmin $, it holds that 
\[ 
    \|f^{-1}\|_{C^1}\lesssim \|f\|_{C^1},
\]
where the multiplicative constant only depends on $\fmin $. Using this, we obtain
\begin{equation*}
    \begin{split}
        \|u\|_{\mathcal C^{2+\eta}} &\lesssim \|\Delta u\|_{\mathcal C^\eta } +\|u\|_{\mathcal C^\eta}\\
        & = \|f^{-1} (\mathcal L_fu -\nabla f\cdot\nabla u) \|_{\mathcal C^\eta } +\|u\|_{C^\eta}\\
        & \lesssim \|f^{-1}\|_{C^1} \big(\|\mathcal L_fu\|_{\mathcal C^\eta} +\|f\|_{\mathcal C^{1+\eta}}\|u\|_{\mathcal C^{1+\eta}} \big)+\|u\|_{C^\eta}\\
        &\lesssim  \|f\|_{C^1} \|\mathcal L_fu\|_{\mathcal C^\eta} + \|f\|_{C^1} \|f\|_{\mathcal C^{1+\eta}}\|u\|_{\mathcal C^{\eta}}^{1/2}\|u\|_{\mathcal C^{2+\eta}}^{1/2} +\|u\|_{\mathcal C^\eta}\\
        &\le \|f\|_{C^1} \|\mathcal L_fu\|_{\mathcal C^\eta} + \|f\|_{C^1} \|f\|_{\mathcal C^{1+\eta}}\Big(\frac{1}{2\eps}\|u\|_{\mathcal C^{\eta}}+ \frac{\eps}{2} \|u\|_{\mathcal C^{2+\eta}}\Big) +\|u\|_{\mathcal C^\eta},
    \end{split}
\end{equation*}
for any $\eps>0$, where we used in the last step that $ab\le a^2/(2\eps) + \eps b^2/2$ for any $a,b\in\R$. Thus, choosing $\eps>0$ large enough (only depending on $\fmin , \|f\|_{\mathcal C^{1+\eta}}$ and $\mathcal O$) and subsequently subtracting the term in the right hand side containing $\|u\|_{\mathcal C^{2+\eta}}$, we have proved as desired that 
$$  \|u\|_{\mathcal C^{2+\eta}}\lesssim \|\mathcal L_{f}u\|_{\mathcal C^{\eta}} +\|u\|_{\mathcal C^{\eta}}.$$
\end{proof}

Next, we state and prove the following lemma which entails that the constants in Weyl's asymptotics for the eigenvalues of $\mathcal L_f$ can be controlled by terms that only depends on upper and lower bounds for the conductivity function $f$.

\begin{lemma}\label{lem-weyl}
There exists a constant $C>0$ only depending $\fmin ,\|f \|_{L^\infty}$ and $\mathcal O$ such that the eigenvalues $(\lambda_{f,j},\ j\ge 0)$ of the operator $-\mathcal L_f$ with domain $\mathcal D(\mathcal L_f)=H^2_N(\mathcal O)$, ordered increasingly, satisfy
\[ 
    C^{-1}j^{2/d} \le \lambda_{f,j} \le  Cj^{2/d}. 
\] 
\end{lemma}

\begin{proof}
Let the action of the quadratic form associated to $-\mathcal L_f$ on functions $u\in H^1(\mathcal O)$ be denoted by
\[ 
    Q_f(u) = -\langle u, \mathcal L_f u \rangle = \int_{\mathcal O} f(x)|\nabla u (x)|^2 dx. 
\]
For any finite-dimensional subspace $L\subseteq \mathcal D(\mathcal L_f)$, define 
\[ 
    \lambda_{f,j}(L) := \sup_{u\in L : \|u\|_{L^2}\le 1} Q_f(u). 
\]
We use the standard `minimax' characterisation of the eigenvalues via minimising the above quantity over all subspaces of dimension $j$:
\[ 
    \lambda_{f,j} = \inf_{L\subseteq \mathcal D(\mathcal L_f): \dim(L)=j} \lambda_{f,j}(L). 
\]
Now let us denote the eigenvalues of the Neumann-Laplace operator $-\Delta = \mathcal L_1$ by $\lambda_j=\lambda_{1,j}$, with associated quadratic form $Q_\Delta(u)= \int_\mathcal O |u(x)|^2 dx$. Using Weyl's law, it holds that for some constant $c\equiv c(\mathcal O)\ge 1$, we have $\lambda_j\in [c^{-1}j^{2/d},cj^{2/d} ] $. It then follows immediately form the above `variational' characterisation of the eigenvalues that there is a further constant $c'\equiv c'(\fmin , \|f\|_{L^\infty})>0$ for which 
\[ 
    \lambda_{f,j} \in [ (c')^{-1}\lambda_j,c'  \lambda_j]
    =[C^{-1}j^{2/d},Cj^{2/d} ], 
\]
with $C\equiv C(\fmin , \|f\|_{L^\infty},\Ocal)>0$, concluding the proof.
\end{proof}

We conclude this section with a lemma on the growth of Sobolev- and H\"older norms of the eigenfunctions $e_{f,j}$, which follow straightforwardly from the previous lemmas in this section.

\begin{lemma}\label{lem-eigenfunction-bounds}
    Let $d\le 3$, $R>0$ and $\eta\in (0,1)$. There exists a constant $C>0$ only depending on $f_{\min}$, $R$, $\mathcal O$ and $\eta$ such that the following holds true. 
    
    i) For all $f\in\mathcal F$ with $\|f\|_{C^1}\le R$ and all $j\ge 0$, $\|e_{f,j}\|_{H^2}\le C(1+j^{2/d})$.

    ii) For all $f\in\mathcal F$ with $\|f\|_{C^{1+\eta}}\le R$ and all $j\ge 0$, $\|e_{f,j}\|_{C^{2+\eta}}\le C(1+j^{5/d})$.
\end{lemma}

\begin{proof}
    Part i) follows from an application of Lemma \ref{norm-equiv} and \ref{lem-weyl}, which together imply
    \[ \|e_{f,j}\|_{H^2} \simeq \|e_{f,j}\|_{L^2} +\|\mathcal L_f e_{f,j}\|_{L^2} \lesssim 1+ j^{2/d}. \]
    Note that the implicit constants, by Lemmas \ref{norm-equiv} and \ref{lem-weyl}, only depend on the relevant quantities in the statement.

To prove part ii), we use Lemma \ref{Holder}, the classical interpolation inequality for H\"older spaces (see, e.g.,~Corollary 1.2.19 from \cite{lunardi}), the Sobolev embedding, Lemma \ref{lem-weyl} as well as the first part of this lemma, to estimate
\begin{equation}
\begin{split}
    \|e_{f,j}\|_{C^{2+\eta}} &\simeq \|e_{f,j}\|_{C^{\eta}} +\|\mathcal L_f e_{f,j}\|_{C^{\eta}}\\
    &\lesssim (1+ \lambda_{f,j})\| e_{f,j}\|_{C^{\eta}}\\
    &\lesssim (1+ \lambda_{f,j})\| e_{f,j}\|_{L^\infty}^{2/(2+\eta)} \| e_{f,j}\|_{C^{2+\eta}}^{\eta/(2+\eta)}\\
    &\lesssim (1+ j^{2/d})\| e_{f,j}\|_{H^2}^{2/(2+\eta)} \| e_{f,j}\|_{C^{2+\eta}}^{\eta/(2+\eta)}\\
    &\lesssim (1+ j^{2/d})^{1+2/(2+\eta)} \|e_{f,j}\|_{C^{2+\eta}}^{\eta/(2+\eta)}.
\end{split}
\end{equation}
The statement follows from dividing both sides by $\|e_{f,j}\|_{C^{2+\eta}}^{\eta/(2+\eta)}$ and noting that $\eta<1$.
\end{proof}
%
%
%

%%%%%%%%%%%%%%%%%%%%%%%%%%%%%%%%%%%%%%%%%%%%%%
\subsection{Estimates for the Neumann heat kernel}

We now turn to deriving estimates for the Neumann heat kernels $p_{t,f}$. Again, we start with $L^2$-type bounds.

\begin{lemma}\label{lem-heatkernel}
Let $p_{t,f}$ denote the heat kernel (i.e.~the transition density function) of the reflected diffusion process \eqref{Eq:SDE}. For all $f\in\Fcal$, there exists a constant $C>0$ only depending on $\fmin ,\|f \|_{L^\infty}$ and $\mathcal O$, such that
\begin{align*}
	\sup_{y\in \mathcal O}\|p_{t,f}(\cdot,y)\|_{L^2} &\le C t^{-d/4},\qquad t>0.
\end{align*}
Moreover, for $d\le 3$ we also have that for some constant $C\equiv C(\mathcal O)>0$,
\[
    \sup_{y\in \mathcal O}\| p_{t,f} (\cdot,y)\|_{H^2} \le C t^{-d/2-2}, \qquad t>0.
\]
In particular, the heat kernel $\varphi_\delta$ of the reflected standard Brownian motion on $\mathcal O$ (corresponding to the case $f=1$) satisfies 
\begin{align*}
	\sup_{y\in \mathcal O}\|\varphi_\delta(\cdot,y)\|_{L^2} &\le C \delta^{-d/4},\qquad \delta>0.
\end{align*}
\end{lemma}

\begin{proof}	
We first note that by the classical Gaussian estimates for the Neumann heat kernel (e.g.~
\cite{daners2000heat})
\[ 
    p_{t,f}(x,y)\le C t^{-d/2} \exp\Big(-\frac{\|x-y\|^2}{Ct}\Big),  
\]
for some $C$ only depending on $\fmin$ and $\|f\|_{L^\infty}$. Then,
\begin{align*}
	\| p_{t,f} (\cdot,y)\|_{L^2}^2 
    &\lesssim \int_\mathcal O t^{-d}\exp\big(-\frac{2\|x-y\|^2 }{Ct}\big)  dx \\
    &\le t^{-d/2} \int_\mathcal O t^{-d/2}\exp\big(-\frac{2\|x-y\|^2 }{Ct}\big) dx \lesssim t^{-d/2},
\end{align*} 
proving the first claim. For the estimate in $H^2(\Ocal)$, assume that $d\le 3$. Using the spectral decomposition
\[ 
    p_{t,f}(x,y)= \sum_{j=0}^\infty e^{-\lambda_{f,j} t}e_{f,j}(x)e_{f,j}(y),\qquad t>0,\qquad x,y \in \mathcal O, 
\]
Then, using the Sobolev embedding and Lemmas \ref{lem-weyl} and \ref{lem-eigenfunction-bounds}, it follows that for some $c>0$,
\begin{align*}
	\|p_{t,f} (\cdot, y)\|_{H^2} &\le \sum_{j=0}^\infty e^{-\lambda_{f,j}t} \|e_{f,j}\|_{H^2} |e_{f,j}(y)|\\
	& \lesssim \sum_{j=0}^\infty e^{-cj^{-2/d}t} \|e_{f,j}\|_{H^2}^2
    \simeq \sum_{j=0}^\infty e^{-cj^{-2/d}t} (1+j^{4/d}),
\end{align*}
which is further upper bounded by
\begin{align*}
	\int_0^\infty e^{-cj^{2/d} t}(1+j^{4/d}) dj&=\int_0^\infty e^{-cx^2} (1+ x^dt^{-d/2})^{4/d} x^{d-1}t^{-d/2} dx\\
	%&\lesssim t^{-d/2}\int_0^\infty e^{-cx^2}x^{d-1}dx +  t^{-d/2-2}\int_0^\infty e^{-cx^2}x^4x^{d-1}dx \\
	&\lesssim t^{-d/2}+t^{-d/2-2},
\end{align*}
which proves the second claim.
\end{proof}

Next, we derive the following uniform estimate.

\begin{lemma}\label{lem-heatkernel-uniform}
Let $\eta\in (0,1)$, $D>0$ and $R>0$. Then, there exists $C\equiv C(D,R,\fmin )>0$ such that for every $f\in\mathcal F$ with                
$\|f\|_{C^{1+\eta}}\le R$ and for every $y\in\mathcal O$,
\[
    \|p_{D,f} (\cdot, y)\|_{C^{2+\eta}}\le C. 
\]
\end{lemma}

\begin{proof} Let us fix $y\in \mathcal O$ and $f$ as in the statement. Recall the notation $u_f(D)= p_{D,f}(\cdot,y)$. Then, using the spectral decomposition, the preceding estimates as well as Lemmas \ref{lem-weyl} and \ref{lem-eigenfunction-bounds}, we see that for some constants $c,C>0$ (only depending on $D, R$, $\fmin$, $\eta$ and $\mathcal O$)
\begin{align*}
    \|u_{f}(D)\|_{C^{2+\eta}}&\le \sup_{y\in\mathcal O}\sum_{j = 0}^\infty e^{-D\lambda_{f,j}}\|e_{f,j}\|_{C^{2+\eta}}|e_{f,j}(y)|\\
    &\le \sum_{j = 0}^\infty e^{-D\lambda_{f,j}}\|e_{f,j}\|_{C^{2+\eta}}\|e_{f,j}\|_{L^\infty}\\
    &\lesssim \sum_{j = 0}^\infty e^{-D\lambda_{f,j}}\|e_{f,j}\|_{C^{2+\eta}}\|e_{f,j}\|_{H^2}
    \lesssim \sum_{j = 0}^\infty e^{-cDj^{2/d}}(1+j^{5/d})(1+j^{2/d}) \lesssim C.
\end{align*}
\end{proof}

Alongside the above upper bounds, we record the following lower bound for the Neumann heat kernel, holding uniformly over balls of H\"older spaces -- see Proposition 4 in \cite{nickl2024inference}. In particular, if $\|f\|_{C^\alpha}\le B$ for some even integer $\alpha>d/2-1$ and some $B>0$, then for every $t>0$
\[ 
    \inf_{x,y\in \mathcal O}p_{t,f}(x,y)\ge C,
\]
for a constant $C\equiv C(t,\mathcal O, d,\fmin , B,\alpha)>0$.

Finally, we conclude with a uniform estimate for the action of the transition operator $P_{t,f}$ over functions in $L^2(\Ocal)$.

\begin{lemma}\label{Pt-infinity-est}
Suppose that $d\le 3$. Then, for any $D_0>0$ and $R>0$ there exist constants $C,C'>0$ only depending on $D_0,R,\fmin,$ and $\mathcal O$ such that for all $t\ge D_0$ and all $f\in\mathcal F$ with $\|f\|_{C^1}\le R$, $P_t:L^2(\Ocal) \to L^\infty(\Ocal)$ is a bounded linear operator satisfying
\begin{equation}\label{l2h-2}
     \|P_{t,f}\varphi \|_{L^\infty}\le C \|P_{t,f}\varphi \|_{H^2} \le C'\| \varphi\|_{(H^2_N)^*}, \qquad \forall\varphi \in L^2(\mathcal O).
\end{equation}
\end{lemma}

\begin{proof}
The first inequality is simply an application of the Sobolev embedding. To prove the second, we first note that
$P_{t,f}:L^2(\mathcal O) \to \mathcal D(\mathcal L_f)=H^2_N(\mathcal O)$ is linear and bounded, with operator norm only depending on $D_0,R,\fmin,$ and $ \mathcal O$. This is seen by applying the previous lemmas in this section, and by using Parseval's identity (twice), to the effect that
\begin{align*}
\|P_{t,f}\varphi\|_{H^2}^2&\simeq \|P_{t,f}\varphi\|_{L^2}^2 +   \|\mathcal L_fP_{t,f} \varphi\|_{L^2}^2\\
&\le \|\varphi\|_{L^2}^2 + \sum_{j= 0}^\infty |\langle e_{f,j}, \varphi\rangle_{L^2}|^2 \lambda_{f,j}^2e^{-2t\lambda_{f,j}}\lesssim \|\varphi\|_{L^2}^2.
\end{align*}
Now, applying this twice with $t/2$ (which is still lower bounded by $D_0/2$) instead of $t$, and using the self-adjointness of $P_{f,t/2}$, we obtain that
\begin{align*}
    \|P_{t,f}\varphi \|_{H^2} 
    &\le \|P_{f,t/2}\varphi \|_{L^2}\\
    &= \sup_{\phi \in L^2(\Ocal):\|\phi\|_{L^2}\le 1} \Bigg| \int_\Ocal \phi(x) P_{f,t/2}\varphi(x)dx  \Bigg|\\
    &= \sup_{\phi \in L^2(\Ocal):\|\phi\|_{L^2}\le 1} \Bigg| 
    \int_\Ocal \varphi(x) P_{f,t/2}\phi(x) dx  \Bigg| \\
   & \lesssim \sup_{\phi \in H^2_N(\Ocal) : \|\phi\|_{H^2}\le 1} \Bigg| \int_\Ocal \phi(x) \varphi (x)dx \Bigg| 
    = \| \varphi \|_{(H^2_N)^*}.
\end{align*}
\end{proof}

%
%
%
%
%

%%%%%%%%%%%%%%%%%%%%%%%%%%%%%%%%%%%%%%%%%%%%%%%
\section{Further numerical results}
\label{App:MoreNum}

In this appendix we provide additional simulation studies and expand the investigation of the numerical results presented in Sections \ref{sec:Num}.

%
%
%

%%%%%%%%%%%%%%%%%%%%%%%%%%%%%%%%%%%%%%%%%%%%%%%
\subsection{Initialisation of the algorithms}
\label{SubApp:Initial}

The choice of the starting points for iterative schemes such as the ones considered in the present article, cf.~\eqref{Eq:AccProb}, \eqref{Eq:ULA} and \eqref{Eq:GradDesc}, is known to be a delicate issue that, in high-dimensional and non-convex settings, may have a profound impact on the overall recovery performance; see e.g.~\cite{bandeira2023free}, where a discussion and further references can be found. To investigate the influence of the initialisation on the employed pCN algorithm, ULA and gradient descent method, we performed a set of numerical experiments based on the same synthetic data set $X^{(n)}$, with $n=50000$, and the same Gaussian prior $\Pi(\cdot)$ used in Section \ref{sec:Num}, each consisting in a run of the three schemes with a differently specified starting point $\vartheta_0$. In particular, alongside the cold start $\vartheta_0=0$ under which the results in Sections \ref{sec:Num} were obtained, we considered the `warm start' $\vartheta_0=\theta_0$, where $\theta_0\in\R^{K+1}$ is the vector of Fourier coefficients of the true (reparametrised) conductivity function $F_0$, the `random start' $\vartheta_0\sim \Pi(\cdot)$, and the `challenging start' $\vartheta_0=-\theta_0$. The results are summarised in the comparative Tables \ref{Tab:pCNInit} - \ref{Tab:GradDescInit}, with each row corresponding to a different starting point.

%

%%%%%%%%%%%%%%%%%%%%%%%%%%%%%%%%%%%%%%%%%%%%%%%
\subsubsection{Initialisation of sampling-based methods}

As revealed by Tables \ref{Tab:pCNInit} and \ref{Tab:ULAInit} respectively, the performances of the pCN algorithm and the ULA were  only slightly impacted in our numerical experiments by the choice of the initialiser, in terms of both the magnitude of the estimation error associated to the resulting posterior mean estimates, and the number of iterations required by the generated chains to move from the starting points to the regions containing higher posterior probability. Indeed, the same burnin times were deemed adequate across the runs of each method (except for the ones with warm start, for which no burnin was necessary).

\begin{table}
\caption{Performance of the pCN algorithm under different starting points}
\label{Tab:pCNInit}
\centering
\renewcommand{\arraystretch}{1.2}
\begin{tabular}{ c|c|c|c|c|c|c} 
 pCN  & $\ell_n(f_{\vartheta_0})$ & n.~iterations & burnin & stepsize & acceptance ratio & $\|F_0 - \bar F_n\|_2$  \\
  \hline
  $\vartheta_0=0$ & 4254.7336 & 25000 & 2500 & .0001 & 29.53\% & .2097\\
 \hline
  $\vartheta_0=\theta_0$ & 7707.9226 & 25000 & 0 & .0001 & 29.49\% & .1812\\
 \hline
 $\vartheta_0\sim\Pi(\cdot)$ & 809.0839 & 25000 & 2500 & .0001 & 30.76\% & .1995 \\
  \hline
 $\vartheta_0=-\theta_0$ & -11193.9608 & 25000 & 2500 & .0001 & 32.01\% & .2094 \\
\end{tabular}
\end{table}

\begin{table}
\caption{Performance of the ULA under different starting points}
\label{Tab:ULAInit}
\centering
\renewcommand{\arraystretch}{1.2}
\begin{tabular}{ c|c|c|c|c|c} 
 ULA & $\log\pi(\vartheta_0|X^{(n)})$ & n.~iterations & burnin & stepsize & $\|F_0 - \bar F_n\|_2$   \\
 \hline
    $\vartheta_0 = 0$ & 425.4734 & 10000 & 250 & .000025 & .2033 \\
    \hline
    $\vartheta_0=\theta_0$ & 7702.2889 & 10000 & 0 & .000025 & .2024  \\
 \hline
 $\vartheta_0\sim\Pi(\cdot)$ & 806.9926 & 10000 & 250 & .000025 & .2035  \\
  \hline
 $\vartheta_0=-\theta_0$ & -11199.5945 & 10000 & 250 & .000025 & .2056  \\
\end{tabular}
\end{table}

\begin{table}
\caption{Performance of the gradient descent method under different starting points}
\label{Tab:GradDescInit}
\centering
\renewcommand{\arraystretch}{1.2}
\begin{tabular}{ c|c|c|c|c} 
MAP & $\log\pi(\vartheta_0|X^{(n)})$ & n.~iterations & stepsize & $\|F_0 - \hat F_n\|_2$ \\
 \hline
 $ \vartheta_0=0$ & 425.4734 & 135 & .00001 & .2507\\
  \hline
  $\vartheta_0=\theta_0$ & 7702.2889 & 35 & .00001 & .0912 \\
 \hline
 $\vartheta_0\sim\Pi(\cdot)$ & 2856.6789 & 140 & .00001 & .2852 \\
  \hline
 $\vartheta_0=-\theta_0$ & -11199.5945 &  372 & .00001 & .5928 \\
\end{tabular}
\end{table}

%%%%%%%%%%%%%%%%%%%%%%%%%%%%%%%%%%%%%%%%%%%%%%%
\subsubsection{Initialisation of gradient descent and the implications of multimodality} 

A significantly stronger dependence on the initialisation step was instead observed for the computation of the MAP estimator via the gradient descent method. As reported in Table \ref{Tab:GradDescInit}, the runs with cold and random starts (first and third row) required a similar number of iterations to converge (according to the criterion established in Section \ref{Sec:GradResults}), and resulted in comparable estimation errors, which are both larger than the ones attained by the posterior mean estimates corresponding to the same initialisation. On the other hand, for the warm start (second row), a greatly reduced number of iterations were sufficient, and the obtained estimation error is by a wide margin the lowest across all the methods and all the experiments (and in fact it is close to the lower bound given by the approximation error from projecting the ground truth on the linear space spanned by the employed basis functions, which is equal to .08477). Finally, the challenging start $\vartheta_0 = -\theta_0$ (fourth row) was observed to produce the opposite effect, heavily slowing down the convergence of the scheme and yielding a sharp deterioration in the reconstruction quality.

The observed influence of the initialisation in this context furnishes a strong indication that the posterior distribution is multimodal, whence different starting points imply the convergence of the gradient descent method to a potentially different local optimum. The numerical results for the (computationally more intensive) pCN algorithm and the ULA are then aligned with the heuristics that sampling-based methods may generally be more robust in non-convex settings. These observations will be further corroborated by an examination of the one-dimensional marginal posterior distributions in Appendix \ref{Sec:MarginDistr} below.

%

%%%%%%%%%%%%%%%%%%%%%%%%%%%%%%%%%%%%%%%%%%%%%%%
\subsection{One-dimensional marginal posterior distributions}
\label{Sec:MarginDistr}

\begin{figure}
\hfill
\includegraphics[width=7cm,height=.75cm]
{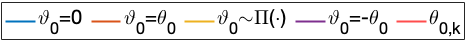}
\hfill
\includegraphics[width=6cm,height=.75cm]{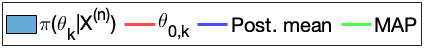}
\hfill
\includegraphics[width=8cm,height=3.55cm]{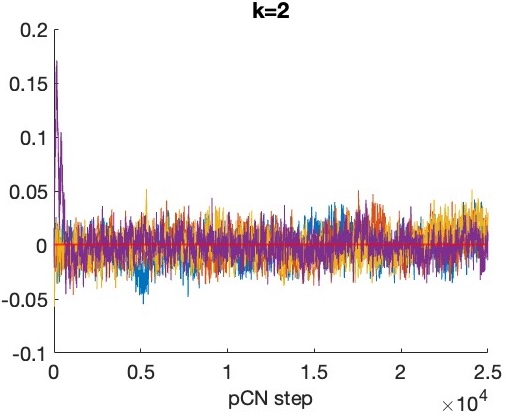}
\includegraphics[width=6cm,height=3.55cm]{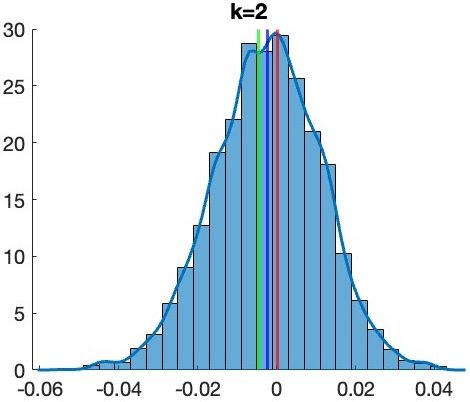}
\includegraphics[width=8cm,height=3.55cm]{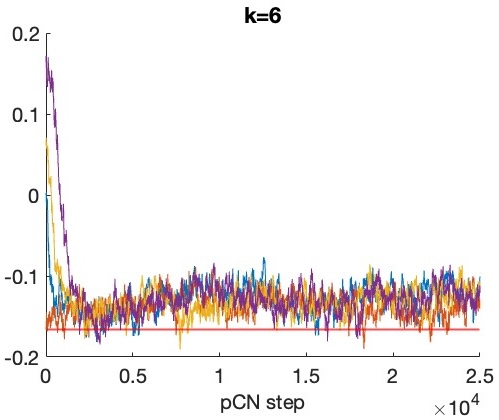}
\includegraphics[width=6cm,height=3.55cm]{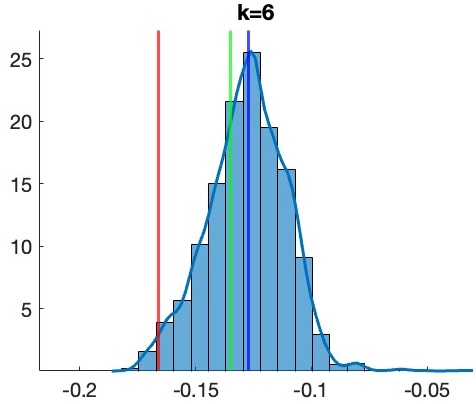}
\includegraphics[width=8cm,height=3.55cm]{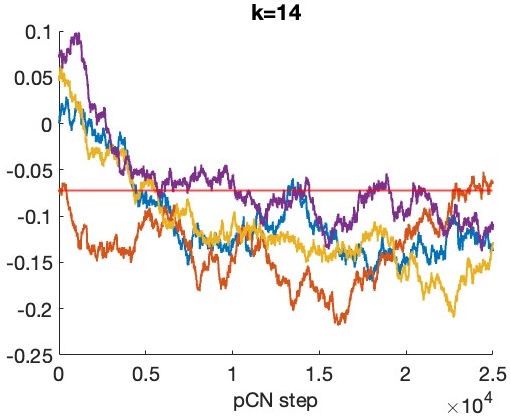}
\includegraphics[width=6cm,height=3.55cm]{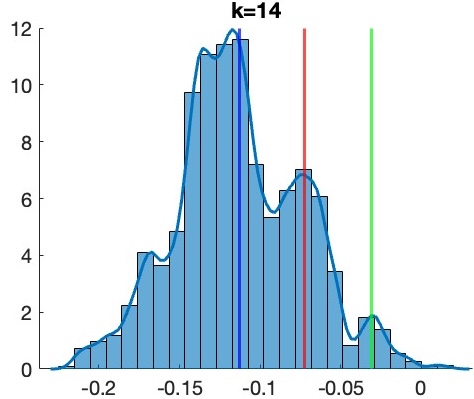}
\includegraphics[width=8cm,height=3.55cm]{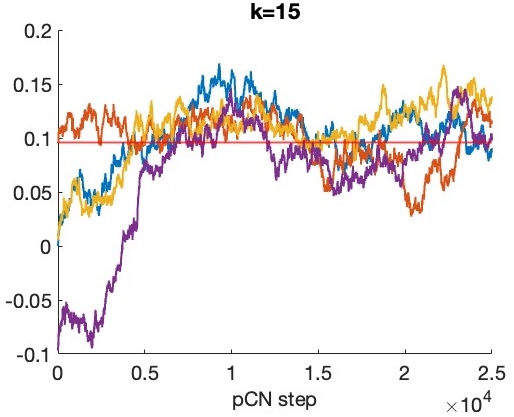}
\includegraphics[width=6cm,height=3.55cm]{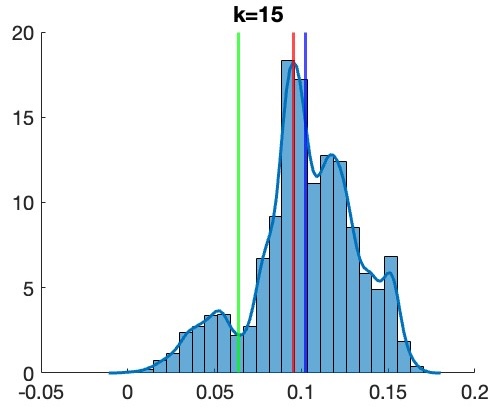}
\includegraphics[width=8cm,height=3.55cm]{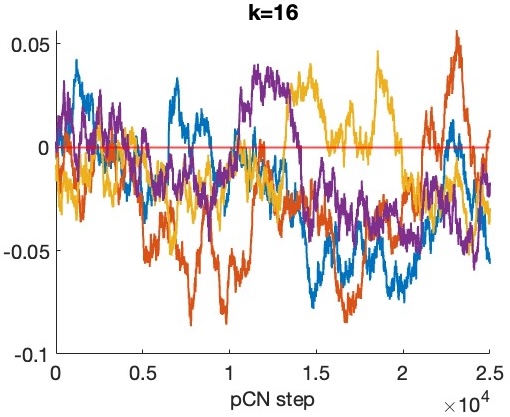}
\includegraphics[width=6cm,height=3.55cm]{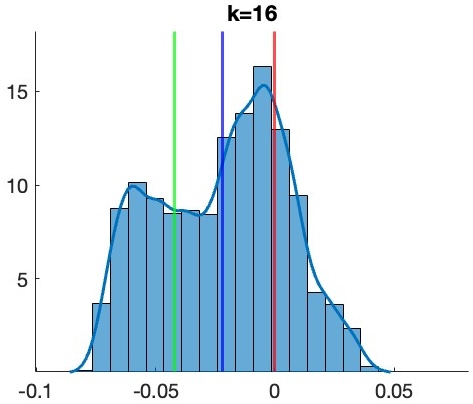}
\caption{Left column: trace-plots for some individual components of 25000 approximate samples from the posterior distribution of $\theta|X^{(n)}$ obtained via the pCN algorithm with different starting points. Right column: resulting approximations to the marginal posterior probability density functions.}
\label{Fig:MarginPost}
\end{figure}

Figure \ref{Fig:MarginPost} (left column) shows the trace-plots $(\vartheta_{m,k}, \ m=0,1,\dots,M)$, with $M=25000$, for some representative individual components (specifically, for $k=2,6,14,15,16$) of the approximate posterior samples obtained via the pCN algorithm in the context of the simulation study discussed in Appendix \ref{SubApp:Initial} above, based on low-frequency observations $X^{(n)}$ with $n=50000$. Trace-plots of different colours corresponds to the four different considered starting points. For the cold start $\vartheta_0=0$ (marked in blue), under which the results in Section \ref{Sec:pCNResults} were obtained, the associated Monte Carlo approximations (after the burnin) to the marginal posterior probability density functions of $\theta_k|X^{(n)}$, where $\theta_k=\langle F,e_k\rangle_{L^2}$, are shown in the right column, alongside the value of the corresponding true Fourier coefficients $\theta_{0,k}$ (vertical red lines) and the approximate posterior mean estimates $\bar\vartheta_{M,k}$ (vertical blue lines). The corresponding components of the MAP estimates calculated in Section \ref{Sec:GradResults} via the gradient descent algorithm  are also reported (vertical green lines). 

%

%%%%%%%%%%%%%%%%%%%%%%%%%%%%%%%%%%%%%%%%%%%%%%%
\subsubsection{Unimodality and Gaussian approximations for the lower frequencies}

An inspection of the plots showcases some interesting features for the statistical task at hand: in particular, for the lower frequencies, including for the first two displayed cases $k=2$ and $k=6$, the obtained marginal posterior distributions of the Fourier coefficients appear to be uni-modal and approximately of Gaussian shape. This lead to a substantial agreement between the posterior mean and the MAP estimates.

We also note that the marginal posterior distribution of $\theta_6|X^{(n)}$ displays a visible shrinkage towards the origin, which arises from the regularisation induced by the employed Gaussian prior. Such remaining finite sample effects are a possible indication of the potential severely ill-posedness of the problem and the resulting logarithmic rate of convergence, cf.~\cite{nickl2024inference}.

%%%%%%%%%%%%%%%%%%%%%%%%%%%%%%%%%%%%%%%%%%%%%%%
\subsubsection{Multimodality for the higher frequencies}

For the higher frequencies (corresponding to the cases $k=13,14,15$ in Figure \ref{Fig:MarginPost}), the marginal posterior distributions instead exhibit multiple local modes, as may be expected from the nonlinearity of the likelihood. In turn, this can be seen to impact the performance of the MAP estimators for which, through the gradient descent algorithm, we can generally only compute local optimisers. In line with the results in Appendix \ref{SubApp:Initial}, the sampling-based posterior mean estimators appear to be overall more robust to such non-convex settings.

A comparison of the shapes of the obtained marginal posterior distributions for the lower and higher frequencies further raises the question of the validity, in the considered statistical model, of the `Bernstein-von Mises' phenomenon and of a global Gaussian limit for the posterior distribution $\Pi(\cdot|X^{(n)})$ as $n\to\infty$.  In related parameter identification problems for the conductivity function in steady-state elliptic PDEs, an impossibility result was established in \cite{nickl2022some}. For the problem at hand, the results for the higher frequencies seem to provide some negative evidence.

%
%
%

%%%%%%%%%%%%%%%%%%%%%%%%%%%%%%%%%%%%%%%%%%%%%%%
\subsection{Stationary Gaussian process priors}
Alongside the random series expansions on orthonormal bases considered in the previous sections, a popular alternative approach to construct Gaussian priors on function spaces defined on $d$-dimensional domains, $d\in\N$, is through the specification of a stationary covariance kernel, namely a symmetric and positive semidefinite function $C:\R^d\times\R^d\to\R$ such that $C(x,y)=C(x+z,y+z)$ for all $x,y,z\in\R^d$; see Chapter 4 in \cite{Rasmussen2006gaussian}. A widely used choice is the Matérn kernel,
\begin{equation}
\label{Eq:MatKer}
    C_{\textnormal{MAT}}(x,y):=
    \frac{2^{1-\alpha}}{\Gamma(\alpha)}
    \left(\frac{|x-y|\sqrt{2\alpha}}{\ell} \right)^\alpha
    B_\alpha\left(\frac{|x-y|\sqrt{2\alpha}}{\ell} \right),
\end{equation}
where $\Gamma$ is the Gamma function, $B_\alpha$ denotes a modified Bessel function of the second kind, and $\alpha,\ell>0$ are hyper-parameters governing the regularity and the length-scale, respectively (cf.~Figure \ref{Fig:MatPrior}). A second example of interest is the squared exponential kernel,
\begin{equation}
\label{Eq:SEKer}
    C_{\textnormal{SE}}(x,y):=e^{-|x-y|^2/(2\ell^2)},
\end{equation}
corresponding to the limit of the Matérn one when  $\alpha\to\infty$.

%

%%%%%%%%%%%%%%%%%%%%%%%%%%%%%%%%%%%%%%%%%%%%%%%
\subsubsection{Posterior inference with stationary Gaussian process priors}

For the considered diffusion domain $\Ocal\subset\R^d$  and given any stationary covariance kernel $C$, let $G=(G(x), \ x\in\Ocal)$ be the associated centred and stationary Gaussian process indexed by $\Ocal$, identified by the relation
\begin{equation}
\label{Eq:FDDs}
    E[G(x)G(y)]=C(x,y), \qquad x,y\in\Ocal.
\end{equation}
In particular, if $C$ is either taken to be equal to the Matérn kernel \eqref{Eq:MatKer} with $\alpha>2+d/2$ (and any $\ell>0$) or the squared exponential kernel \eqref{Eq:SEKer} then, arguing as in p.~330f in \cite{GV17}, the (cylindrically defined) law $\Pi(\cdot)$ of the random function $G$ can be shown to be supported on $C^2(\Ocal)$ (in fact, on $C^\infty(\Ocal)$ in the latter case), and thus may serve as an appropriate prior model for the (reparametrised) conductivity function $F$ in \eqref{Eq:Reparam}. Prior distributions of this kind can be implemented in practice by discretising the parameter space according to the expansion
\begin{equation}
\label{Eq:DiscretisationScheme2}
    F(x)\equiv F_\theta(x):= \sum_{k=0}^K \theta_k \eta_k(x), \qquad
    K\in\N,\qquad \theta_0,\dots,\theta_K\in\R,
    \qquad x\in\Ocal,
\end{equation}
where now $\eta_0,\dots,\eta_K$ are piece-wise linear functions associated to a grid of points $y_0,\dots,y_K\in\Ocal$ (e.g.~the ones resulting from a triangulation of the domain), completely determined by the identity $\eta_k(y_{k'})=1_{\{k=k'\}}$. As a consequence, the function $F_\theta$ in \eqref{Eq:DiscretisationScheme2} satisfies $F_\theta(y_k)=\theta_k$ for all $k=0,\dots,K$, and for any $x\in\Ocal$ the value $F_\theta(x)$ is found by linearly interpolating the pairs $\{(y_k,\theta_k), \ k=0,\dots,K\}$. Accordingly, modelling $F$ via the stationary Gaussian process prior $\Pi(\cdot)$ with covariance structure \eqref{Eq:FDDs} corresponds to assigning to the vector of function evaluations $\theta:=(\theta_0,\dots,\theta_K)\in\R^{K+1}$ the multivariate Gaussian prior
\begin{equation}
\label{Eq:DiscrStatPrior}
    \theta\sim N(0,\Lambda),\qquad \Lambda:=(C(y_k,y_{k'}))_{k,k'=0}^K\in\R^{K+1,K+1}.
\end{equation}

\begin{figure}
\includegraphics[width=4.7cm]{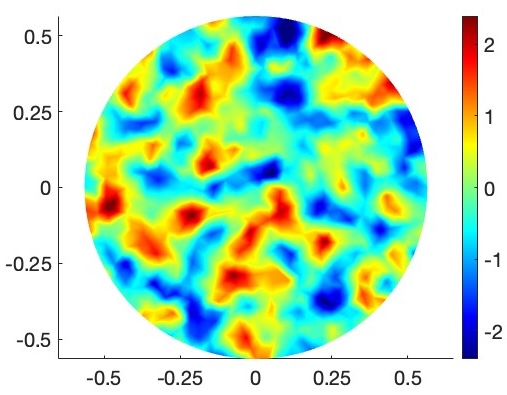}
\includegraphics[width=4.7cm]{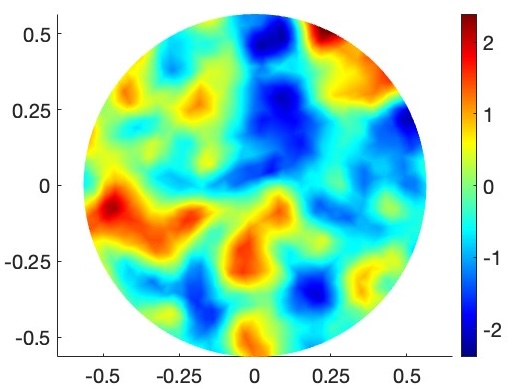}
\includegraphics[width=4.7cm]{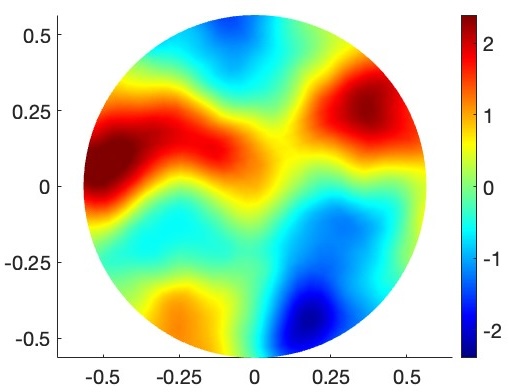}
\caption{Left to right: samples from the Matérn process priors with regularity $\alpha=2.5$ and with length scales $\ell=.05,.1,.25$, respectively.}
\label{Fig:MatPrior}
\end{figure}

\begin{figure}
\includegraphics[width=4.7cm]{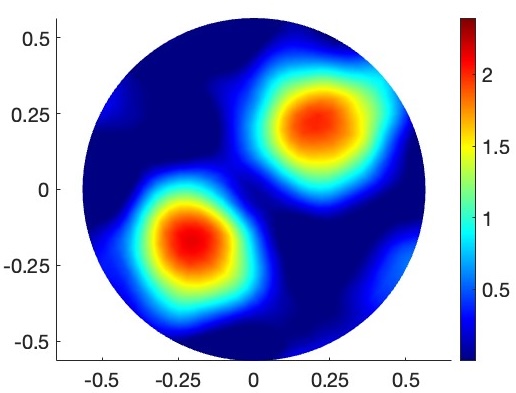}
\includegraphics[width=4.7cm,height=3.55cm]{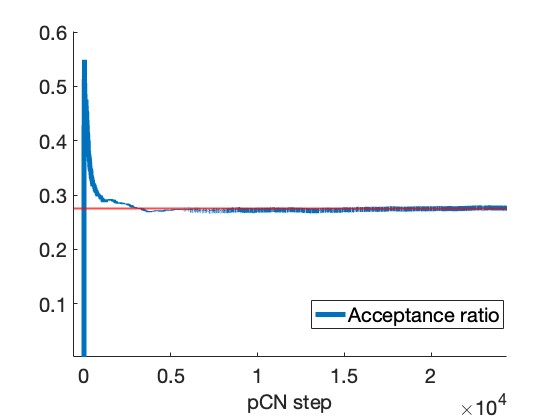}
\includegraphics[width=4.7cm,height=3.55cm]{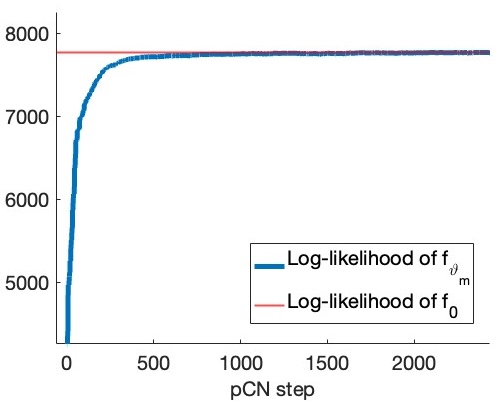}
\caption{Left: the posterior mean estimate $\bar F_n$ arising from a Matérn process prior with hyper-parameters $\alpha=2.5$ and $\ell=.25$, obtained via the pCN algorithm, to be compared to the ground truth $F_0$ shown in Figure \ref{Fig:StatProbl} (left). The overall computational time was 58 minutes. Centre: the acceptance ratio along the iterations of the pCN algorithm. Right: the log-likelihood for the first 2500 chain steps.}
\label{Fig:MatEstim}
\end{figure}

With such set-up, inference based on the posterior distribution of $\theta|X^{(n)}$ can be implemented by readily adapting the gradient-free and gradient-based methods outlined in Section \ref{Sec:Algo}, replacing the discretisation scheme \eqref{Eq:DiscretisationScheme} with \eqref{Eq:DiscretisationScheme2} and formally substituting in all the relevant equations the Neumann-Laplacian eigenfunctions $\{1,e_1,\dots,e_K\}$ used in the former with the linear interpolation functions $\{\eta_0,\dots,\eta_K\}$, and the diagonal multivariate Gaussian prior \eqref{Eq:TruncatedGP} with the discretised stationary Gaussian one defined by \eqref{Eq:DiscrStatPrior}. The numerical routines \eqref{Eq:NumLikelihood} and \eqref{Eq:NumGrad} for the evaluation of the likelihood and the gradient of the log-posterior density require no further modifications.

%

%

%%%%%%%%%%%%%%%%%%%%%%%%%%%%%%%%%%%%%%%%%%%%%%%
\subsubsection{Numerical experiments}

Based on the same data set $X^{(n)}$, with $n=50000$, underlying the simulation studies presented in Section \ref{sec:Num} and in Appendices \ref{SubApp:Initial} and \ref{Sec:MarginDistr}, we implemented posterior inference with the Matérn process prior. For brevity, let us focus on the results obtained via the pCN algorithm.

We set the hyper-parameters for the covariance kernel \eqref{Eq:MatKer}  to $\alpha=2.5$ and $\ell=.25$, and employed the discretisation scheme \eqref{Eq:DiscretisationScheme2} with $K=881$ linear interpolation functions defined over an unstructured triangular mesh covering the domain. Figure \ref{Fig:MatEstim} (left) shows the obtained posterior mean estimate of the reparametrised conductivity function, computed through the ergodic average of $M=25000$ samples from the pCN algorithm, initialised at the cold start $\vartheta_0=0$. The stepsize was tuned to $\delta=.000375$, with which a final acceptance ratio of $27.58\%$ was obtained; see Figure \ref{Fig:MatEstim} (centre). A burnin phase comprising the first 2500 iterates was identified, which we visualise in Figure \ref{Fig:MatEstim} (right) via the trace-plot of the log-likelihood. The obtained $L^2$-estimation error is equal to $.1872$, yielding a relative error of $22.42\%$. The procedure required an overall computational time of 58 minutes on a MacBook Pro with M1 processor, with an average of .14 seconds per iterate.

%
%
%

%%%%%%%%%%%%%%%%%%%%%%%%%%%%%%%%%%%%%%%%%%%%%%%
\subsection{Additional simulation studies}

\begin{figure}
\includegraphics[width=4.7cm]{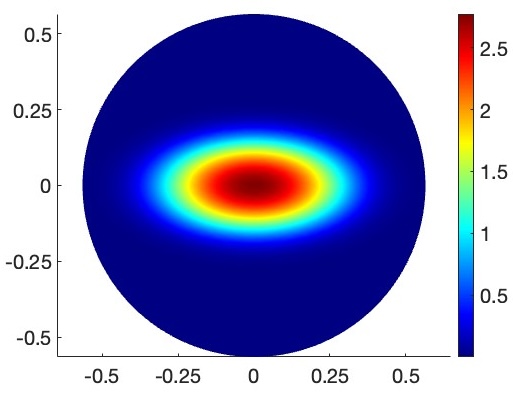}
\includegraphics[width=4.7cm]{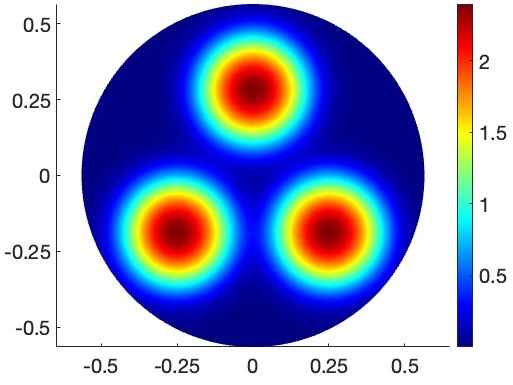}
\includegraphics[width=4.7cm]{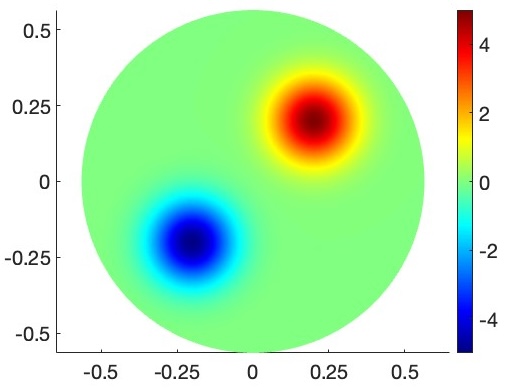}
\caption{Left to right: the (reparametrised) conductivity functions $F^{(1)}_0$, $F^{(2)}_0$ and $F^{(3)}_0$.}
\label{Fig:NewTruths}
\end{figure}

\begin{figure}
\includegraphics[width=4.7cm]{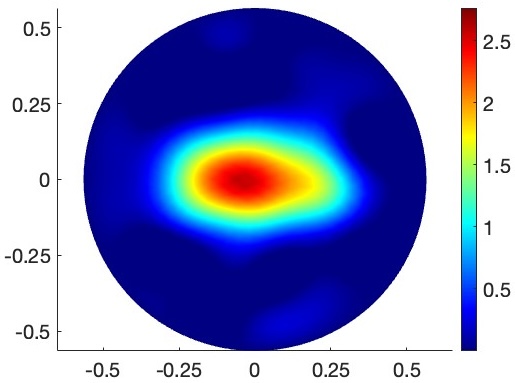}
\includegraphics[width=4.7cm]{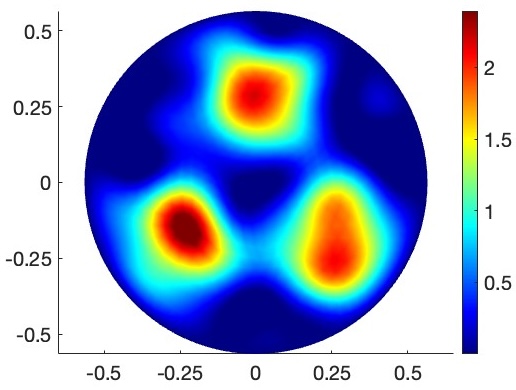}
\includegraphics[width=4.7cm]{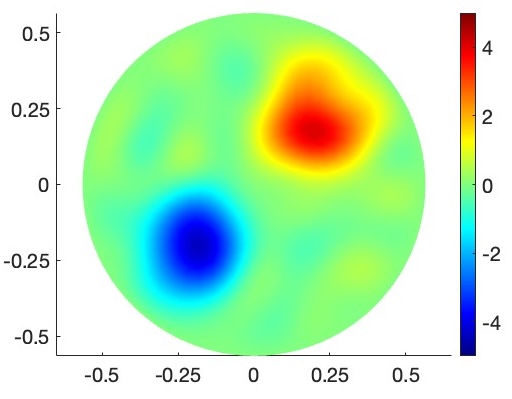}\\
\includegraphics[width=4.7cm]{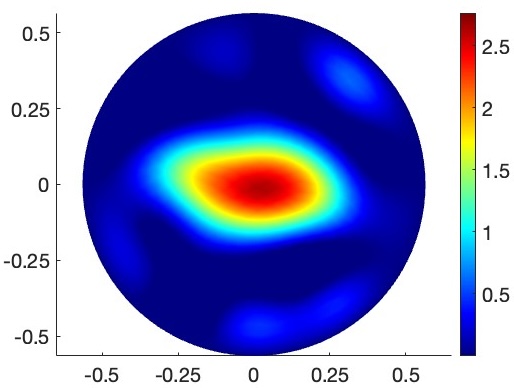}
\includegraphics[width=4.7cm]{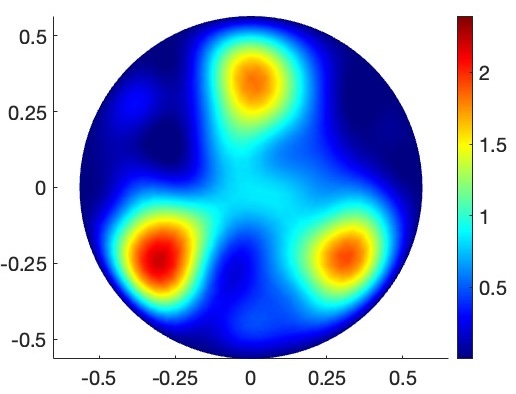}
\includegraphics[width=4.7cm]{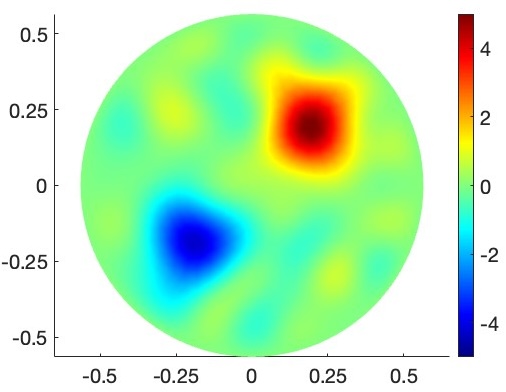}\\
\includegraphics[width=4.7cm]{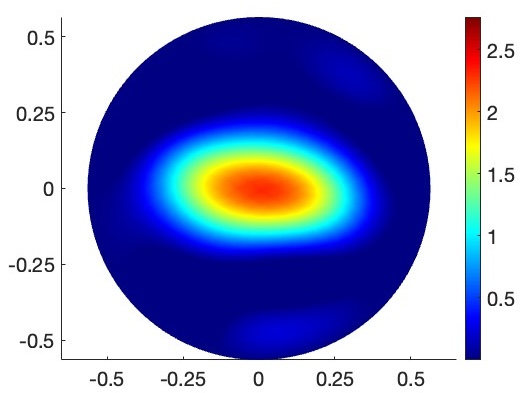}
\includegraphics[width=4.7cm]{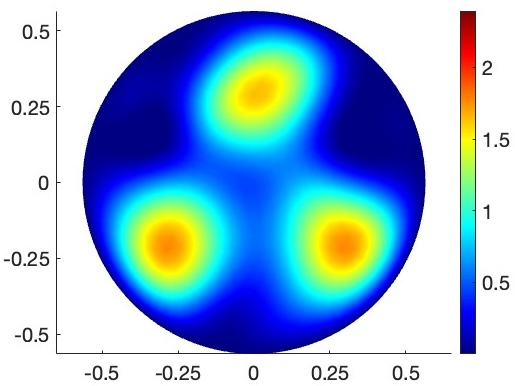}
\includegraphics[width=4.7cm]{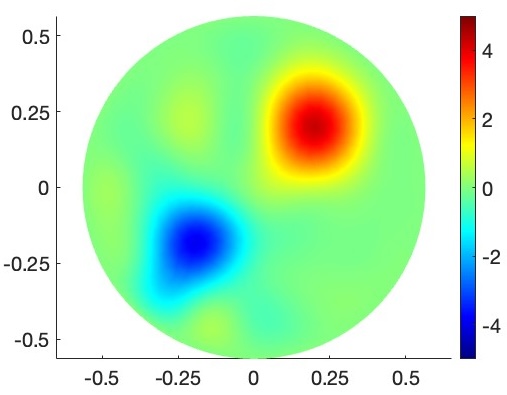}
\caption{Left column, top to bottom, respectively: the posterior means (computed via the pCN algorithm and the ULA) and the MAP estimate for the ground truth $F_0^{(1)}$ shown in Figure \ref{Fig:NewTruths} (left). Prior variability: $\sigma^2=100$. For pCN, the stepsize was set to $\delta=.00025$, the burnin to $2500$, and the acceptance ratio was $42.85\%$. For ULA, $\delta=.000025$, burnin: $250$. For gradient descent,  $\delta=.00001$, and 99 iterations were necessary for convergence. Central column: estimates for $F_0^{(2)}$, shown in Figure \ref{Fig:NewTruths} (centre). Prior variability: $\sigma^2=500$. For pCN, $\delta=.0001$, burnin: 2500, acceptance ratio: $32.19\%$. For ULA, $\delta=.000025$, burnin: 250. For gradient descent, $\delta=.00005$, number of iterations: 249. Right column: estimates for $F_0^{(3)}$, shown in Figure \ref{Fig:NewTruths} (right). Prior variability: $\sigma^2=500$. For pCN, $\delta=.0001$, burnin: 5000, acceptance ratio: $27.75\%$. For ULA, $\delta=.000025$, burnin: 250. For gradient descent, $\delta=.00005$, number of iterations: 152. }
\label{Fig:NewEstimates}
\end{figure}

\begin{table}
\caption{Recovery performances for the posterior mean $\bar F_n$ and the MAP estimate $\hat F_n$, for different ground truths}
\label{Tab:NewResults}
\centering
\renewcommand{\arraystretch}{1.75}
\begin{tabular}{ c|c|c|c|c|c} 
 & $\|F_0\|_2$ & proj. error & $\|F_0 - \bar F_n\|_2$, pCN & $\|F_0 - \bar F_n\|_2$, ULA & $\|F_0 - \hat F_n\|_2$ \\
  \hline
  $F_0^{(1)}$ & .7627 & .0759 & .1671 & .20548 & .1815 \\
 \hline
  $F_0^{(2)}$ & .9623 & .0834 & .2295 & .3461 & .3306 \\
 \hline
 $F_0^{(3)}$ & 1.2275 & .1212 & .3751  & .30848 & .2972 
\end{tabular}
\end{table}

We conclude this section presenting some further empirical investigations in which we considered the recovery of three additional true conductivity functions, respectively specified, under the reparametrisation $F=\log(f - \fmin)$, with $\fmin=.1$, by
\begin{align*}
    F^{(1)}_0(x_1,x_2)&=\log\big(1 + 15e^{-(5x_1)^2-(10x_2)^2}\big);\\
    F^{(2)}_0(x_1,x_2)&=\log\big(1+ 10e^{-(8x_1)^2-(8x_2-2.25)^2} + 10e^{-(8x_1+2)^2-(8x_2+1.5)^2}\\
    &\quad + 10e^{-(8x_1-2)^2-(8x_2+1.5)^2}\big);
\\
F^{(3)}_0(x_1,x_2)&=5e^{-(7.5x_1-1.5)^2-(7.5x_2-1.5)^2} - 5e^{-(7.5x_1+1.5)^2-(7.5x_2+1.5)^2},
\end{align*}
for $(x_1,x_2)\in\Ocal$; see Figure \ref{Fig:NewTruths}. For each of these, we generated synthetic data sets of discrete observations $X^{(n)}$, with $n=50000$, as described in Section \ref{sec:DataGen}, sampling from the Euler-Maruyama approximations of the corresponding continuous trajectories at low  `frequency' $\delta_t/D=.0001$. Next, for each set of observations, we implemented posterior inference with a truncated Gaussian series priors based on the Neumann-Laplacian eigenpairs, defined as in \eqref{Eq:TruncatedGP}, numerically computing the associated posterior mean estimates via the pCN algorithm and the ULA, and the MAP estimates through the gradient descent method. Across the three collections of experiments, the same truncation level $K=68$ and the same regularity parameter $\alpha=1$ for the prior were used. Each run of the pCN algorithm and of the ULA comprised 25000 and 10000 steps respectively, while each instance of the gradient descent method was iterated until the fulfilment of the convergence criterion laid out in Section \ref{Sec:GradResults}. All the schemes were initialised with cold starts. The obtained results are summarised in Table \ref{Tab:NewResults} and visualised in Figure \ref{Fig:NewEstimates}. The computation times were in line with those of the experiments presented in the previous sections.

\end{document}